\documentclass[11pt, a4paper]{article}

\usepackage{authblk}
\usepackage[margin=1in]{geometry}
\usepackage[dvipsnames]{xcolor}

\usepackage{amsmath}
\usepackage{amsthm}
\usepackage{amssymb}
\usepackage{amsfonts}
\usepackage{mathtools}
\usepackage{optidef}
\allowdisplaybreaks[1]

\usepackage{enumerate}
\usepackage{enumitem}
\usepackage{float}
\usepackage[title]{appendix}

\usepackage{adjustbox}
\usepackage[nodisplayskipstretch]{setspace}

\usepackage{indentfirst}
\usepackage{comment}

\usepackage{booktabs}
\usepackage{array}
\usepackage{colortbl}
\usepackage{subfigure}
\usepackage{makecell}

\usepackage[english]{babel}
\usepackage[T1]{fontenc}
\AtBeginDocument{%
  \DeclareFontShape{T1}{cmr}{m}{scit}{<->ssub*cmr/m/sc}{}%
}

\usepackage[algoruled,algosection]{algorithm2e}

\usepackage[pagebackref,colorlinks,linkcolor=NavyBlue,anchorcolor=blue,citecolor=NavyBlue]{hyperref}

\usepackage[capitalize,noabbrev]{cleveref}
\crefname{ineq}{inequality}{inequalities}
\creflabelformat{ineq}{#2{\upshape(#1)}#3}
\crefname{fact}{fact}{facts}
\crefname{equation}{equation}{equations}
\crefname{remark}{remark}{remarks}
\crefname{conjecture}{conjecture}{conjectures}
\crefname{problem}{problem}{problems}
\crefname{algorithm}{protocol}{protocols} 

\usepackage{thmtools}
\declaretheorem[style=plain,numberwithin=section]{theorem}
\declaretheorem[style=plain,numberlike=theorem]{lemma,corollary}
\declaretheorem[style=remark,numberlike=theorem]{remark}
\declaretheorem[style=plain,numberlike=theorem]{definition}
\declaretheorem[style=plain,numberwithin=theorem]{proposition}

\declaretheorem[style=plain,numberlike=theorem]{problem,conjecture}
\numberwithin{equation}{section}
\numberwithin{figure}{section}

\newcommand{\Ptime}{\textnormal{\textsf{P}}\xspace}
\newcommand{\EXP}{\textnormal{\textsf{EXP}}\xspace}

\newcommand{\BQP}{\textnormal{\textsf{BQP}}\xspace}

\newcommand{\NP}{\textnormal{\textsf{NP}}\xspace}

\newcommand{\NEXP}{\textnormal{\textsf{NEXP}}\xspace}

\newcommand{\QMA}{\textnormal{\textsf{QMA}}\xspace}

\newcommand{\AM}{\textnormal{\textsf{AM}}\xspace}

\newcommand{\IP}{\textnormal{\textsf{IP}}\xspace}

\newcommand{\QIP}{\textnormal{\textsf{QIP}}\xspace}
\newcommand{\QMAM}{\textnormal{\textsf{QMAM}}\xspace}
\newcommand{\QIPL}{\textnormal{\textsf{QIPL}}\xspace}
\newcommand{\QMAL}{\textnormal{\textsf{QMAL}}\xspace}
\newcommand{\QMAML}{\textnormal{\textsf{QMAML}}\xspace}
\newcommand{\QIPLconst}{\texorpdfstring{\textnormal{\textsf{QIPL}\textsubscript{O(1)}}}\xspace}
\newcommand{\QIPUL}{\texorpdfstring{\textnormal{\textsf{QIP}\textsubscript{\textnormal{U}}\textsf{L}}}\xspace}

\newcommand{\QIPLHC}{\texorpdfstring{\textnormal{\textsf{QIPL}\textsuperscript{\textnormal{HC}}}}\xspace}
\newcommand{\HC}{\mathrm{HC}}

\newcommand{\QSZK}{\textnormal{\textsf{QSZK}}\xspace}

\newcommand{\HVQSZK}{\texorpdfstring{\textnormal{\textsf{QSZK}\textsubscript{HV}}}\xspace}

\newcommand{\QSZKL}{\textnormal{\textsf{QSZKL}}\xspace}
\newcommand{\QSZKLHV}{\texorpdfstring{\textnormal{\textsf{QSZKL}\textsubscript{\textnormal{HV}}}}\xspace}
\newcommand{\QSZKLstar}{\mathsf{QSZKL}^{\star}_{\mathrm{HV}}}
\newcommand{\QSZKUL}{\texorpdfstring{\textnormal{\textsf{QSZK}\textsubscript{\textnormal{U}}\textsf{L}}}\xspace}
\newcommand{\QSZKULHV}{\texorpdfstring{\textnormal{\textsf{QSZK}\textsubscript{\textnormal{U}}\textsf{L}}\textsubscript{\textnormal{HV}}}\xspace}
\newcommand{\QSZKULstar}{\mathsf{QSZK_{\textnormal{U}}L}^{\star}_{\mathrm{HV}}}

\newcommand{\view}[2]{\mathtt{view}_{{#1}\rightleftharpoons{#2}}}
\newcommand{\protocol}[2]{{#1}\!\rightleftharpoons\!{#2}}
\newcommand{\PSPACE}{\textnormal{\textsf{PSPACE}}\xspace}

\newcommand{\Lspace}{\textnormal{\textsf{L}}\xspace}

\newcommand{\BPL}{\textnormal{\textsf{BPL}}\xspace}
\newcommand{\BQL}{\textnormal{\textsf{BQL}}\xspace}
\newcommand{\BQUL}{\texorpdfstring{\textnormal{\textsf{BQ}\textsubscript{U}\textsf{L}}}\xspace}

\newcommand{\SAC}{\textnormal{\textsf{SAC}}\xspace}
\newcommand{\AC}{\textnormal{\textsf{AC}}\xspace}
\newcommand{\NC}{\textnormal{\textsf{NC}}\xspace}

\newcommand{\NL}{\textnormal{\textsf{NL}}\xspace}
\newcommand{\LOGCFL}{\textnormal{\textsf{LOGCFL}}\xspace}

\newcommand{\textoverline}[1]{$\overline{\mbox{#1}}$}
\newcommand{\QSD}{\textnormal{\textsc{QSD}}\xspace}
\newcommand{\coQSD}{\texorpdfstring{\textnormal{\textoverline{\textsc{QSD}}}}\xspace}
\newcommand{\GapQSD}{\textnormal{\textsc{GapQSD}}\xspace}

\newcommand{\GapQSDlog}{\texorpdfstring{\textnormal{\textsc{GapQSD}\textsubscript{log}}}\xspace}
\newcommand{\coGapQSDlog}{\texorpdfstring{\textnormal{\textoverline{\textsc{GapQSD}}\textsubscript{log}}}\xspace}
\newcommand{\IndivProdQSD}{\textnormal{\textsc{IndivProdQSD}}\xspace}

\newcommand{\coIndivProdQSD}{\texorpdfstring{\textnormal{\textoverline{\textsc{IndivProdQSD}}}}\xspace}

\newcommand{\threeSAT}{\textnormal{\textsc{3-SAT}}\xspace}

\DeclarePairedDelimiter\rbra{\lparen}{\rparen}
\DeclarePairedDelimiter\sbra{\lbrack}{\rbrack}
\DeclarePairedDelimiter\cbra{\{}{\}}
\DeclarePairedDelimiter\abs{\lvert}{\rvert}
\DeclarePairedDelimiter\norm{\lVert}{\rVert}
\DeclarePairedDelimiter\ceil{\lceil}{\rceil}
\DeclarePairedDelimiter\floor{\lfloor}{\rfloor}
\let\ket\relax\DeclarePairedDelimiter\ket{\lvert}{\rangle}
\let\bra\relax\DeclarePairedDelimiter\bra{\langle}{\rvert}

\newcommand{\ketbra}[2]{\ensuremath{\ket{#1}\!\bra{#2}}}
\newcommand{\braket}[3]{\langle #1 | #2 | #3 \rangle}

\newcommand{\Tr}{\mathrm{Tr}}

\newcommand{\td}{\mathrm{T}}

\newcommand{\F}{\mathrm{F}}

\newcommand{\yes}{{\rm yes}}
\newcommand{\no}{{\rm no}}
\newcommand{\Cnt}{\mathrm{Cnt}}
\newcommand{\clean}{\mathrm{clean}}
\newcommand{\acc}{\mathrm{acc}}
\newcommand{\pacc}{p_{\mathrm{acc}}}
\newcommand{\var}{\mathrm{var}}
\newcommand{\cl}{\mathrm{cl}}
\newcommand{\Enc}{\mathrm{Enc}}

\newcommand{\Had}{\textnormal{\textsc{H}}\xspace}
\newcommand{\CNOT}{\textnormal{\textsc{CNOT}}\xspace}
\newcommand{\M}{\textnormal{\textsc{M}}\xspace}
\newcommand{\T}{\textnormal{\textsc{T}}\xspace}
\newcommand{\SWAP}{\textnormal{\textsc{SWAP}}\xspace}

\renewcommand{\Pr}[1]{\mathrm{Pr}\left[#1\right]}

\newcommand{\binset}{\{0,1\}}

\newcommand{\Out}{\mathrm{out}}

\DeclareMathOperator\poly{poly}
\DeclareMathOperator\polylog{polylog}

\newcommand{\bbR}{\mathbb{R}}
\newcommand{\bbN}{\mathbb{N}}

\newcommand{\calA}{\mathcal{A}}
\newcommand{\calB}{\mathcal{B}}
\newcommand{\calE}{\mathcal{E}}

\newcommand{\calI}{\mathcal{I}}
\newcommand{\calJ}{\mathcal{J}}

\newcommand{\calM}{\mathcal{M}}

\newcommand{\sfA}{\mathsf{A}}
\newcommand{\sfB}{\mathsf{B}}
\newcommand{\sfE}{\mathsf{E}}
\newcommand{\sfM}{\mathsf{M}}

\newcommand{\sfQ}{\mathsf{Q}}

\newcommand{\sfS}{\mathsf{S}}
\newcommand{\sfT}{\mathsf{T}}
\newcommand{\sfV}{\mathsf{V}}
\newcommand{\sfW}{\mathsf{W}}
\newcommand{\sfZ}{\mathsf{Z}}
\newcommand{\ttM}{\mathtt{M}}
\newcommand{\ttQ}{\mathtt{Q}}
\newcommand{\ttR}{\mathtt{R}}
\newcommand{\ttW}{\mathtt{W}}
\newcommand{\ttZ}{\mathtt{Z}}
\newcommand{\hatP}{\widehat{P}}
\newcommand{\hatV}{\widehat{V}}

\begin{document}
\setlength{\abovedisplayskip}{6pt}
\setlength{\belowdisplayskip}{6pt}

\title{Space-bounded quantum interactive proof systems}

\author[1]{Fran\c{c}ois Le Gall\thanks{Email: \url{legall@math.nagoya-u.ac.jp}}}
\author[1]{Yupan Liu\thanks{Email: \url{yupan.liu.e6@math.nagoya-u.ac.jp}}}
\author[2]{Harumichi Nishimura\thanks{Email: \url{hnishimura@i.nagoya-u.ac.jp}}}
\author[3,1]{Qisheng Wang\thanks{Email: \url{QishengWang1994@gmail.com}}}

\affil[1]{Graduate School of Mathematics, Nagoya University}
\affil[2]{Graduate School of Informatics, Nagoya University}
\affil[3]{School of Informatics, University of Edinburgh}

\date{}
\maketitle
\pagenumbering{roman}
\thispagestyle{empty}

\begin{abstract}
     We introduce two models of space-bounded quantum interactive proof systems, ${\sf QIPL}$ and ${\sf QIP_{\rm U}L}$. 
     The ${\sf QIP_{\rm U}L}$ model, a space-bounded variant of quantum interactive proofs (${\sf QIP}$) introduced by \hyperlink{cite.Wat99QIP}{Watrous (CC 2003)} and \hyperlink{cite.KW00}{Kitaev and Watrous (STOC 2000)}, restricts verifier actions to unitary circuits. In contrast, ${\sf QIPL}$ allows logarithmically many pinching intermediate measurements per verifier action, making it the weakest model that encompasses the classical model of \hyperlink{cite.CL95}{Condon and Ladner (JCSS 1995)}. 

     We characterize the computational power of ${\sf QIPL}$ and ${\sf QIP_{\rm U}L}$. When the message number $m$ is polynomially bounded, ${\sf QIP_{\rm U}L} \subsetneq {\sf QIPL}$ unless ${\sf P} = {\sf NP}$: 
     \begin{itemize} 
        \item ${\sf QIPL}^{\rm HC}$, a subclass of ${\sf QIPL}$ defined by a high-concentration condition on \textit{yes} instances, exactly characterizes ${\sf NP}$.
        \item ${\sf QIP_{\rm U}L}$ is contained in ${\sf P}$ and contains ${\sf SAC}^1 \cup {\sf BQL}$, where ${\sf SAC}^1$ denotes problems solvable by classical logarithmic-depth, semi-unbounded fan-in circuits. 
    \end{itemize} 
    However, this distinction vanishes when $m$ is constant. Our results further indicate that (pinching) intermediate measurements uniquely impact space-bounded quantum interactive proofs, unlike in space-bounded quantum computation, where ${\sf BQL}={\sf BQ_{\rm U}L}$. 
    
    We also introduce space-bounded unitary quantum statistical zero-knowledge (${\sf QSZK_{\rm U}L}$), a specific form of ${\sf QIP_{\rm U}L}$ proof systems with statistical zero-knowledge against any verifier. This class is a space-bounded variant of quantum statistical zero-knowledge (${\sf QSZK}$) defined by \hyperlink{cite.Wat09QSZK}{Watrous (SICOMP 2009)}. We prove that ${\sf QSZK_{\rm U}L} = {\sf BQL}$, implying that the statistical zero-knowledge property negates the computational advantage typically gained from the interaction.
\end{abstract}

\newpage
\tableofcontents
\thispagestyle{empty}

\newpage
\pagenumbering{arabic}
\section{Introduction}

Recent advancements in quantum computation with a limited number of qubits have been achieved from both theoretical and experimental perspectives. Theoretical work began in the late 1990s, focusing on feasible models of quantum computation operating under space restrictions, where the circuit acts on $O(\log n)$ qubits and consists of $\poly(n)$ elementary gates~\cite{Wat99,Wat03}. These models, referred to as quantum logspace, were later shown during the 2010s to offer a quadratic space advantage for certain problems over the best known classical algorithms~\cite{TS13,FL18}, which saturates the classical simulation bound. In recent years, this area has gained increased attention, particularly in eliminating (pinching) intermediate measurements in these models~\cite{FR21,GRZ21}, and through further developments~\cite{GR22,Zhandry24}. Motivated by these achievements in quantum logspace, we are interested in exploring the power of the quantum interactive proof systems where the verifier is restricted to quantum logspace. 

To put it simply, in a single-prover (quantum) interactive proof system for a promise problem $(\calI_{\yes},\calI_{\no})$, a computationally weak (possibly quantum) \textit{verifier} interacts with a computationally all-powerful but untrusted \textit{prover}. In quantum scenarios, the prover and verifier may share entanglement during their interactions. Given an input $x\in \calI_\yes \cup \calI_\no$, the prover claims that $x \in \calI_{\yes}$, but the verifier does not simply accept this claim. Instead, an interactive protocol is initiated, after which the verifier either ``accepts'' or ``rejects'' the claim. The protocol has completeness parameter $c$, meaning that if $x$ is in $\calI_{\yes}$ and the prover honestly follows the protocol, the verifier accepts with probability at least $c$. The protocol has soundness parameter $s$, meaning that if $x$ is in $\calI_{\no}$ then the verifier accepts with probability at most $s$, regardless of whether the prover follows the protocol. Typically, an interactive protocol for $(\calI_{\yes},\calI_{\no})$ has completeness $c=2/3$ and soundness $s=1/3$. 

\paragraph{Interactive proof systems with time-bounded verifier.}
The exploration of classical interactive proof systems (\IP{}) was initiated in the 1980s~\cite{Babai85,GMR85}. In these proof systems, the verifier is typically bounded by polynomial time, and $\IP[m]$ represents interactive protocols involving $m$ messages during interactions. Particularly, when the verifier's messages are merely random bits, these \textit{public-coin} proof systems are known as \textit{Arthur-Merlin proof systems}~\cite{Babai85}. Shortly thereafter, it was established that any constant-message \IP{} protocol can be parallelized to a two-message public-coin protocol, captured by the class \AM{}, and thus $\IP[O(1)]$ is contained in the second level of the polynomial-time hierarchy~\cite{Babai85,GS86}. However, \IP{} protocols with a polynomial number of messages have been shown to be exceptionally powerful, as demonstrated by the seminal result $\IP = \PSPACE$~\cite{LFKN92,Shamir92}. Consequently, \IP{} protocols with a polynomial number of messages generally cannot be parallelized to a constant number of messages unless the polynomial-time hierarchy collapses.\footnote{The assumption that the polynomial-time hierarchy does not collapse generalizes the conjecture that $\Ptime \subsetneq \NP$.} 

About fifteen years after the introduction of interactive proof systems (and a model of quantum computation), the study of quantum interactive proof systems (\QIP{}) began~\cite{Wat99QIP}. Remarkably, any \QIP{} protocol with a polynomial number of messages can be parallelized to three messages~\cite{KW00}. 
A quantum Arthur-Merlin proof system was subsequently introduced in~\cite{MW05}, and any three-message \QIP{} protocol can be transformed into this form (\QMAM{}).
By the late 2000s, the computational power of \QIP{} was fully characterized: The celebrated result $\QIP{}=\PSPACE{}$~\cite{JJUW11} established that \QIP{} is not more powerful than \IP{} as long as the gap $c-s$ is at least polynomially small. 
However, when the gap $c-s$ is double-exponentially small, this variant of \QIP{} is precisely characterized by \EXP{}~\cite{IKW12}. 
In the late 2010s, another quantum counterpart of the Arthur-Merlin proof system was considered in~\cite{KLGN19}, where the verifier's message is either random bits or halves of EPR pairs, leading to a quadrichotomy theorem that classifies the corresponding \QIP{} protocols. 

\paragraph{Interactive proof systems with space-bounded verifier.}
The investigation of (classical) interactive proof systems with space-bounded verifiers started in the late 1980s~\cite{DS89,Condon91}, alongside research on time-bounded verifiers. 
Notably, by using the fingerprinting lemma~\cite{Lipton90}, Condon and Ladner \cite{CL95} showed that the class of (private-coin) classical interactive proof systems with logarithmic-space verifiers using $O(\log{n})$ random bits exactly characterizes $\NP$. 
In parallel, public-coin space-bounded classical interactive proofs were explored in the early 1990s~\cite{Fortnow89,FL93,Condon92survey}. By around 2010, it was established that such space-bounded protocols with $\poly(n)$ public coins precisely characterize $\Ptime$~\cite{GKR15}. 

Space-bounded Merlin-Arthur-type proof systems were also studied in the early 1990s. In particular, when the verifier operates in classical logspace with $O(\log n)$ random bits and has \textit{online access} to a $\poly(n)$-bit message, the proof system exactly characterizes \NP{}~\cite{Lipton90}. 
More recently, restricting the computational power of the honest prover to quantum logspace (\BQL{}) has led to a counterpart \textit{classical} proof system that exactly characterizes \BQL{}~\cite{GRZ23}.

Although research has been conducted on quantum interactive proofs where the verifier uses quantum finite automata~\cite{NY09,NY15,Yaka13}, analogous to classical work~\cite{DS89}, to our knowledge no prior work has addressed space-bounded counterparts of quantum interactive proofs that align with the circuit-based model defined in~\cite{KW00, Wat99QIP}. In the case \textit{without interaction}, space-bounded quantum Merlin-Arthur proof systems have been studied recently. When the verifier has \textit{direct access} to an $O(\log n)$-qubit message, meaning it can process the message directly in its workspace qubits, this variant (\QMAL{}) is as weak as \BQL{}~\cite{FKLMN16,FR21}. However, when the (unitary) verifier has online access to a $\poly(n)$-qubit message, where each qubit in the message state is read-once, this variant is as strong as \QMA{}~\cite{GR23online}.\footnote{An exponentially up-scaled quantum counterpart of the space-bounded Merlin-Arthur-type proof system from~\cite{Lipton90}, with \textit{classical} messages, was also considered in~\cite{GR23online}. The variant with \textit{unitary} quantum \textit{polynomial}-space verifier, implicitly allowing $\poly(n)$ random bits, precisely corresponds to \NEXP{}.\label{footref:streaming-proof}}

It is important to note that online and direct access to messages during interactions makes no difference for time-bounded interactive or Merlin-Arthur-type proof systems, whether classical or quantum. This distinction arises from the nature of space-bounded computation.

\subsection{Main results}

\paragraph{Definitions of \QIPL{} and \QIPUL{}.} We introduce \textit{space-bounded quantum interactive proof systems} and their unitary variant, denoted as \QIPL{} and \QIPUL{}, respectively. In these proof systems, the verifier $V$ operates in quantum logspace and has direct access to messages during interaction with the prover $P$. Specifically, in a $2l$-turn (message) space-bounded quantum interactive proof system for a promise problem $(\calI_{\yes},\calI_{\no})$, this proof system $\protocol{P}{V}$ consists of the prover's private register $\sfQ$, the message register $\sfM$, and the verifier's private register $\sfW$. Both $\sfM$ and $\sfW$ are of size $O(\log n)$, with $\sfM$ being accessible to both the prover and the verifier.\footnote{Our definitions of \QIPL{} and \QIPUL{} can be straightforwardly extended to the corresponding proof systems with an odd number of messages, as shown in \Cref{fig:QIPL-odd}.} 
\begin{figure}[ht!]
    \centering
    \includegraphics[width=\textwidth]{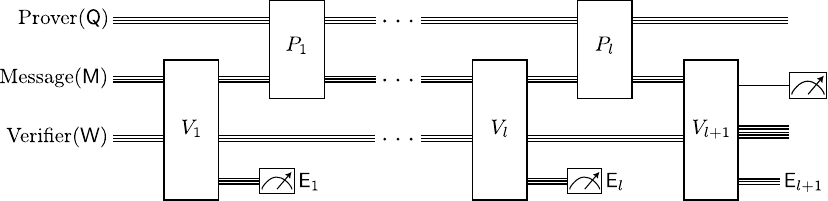}
    \caption{A $2l$-turn single-prover space-bounded quantum interactive proof system (\QIPL{}), where each environment register $\sfE_j$ is introduced by applying the principle of deferred measurements to an almost-unitary quantum circuit $\widetilde{V}_j$, resulting in the isometric quantum circuit $V_j$.}
    \label{fig:QIPL}
\end{figure}

The verifier $V$ maps an input $x\in \calI_{\yes} \cup \calI_{\no}$ to a sequence $(V_1,\cdots,V_{l+1})$, with $V_j$ for $j \in [l]$ representing the verifier's actions at the $(2j-1)$-th turn, and $V_{l+1}$ representing the verifier's action just before the final measurement. The primary difference between \QIPL{} and \QIPUL{} proof systems lies in the verifier's action $V_j$ for $j\in[l]$: 
\begin{itemize}[itemsep=0.33em, topsep=0.33em, parsep=0.33em]
    \item In \QIPL{} proof systems, each $V_j$ corresponds to an \textit{almost-unitary} quantum circuit $\widetilde{V}_j$ that includes $O(\log n)$ \textit{pinching} intermediate measurements in the computational basis.\footnote{We note that restricting the number of pinching measurements in each verifier turn $V_j$ from polynomial in $n$ to $O(\log{n})$ does not cause any loss of generality, provided that the \QIPL{} proof system has a polynomial number of turns. See \Cref{remark:restrict-measurement-times} for further details.}
    Each such measurement maps a single-qubit state $\rho$ to $\sum_{b\in\binset} \Tr(\ketbra{b}{b} \rho)\ketbra{b}{b}$.\footnote{Pinching intermediate measurements naturally arise in space-bounded quantum computation, particularly in recent developments on eliminating intermediate measurements in quantum logspace~\cite{GRZ21,GR22}. In this context, the quantum channels that capture the space-bounded computation are \textit{unital} in the case of qubits. }
    The $O(\log{n})$ bound reflects the maximum number of measurement outcomes that can be stored in logspace, aligning with the verifier's direct access to the message. 
    For convenience, we apply the principle of deferred measurements (e.g.,~\cite[Section 4.4]{NC10}), transforming the circuit $\widetilde{V}_j$ into an \textit{isometric} quantum circuit $V_j$ with a newly introduced environment register $\sfE_j$,\footnote{An $O(\log{n})$-qubit isometric quantum circuit utilizes $O(\log n)$ ancillary gates, with each ancillary gate introducing an ancillary qubit $\ket{0}$. For further details, please refer to \Cref{def:space-bounded-quantum-circuits}. } which is measured at the end of that turn, with the measurement outcome denoted by $u_j$, as illustrated in \Cref{fig:QIPL}.
    Furthermore, each environment register $\sfE_j$ remains private to the verifier and becomes inaccessible after the round that starts with the verifier's $j$-th action. 
    \item In \QIPUL{} proof systems, each $V_j$ is a unitary quantum circuit. 
\end{itemize}

The prover's actions can be similarly described by unitary quantum circuits. A proof system $\protocol{P}{V}$ is said to \textit{accept} if, after the verifier performs $V_{l+1}$ and measures the designated output qubit in the computational basis, the outcome is $1$. 
Additionally, we require \textit{a strong notion of uniformity} for the verifier's mapping: the description of the sequence $(V_1,\cdots,V_{l+1})$ must be computable by a single deterministic logspace Turing machine.\footnote{A weaker notion of uniformity only requires that the description of each $V_j$ can be individually computed by a deterministic logspace Turing machine. It is important to note that these distinctions do not arise in the time-bounded setting, as the composition of a polynomial number of deterministic polynomial-time Turing machines can be treated as a single deterministic polynomial-time Turing machine.} 
Lastly, for \QIPLHC{} proof systems, we impose an additional restriction on \textit{yes} instances: the distribution of intermediate measurement outcomes $u = (u_1,\cdots,u_l)$, conditioned on acceptance, must be \textit{highly concentrated}. More precisely, let $\omega(V)|^u$ be the contribution of $u$ to $\omega(V)$, where $\omega(V)$ is the maximum acceptance probability of $\protocol{P}{V}$. Then, there must exist a $u^*$ such that $\omega(V)|^{u^*} \geq c(n)$.

We denote $m$-turn space-bounded quantum interactive proof systems with completeness $c$ and soundness $s$ as $\QIPL_m[c,s]$, and their unitary variant as $\QIPUL_m[c,s]$. 
In particular, we adopt the following notations, which naturally extend to \QIPUL{}: 
\[ \QIPL_m \coloneqq \QIPL_m[2/3,1/3] \text{ and } \QIPL \coloneqq \cup_{1 \leq m \leq \poly(n)} \QIPL_m.\]

In \textit{constant}-turn scenarios, it is crucial to emphasize that the proof systems $\QIPL_{O(1)}[c,s]$ and $\QIPUL_{O(1)}[c,s]$ can directly simulate each other, as the environment registers $\sfE_1,\cdots,\sfE_{O(1)}$ collectively holds $O(\log{n})$ qubits.\footnote{This equivalence follows directly from the principle of deferred measurements. However, for constant-turn space-bounded quantum interactive proofs, allowing each verifier action to involve $\poly(n)$ pinching intermediate measurements might increase the proof system's power beyond the unitary case. This is because current techniques for proving results such as $\BQL = \BQUL$~\cite{FR21,GRZ21,GR22} do not directly apply in this context.\label{footref:constant-turn-QIPL}} Therefore, for simplicity, we define $\QIPL_{O(1)}[c,s]$ proof systems in which the verifier's actions are implemented by \textit{unitary} quantum circuits. 

\paragraph{Space-bounded (unitary) quantum interactive proofs.} Our first theorem serves as a quantum analog of the classical work by Condon and Ladner~\cite{CL95}:
\begin{theorem}[Informal of \Cref{thm:QIPL-eq-NP}]
    \label{thm:QIPL-eq-NP-informal}
    $\NP = \QIPLHC \subseteq \QIPL$.
\end{theorem}

Interestingly, \Cref{thm:QIPL-eq-NP-informal} suggests that the \QIPLHC{} model can be viewed as the \textit{weakest} model that encompasses space-bounded (private-coin) classical interactive proofs, as considered in~\cite{CL95}. Our definitions of \QIPL{} and its subclass \QIPLHC{} aim to introduce quantum counterparts that include these classical proof systems, ensuring that soundness against classical messages also holds for quantum messages. 
Similar soundness issues challenged multi-prover scenarios (e.g., proving $\mathsf{MIP} \subseteq \mathsf{MIP^*}$) for nearly a decade~\cite{CHTW04,IV12}, while in the single-prover settings (e.g., proving $\IP \subseteq \QIP$), it is typically resolved by measuring the prover's quantum messages and treating the outcomes as classical messages (e.g.,~\cite[Claim 1]{AN02}). 

However, space-bounded \textit{unitary} quantum interactive proofs (\QIPUL{}), which denote the most natural space-bounded counterpart to quantum interactive proofs as defined in~\cite{KW00,Wat03}, do not directly achieve the stated soundness guarantee. Hence, \QIPUL{} may be computationally weaker than \QIPL{}. 
Our second theorem characterizes the computational power of \QIPUL{}: 
\begin{theorem}[Informal of \Cref{thm:LOGCFL-in-QIPUL,thm:QIPUL-in-P}]
    \label{thm:QIPUL-bounds-informal} The following holds\emph{:}
    \[\SAC^1 \cup \BQL \subseteq \QIPUL \subseteq \cup_{c(n)-s(n) \geq 1/\poly(n)} \QIPL_{O(1)}[c,s] \subseteq \Ptime.\] 
\end{theorem}

\Cref{thm:QIPL-eq-NP-informal,thm:QIPUL-bounds-informal} suggest that \QIPUL{} is indeed \textit{weaker} than \QIPL{} unless $\Ptime = \NP$. Interestingly, this distinction from the unitary case arises even when each verifier action is slightly more powerful than a unitary quantum circuit.
It is also noteworthy that the class $\SAC^1$ is equivalent to \LOGCFL{}~\cite{Venkateswaran87}, which contains \NL{} and is contained in $\AC^1$.\footnote{For more details on the computational power of $\SAC^1$ and related complexity classes, see \Cref{subsec:classical-concepts-and-tools}. }
Our third theorem, meanwhile, focuses on space-bounded quantum interactive proof systems with a constant number of messages: 
\begin{theorem}[Informal of \Cref{thm:QIPLconst-in-NC}] 
    \label{thm:QIPLconst-informal}
    For any $c(n)-s(n) \geq \Omega(1)$, $\QIPL_{O(1)}[c,s] \subseteq \NC$. 
\end{theorem}

To compare with time-bounded classical or quantum interactive proofs, we summarize our three theorems in \Cref{table:QIPL-complexity}. Notably, our two models of space-bounded quantum interactive proofs, \QIPL{} and \QIPUL{}, demonstrate behavior that is distinct from both: 
\begin{itemize}[itemsep=0.33em, topsep=0.33em, parsep=0.33em]
    \item For (time-bounded) classical interactive proofs, all proof systems with $m \leq O(1)$ (the regime of the last row in \Cref{table:QIPL-complexity}) are contained in the second level of the polynomial-time hierarchy~\cite{Babai85,GS86}, whereas the class of proof systems with $m=\poly(n)$ (the regime of the second and third rows in \Cref{table:QIPL-complexity}) exactly characterizes \PSPACE{}~\cite{LFKN92,Shamir92}. 
    \item For (time-bounded) quantum interactive proofs, all proof systems with parameters listed in \Cref{table:QIPL-complexity} precisely capture \PSPACE{}~\cite{Wat99QIP,KW00,JJUW11}.
\end{itemize}

\begin{table}[ht!]
\centering
\adjustbox{max width=\textwidth}{
\begin{tabular}{cccc}
    \specialrule{0.1em}{2pt}{2pt}
    & Models
    & \makecell{Constant gap\\ \footnotesize{$c(n)-s(n) \geq \Omega(1)$}} 
    & \makecell{Polynomial small gap\\ \footnotesize{$c(n)-s(n) \geq 1/\poly(n)$}} \\
    \specialrule{0.1em}{2pt}{2pt}
    \makecell{\footnotesize{The number of messages:}\\$m(n) = \poly(n)$} & 
    $\QIPLHC (\subseteq \QIPL)$ &
    \makecell{ \NP{} \\ \footnotesize{\Cref{thm:QIPL-eq-NP-informal}}} &
    \makecell{ \NP{} \\ \footnotesize{\Cref{thm:QIPL-eq-NP-informal}}}\\
    \specialrule{0.05em}{2pt}{2pt}
    \makecell{\footnotesize{The number of messages:}\\$m(n) = \poly(n)$} & 
    \QIPUL{} &
    \makecell{ contains $\SAC^1 \!\cup\! \BQL$ \& in \Ptime{}\\ \footnotesize{\Cref{thm:QIPUL-bounds-informal}}} &
    \makecell{ contains $\SAC^1 \!\cup\! \BQL$ \& in \Ptime{}\\ \footnotesize{\Cref{thm:QIPUL-bounds-informal}}}\\
    \specialrule{0.05em}{2pt}{2pt}
    \makecell{\footnotesize{The number of messages:}\\$3 \leq m(n) \leq O(1)$}
    & \QIPL{} \& \QIPUL{}
    & \makecell{in \NC{}\\ \footnotesize{\Cref{thm:QIPLconst-informal}}} 
    & \makecell{ contains $\SAC^1 \!\cup\! \BQL$ \& in \Ptime{}\\ \footnotesize{\Cref{thm:QIPUL-bounds-informal}}}\\
    \specialrule{0.1em}{2pt}{2pt}
\end{tabular}
}
\caption{The computational power of $\QIPL$ and $\QIPUL$ with different parameters.}
\label{table:QIPL-complexity}
\end{table}

The central intuition underlying \Cref{table:QIPL-complexity} is that \textit{parallelization}~\cite{KW00,KKMV09}, perhaps the most striking complexity-theoretic property of \QIP{} proof systems, distinguishes \QIPUL{} from \QIPL{}. Quantum logspace operates within a polynomial-dimensional Hilbert space, remaining computationally weak even with a constant number of interactions, and is (at least) contained in \Ptime{}. In \QIPUL{}, the verifier's actions are \textit{reversible} and \textit{dimension-preserving}, allowing direct application of parallelization techniques from~\cite{KKMV09}. In contrast, \QIPL{} and its reversible generalization lack dimension preservation, requiring significantly more than $O(\log{n})$ space to parallelize the verifier's actions, which prevents parallelization. 

\paragraph{Space-bounded unitary quantum statistical zero-knowledge.} We also introduce \textit{(honest-verifier) space-bounded unitary quantum statistical zero-knowledge}, denoted as \QSZKULHV{}. This term refers to a specific form of space-bounded quantum proofs that possess statistical zero-knowledge against an honest verifier. 
Specifically, a space-bounded unitary quantum interactive proof system possesses this zero-knowledge property if there exists a quantum logspace simulator that approximates the snapshot states (``the verifier's view'') on the registers $\sfM$ and $\sfW$ after each turn of this proof system, where each state approximation must be very close (``indistinguishable'') to the corresponding snapshot state with respect to the trace distance. 

Our definition \QSZKULHV{} serves as a space-bounded variant of honest-verifier (unitary) quantum statistical zero-knowledge, denoted by \HVQSZK{}, as introduced in~\cite{Wat02QSZK}. 
Our fourth theorem establishes that the statistical zero-knowledge property completely negates the computational advantage typically gained through the interaction:
\begin{theorem}[Informal of \Cref{thm:QSZKL-eq-BQL}]
    \label{thm:QSZKL-eq-BQL-informal}
    $\QSZKUL = \QSZKULHV = \BQL$.
\end{theorem}

In addition to \QSZKULHV{}, we can define \QSZKUL{} in line with~\cite{Wat09QSZK}, particularly considering space-bounded unitary quantum statistical zero-knowledge against \textit{any verifier} (rather than an honest verifier). Following this definition, $\BQL \subseteq \QSZKUL \subseteq \QSZKULHV$. Interestingly, \Cref{thm:QSZKL-eq-BQL-informal} serves as a direct space-bounded counterpart to $\QSZK = \HVQSZK$~\cite{Wat09QSZK}.

The intuition behind \Cref{thm:QSZKL-eq-BQL-informal} is that the snapshot states after each turn capture all the essential information in the proof system, such as allowing optimal prover strategies to be ``recovered'' from these states~\cite[Section 7]{MY23}. In space-bounded scenarios, space-efficient quantum singular value transformation~\cite{LGLW23} enables fully utilizing this information.

Finally, we emphasize that our consideration of this zero-knowledge property is purely complexity-theoretic. A full comparison with other notions of (statistical) zero-knowledge is beyond this scope. For more on classical and quantum statistical zero-knowledge, see~\cite{Vad99} and~\cite[Chapter 5]{VW16}. 

\subsection{Proof techniques}

Standard techniques for quantum interactive proofs are typically developed under the restriction that the verifier is \textit{unitary}. While this restriction does not limit generality in time-bounded settings (e.g., see~\cite[Section 4.1.4]{VW16}), it presents difficulties in the context of space-bounded quantum interactive proofs, where verifiers may not be unitary. In what follows, we highlight the challenges that arise and briefly explain how we address them.  

\subsubsection{Upper bounds for \QIPLHC{} and \QIPUL{}}

\paragraph{$\QIPLconst \subseteq \Ptime$.} We prove this inclusion using a semi-definite program (SDP) for a given \QIPLconst{} proof system, adapted from the SDP formulation for \QIP{} in~\cite{VW16,Watrous16tutorial}. Together with the turn-halving lemma, specifically \Cref{thm:QIPL-properties-informal}\ref{thmitem:QIPL-parallelization-informal}, this inclusion implies that $\QIPUL \subseteq \Ptime$. 

Consider a $(2l)$-turn \QIPLconst{} proof system $\protocol{P}{V}$, where $l \leq O(1)$.
Let $\rho_{\ttM_j\ttW_j}$ and $\rho_{\ttM'_j\ttW_j}$, for $j \in [l]$, denote snapshot states in the register $\sfM$ and $\sfW$ after the $(2j-1)$-st turn and the $(2j)$-th turn in $\protocol{P}{V}$, respectively, as illustrated in \Cref{fig:QIPL-even}. The variables in this SDP correspond to these snapshot states after each prover's action, particularly $\rho_{\ttM'_j\ttW_j}$ for $j\in[l]$, while the objective function is the maximum acceptance probability $\omega(V)$ of $\protocol{P}{V}$. 
Since the verifier's actions are \textit{unitary} circuits, these variables can be treated independently. Hence, the SDP program mainly consists of two types of constraints, assuming that all variables are valid quantum states: 
\begin{enumerate}[label={\upshape(\roman*)}, itemsep=0.33em, topsep=0.33em, parsep=0.33em]
    \item Verifier's actions only operate on the registers $\sfM$ and $\sfW$: 
    \[\rho_{\ttM_{j}\ttW_{j}} = V_{j} \rho_{\ttM'_{j-1}\ttW_{j-1}} V_{j}^{\dagger} \text{ for } j\in \{2,\cdots,l\}, \text{ and } \rho_{\ttM_1\ttW_1} = V_1 \ketbra{\bar{0}}{\bar{0}}_{\sfM\sfW}V_1^{\dagger}.\] 
    \item Prover's actions do not change the verifier's private register: 
    \begin{equation}
        \label{eq:QSZKL-constraint}
        \Tr_{\ttM_j}(\rho_{\ttM_j\ttW_j}) = \Tr_{\ttM'_j}(\rho_{\ttM'_j\ttW_j}) \text{ for } j \in [l].
    \end{equation} 
\end{enumerate}
Since the variables in this SDP collectively hold $O(\log{n})$ qubits, a standard SDP solver (e.g.,~\cite{GM12}) provides a deterministic polynomial-time algorithm for approximately solving it. 

\paragraph{$\QIPLHC \subseteq \NP$.} 
We now extend the above SDP formulation to $l$-round \QIPL{} proof systems, in which the verifier's $j$-th action $\widetilde{V_j}$ is an \textit{almost-unitary} quantum circuit that allows $O(\log{n})$ pinching intermediate measurements. For simplicity, we instead consider the corresponding isometric quantum circuit $V_j$, which introduces a new environment register $\sfE_j$ measured at the end of the turn, with the outcome denoted by $u_j$. 

Recall that $\omega(V)|^u$ represents the contribution of the measurement outcome branch $u=(u_1,\cdots,u_l)$ to the maximum acceptance probability $\omega(V)$. 
Owing to the high-concentration condition, it suffices to consider an \textit{approximation} $\widehat{\omega}(V)|^u$ of $\omega(V)|^u$ for some specific branch $u$, satisfying $\omega(V)|^u \leq \widehat{\omega}(V)|^u \leq \omega(V)$.
These bounds follow from two facts: (1) pinching measurements eliminate coherence between subspaces corresponding to different branches, which enables $\omega(V)|^u$ to be approximately optimized in isolation; and (2) the acceptance probability of any associated global prover strategy across all branches cannot exceed $\omega(V)$. 

By extending the SDP formulation of $\QIPLconst{}$ proof systems, we construct a family of SDP programs depending on the measurement outcome branches $\{u\}$. 
Let $\rho_{\ttM_j\ttW_j}\!\otimes\! \ketbra{u_j}{u_j}_{\sfE_j}$ denote the \textit{unnormalized} snapshot states after measuring $\sfE_j$. 
For a fixed branch $u$, the associated SDP program includes the following three types of constraints: 
\begin{enumerate}[label={\upshape(\roman*')}, itemsep=0.33em, topsep=0.33em, parsep=0.33em]
    \item $\rho_{\ttM_{j}\ttW_{j}} \!\otimes\! \ketbra{u_{j}}{u_{j}}_{\sfE_{j}} = \big( I_{\ttM_{j}\ttW_{j}} \!\otimes\! \ketbra{u_{j}}{u_{j}}_{\sfE_{j}}\big) V_{j}  \rho_{\ttM'_{j-1}\ttW_{j-1}} V_{j}^{\dagger}$ for $j\in \{2,\cdots,l\}$, and 
    
    $\rho_{\ttM_1\ttW_1} \!\otimes\! \ketbra{u_1}{u_1}_{\sfE_1} = \big( I_{\ttM_1\ttW_1} \!\otimes\! \ketbra{u_1}{u_1}_{\sfE_1}\big) V_1 \ketbra{\bar{0}}{\bar{0}}_{\sfM\sfW}V_1^{\dagger}.$
    \item $\Tr_{\ttM_j}(\rho_{\ttM_j\ttW_j} \!\otimes\! \ketbra{u_j}{u_j}_{\sfE_j}) = \Tr_{\ttM'_j}(\rho_{\ttM'_j\ttW_j} \!\otimes\! \ketbra{u_j}{u_j}_{\sfE_j})$ for $j \in [l]$.
    \item $\Tr(\rho_{\ttM_j\ttW_j} \!\otimes\! \ketbra{u_j}{u_j}_{\sfE_j}) = \Tr(\rho_{\ttM'_j\ttW_j} \!\otimes\! \ketbra{u_j}{u_j}_{\sfE_j})$ for $j \in [l]$.
\end{enumerate}
Notably, for the third type of constraints, both sides evaluate to exactly $1$ when the verifier is unitary, as in the cases of \QIP{} and \QIPUL{}. In contrast, for \QIPL{} proof systems, the value varies across different measurement outcome branches and remains bounded above by $1$. Crucially, this value is entirely determined by the verifier's actions and cannot be altered by the prover. 

Next, we explain the \NP{} containment. The classical witness $w$ consists of an $l$-tuple $u$, indicating a specific SDP program, and a feasible solution $(\rho_{\ttM'_1\ttW_1}, \cdots, \rho_{\ttM'_l\ttW_l})$ to this SDP program. This solution can be represented by $l$ square matrices of dimension $\poly(n)$, thus having polynomial size. The verification procedure involves checking (1) whether the solution encoded in $w$ satisfies these SDP constraints based on $u$; and (2) whether $\widehat{\omega}(V)|^{u} \geq c(n)$. All these checks can be verified using basic matrix operations in deterministic polynomial time. 

\subsubsection{Basic properties for \QIPL{} and \QIPUL{}}
We begin by outlining three basic properties of space-bounded (unitary) quantum interactive proof systems, which are dependent on the parameters $c(n)$, $s(n)$, and $m(n)$:

\begin{theorem}[Properties for \QIPL{} and \QIPUL{}, informal of \Cref{thm:QIPL-basic-properties,lemma:QIPL-halving-parallelization}] 
    \label{thm:QIPL-properties-informal}
    Let $c(n)$, $s(n)$, and $m(n)$ be functions such that $0 \leq s(n) < c(n) \leq 1$, $c(n)-s(n) \geq 1/\poly(n)$, and $1 \leq m(n) \leq \poly(n)$.  Then, it holds that\emph{:}
    \begin{enumerate}[label={\upshape(\arabic*)}, itemsep=0.33em, topsep=0.33em, parsep=0.33em] 
        \item \textbf{\emph{Closure under perfect completeness}}.  \label{thmitem:QIPL-perfect-completeness-informal}
        \[ \QIPL_m[c,s] \subseteq \QIPL_{m+2}[1,1-(c-s)^2/2] \text{ and } \QIPUL_m[c,s] \subseteq \QIPUL_{m+2}[1,1-(c-s)^2/2]. \]
        \item \textbf{\emph{Error reduction}}. For any polynomial $k(n)$, \label{thmitem:QIPL-error-reduction-informal}
        \[ \QIPL_m[c,s] \subseteq \QIPL_{m'}\big[1,2^{-k}\big] \text{ and } \QIPUL_m[c,s] \subseteq \QIPUL_{m'}\big[1,2^{-k}\big].\] 
        Here, $m' \coloneqq O\big(km/\log\frac{1}{1-(c-s)^2/2}\big)$.
        \item \textbf{\emph{Parallelization}}. $\QIPUL_{4m+1}[1,s] \subseteq \QIPUL_{2m+1}[1,(1+\sqrt{s})/2]$. \label{thmitem:QIPL-parallelization-informal}
    \end{enumerate}
\end{theorem}

Achieving perfect completeness for \QIPL{} and \QIPUL{} proof systems, particularly \Cref{thm:QIPL-properties-informal}\ref{thmitem:QIPL-perfect-completeness-informal}, can be adapted from the techniques used in \QIP{} proof systems~\cite[Section 4.2.1]{VW16} (or~\cite[Section 3]{KW00}) by adding two additional turns. However, there are important subtleties to consider when establishing the other properties in \Cref{thm:QIPL-properties-informal}. 

\paragraph{Error reduction via sequential repetition.}
Since each message is of size $O(\log n)$, error reduction via \textit{parallel repetition} does not apply to \QIPL{} and \QIPUL{} when the gap $c-s$ is polynomially small, regardless of the number of messages.\footnote{Still, error reduction via parallel repetition works for \QIPL{} when the gap $c-s \geq \Omega(1)$; see \Cref{lemma:QIPL-error-reduction-parallel}.} Alternatively, error reduction via \textit{sequential repetition} requires that the registers $\sfM$ and $\sfW$ (the ``workspace'') must be in the all-zero state (``cleaned'') before each execution of the original proof systems. While this is trivial for \QIP{} proof systems, it poses a challenge for \QIPL{} and \QIPUL{} proof systems because the (almost-)unitary quantum logspace verifier cannot achieve this on its own. 

To establish \Cref{thm:QIPL-properties-informal}\ref{thmitem:QIPL-error-reduction-informal}, our solution is to have \textit{the prover ``clean'' the workspace} while ensuring that the prover behaves honestly. This is achieved through the following proof system: The verifier applies a multiple-controlled adder before each proof system execution, with the adder being activated only when the control qubits are all zero. The verifier then measures the register that the adder acts on and accepts if (1) the workspace is ``cleaned'' for each execution and (2) \textit{all} outcomes of the original proof system executions are acceptance. 

\paragraph{Parallelization and strict uniformity condition for the verifier's mapping.}
The original parallelization technique proposed in~\cite[Section 4]{KW00} applies only to \QIPUL{} (also \QIPL{}) proof systems with a constant number of messages. This limitation stems from the requirement that the prover sends the snapshot states for all $m$ turns in a single message. As $m$ increases, the size of this message grows to $O(m \log n)$, which becomes $\omega(\log n)$ when $m = \omega(1)$. 

To overcome this issue, we adapt the technique from~\cite[Section 4]{KKMV09}, a ``dequantized'' version of the original approach that fully utilizes the \textit{reversibility} of the verifier's actions. Instead of sending all snapshot states in one message, the new verifier performs the original verifier's action or its reverse at any turn in a single action. Specifically, when applying this method to a $(4m+1)$-turn \QIPUL{} proof system $\protocol{P}{V}$, the prover starts by sending only the snapshot state after the $(2m+1)$-st turn. The verifier then chooses $b\in\binset$ uniformly at random: if $b=0$, the verifier continues to interact with the prover according to $\protocol{P}{V}$, keeping the acceptance condition unchanged; while if $b=1$, the verifier executes $\protocol{P}{V}$ in reverse, and finally accepts if its private qubits are all zero.
This proof system, which halves the number of turns, is referred to as the \textit{turn-halving lemma}, as detailed in \Cref{thm:QIPL-properties-informal}\ref{thmitem:QIPL-parallelization-informal}.

Next, we establish \Cref{thm:QIPUL-bounds-informal} by applying the turn-halving lemma $O(\log n)$ times.\footnote{An operation based on $r$ random bits can be simulated by a corresponding unitary controlled by the state $\ket{+}^{\otimes r}$, where $\ket{+} \coloneqq \frac{1}{\sqrt{2}}(\ket{0}+\ket{1})$. Thus, simulating $O(\log n)$ random bits across all turns of the proof system requires $O(\log n)$ ancillary qubits in total, which is feasible for the unitary quantum logspace verifier in \QIPUL{}.\label{footref:simulating-random-bits}} Specifically, any \QIPUL{} proof system with a polynomial number of messages can be parallelized to three messages,\footnote{Although the turn-halving lemma does not directly apply to \QIPL{} proof systems, a similar reasoning works for its reversible generalization $\QIPL^{\diamond}$, reducing a constant number of messages to three.} while the gap $c-s$ of the resulting proof system becomes polynomially small. 
However, this reasoning poses a challenge: the resulting verifier must know all original verifier actions, necessitating a strong notion of uniformity for the verifier's mapping in our definition of \QIPUL{}. 
In addition, to prove \Cref{thm:QIPLconst-informal}, we adopt a similar approach to that used for \QIP{}, particularly $\QIP[3] \subseteq \QMAM{}$~\cite{MW05}, which inspired the turn-halving lemma~\cite[Section 4]{KKMV09}, and an exponentially down-scaling version of the work~\cite{JJUW11}.

\subsubsection{Lower bounds for \QIPL{} and \QIPUL{}}

\paragraph{$\NP{} \subseteq \QIPL{}$.}
This inclusion draws inspiration from the interactive proof system in~\cite[Lemma 2]{CL95} and presents a challenge in adapting this proof system to the \QIPL{} setting. 
Notably, our construction essentially provides a \QIPLHC{} proof system, since the pinching measurement outcomes are \textit{unique} (even stronger than the high-concentration condition) for \textit{yes} instances. 

We start by outlining this \QIPL{} proof system for \threeSAT{}. 
Consider a \threeSAT{} formula 
\[\phi = C_1 \vee C_2 \vee C_3 = (x_1 \vee x_2 \vee x_3) \wedge (\neg x_4 \vee \neg x_2 \vee x_3) \wedge (x_4 \vee \neg x_1 \vee \neg x_3)\]
with $k=3$ clauses and $n=4$ variables. An assignment $\alpha$ of $\phi$ assigns each variable $x_j$ for $j \in [n]$ a value $\alpha_j$ of either $\top$ (true) or $\perp$ (false). 
To verify whether $\phi$ is satisfied by the assignment $\alpha$, we encode $\phi(\alpha)$ as $\Enc(\phi(\alpha))$, consisting of $3k$ triples $(l,i,v)$, where $l$ denotes the literal (either $x_j$ or $\neg x_j$), $i$ represents the $i$-th clause, and $v$ denotes the value assigned to $l$. The prover's actions are divided into two phases: 
\begin{enumerate}[label={\upshape(\roman*)}, itemsep=0.33em, topsep=0.33em, parsep=0.33em]
    \item \textsc{Consistency Check} (for variables). The prover sends one by one all the triples $(l,i,v)$ in $\Enc(\phi(\alpha))$, ordered by the variable $\var(l)$ corresponding to the literal $l$; 
    \item \textsc{Satisfiability Check} (for clauses). For each $i\in\{1,\ldots,k\}$, the prover sends the three triples $(l_1,i,v_1)$, $(l_2,i,v_2)$, and $(l_3,i,v_3)$ in $\Enc(\phi(\alpha))$.
\end{enumerate}

The verifier's actions are as follows. 
To prevent the prover from entangling with the verifier and revealing the private coins, the verifier measures the received messages in the computational basis at the beginning of each action, interpreting the measurement outcomes as the prover's messages. Therefore, it suffices to establish soundness \textit{against classical messages}. 

We now focus on this specific proof system. In Phase (i), the verifier checks whether the assigned values to the same variable are consistent. Since the verifier's actions are almost-unitary circuits and \textit{cannot discard information}, this seems challenging. Our solution is that the verifier keeps only the current and the previous triples, returning the previous triple to the prover in the next turn. In Phase (ii), the verifier checks whether each batch of three triples is satisfied and returns them immediately. 
Lastly, to ensure that the multisets of triples from Phase (i) and (ii) are identical, the verifier computes the ``fingerprint'' of these multisets,\footnote{See \Cref{subsec:classical-concepts-and-tools} for the definition of the fingerprint of a multiset. The computation of each fingerprint requires $O(\log n)$ random bits, which can be simulated in a \QIPL{} proof system; see \Cref{footref:simulating-random-bits} for details.} triple by triple, and compares the fingerprints from both phases at the end. 
The verifier accepts if all checks succeed. 

Using the fingerprinting lemma~\cite{Lipton90}, we prove the correctness of this proof system, showing that $\threeSAT{} \in \QIPL_{8k}[1,1/3]$. 
Interestingly, when combined with the inclusion $\QIPLHC{} \subseteq \NP{}$, this proof system implies (indirect) error reduction for \QIPLHC{} (see \Cref{remark:QIPLHC-properties}). 

\paragraph{$\SAC^1 \subseteq \QIPUL$.} This inclusion is inspired by the interactive proof system in~\cite[Section 3.4]{Fortnow89}. By using error reduction for \QIPUL{}, specifically \Cref{thm:QIPL-properties-informal}\ref{thmitem:QIPL-error-reduction-informal}, it remains to demonstrate that $\SAC^1 \subseteq \QIPUL[1,1-1/\poly(n)]$. 
A Boolean circuit is defined as a (uniform) $\SAC^1$ circuit $C$ if it is an $O(\log{n})$-depth Boolean circuit that employs unbounded fan-in OR gates, bounded fan-in AND gates, and negation gates at the input level. 

The interactive proof system for evaluating the circuit $C$ starts at its top gate. If the gate is an OR, the prover selects a child gate; if it's an AND, the verifier flips a coin to select one. This process repeats until reaching an input $x_i$ or its negation, with the verifier accepting if $x_i=1$ or $x_i=0$, respectively.
Since the computational paths in $C$ do not interfere, extending soundness against classical messages, following directly from~\cite[Section 3.4]{Fortnow89}, to quantum messages can be done by measuring the registers $\sfM$ and $\sfW$ in the computational basis at the end of the verifier's last turn. Finally, given that $C$ has $O(\log{n})$ depth, implementing the verifier's actions requires only $O(\log n)$ ancillary qubits, which is indeed achievable by a unitary verifier. 

\subsubsection{The equivalence of \QSZKUL{} and \BQL{}}

We demonstrate \Cref{thm:QSZKL-eq-BQL-informal} by introducing a \QSZKULHV{}-complete problem: 

\begin{theorem}[Informal of \Cref{thm:IndivProdQSD-QSZKLcomplete}]
    \label{thm:IndivProdQSD-QSZKLcomplete-informal}
    \IndivProdQSD{} is \QSZKULHV{}-complete.
\end{theorem}

We begin by informally defining the promise problem \textsc{Individual Product State Distinguishability}, denoted by $\IndivProdQSD[k(n),\alpha(n),\delta(n)]$, where the parameters satisfy $\alpha(n) - k(n) \cdot \delta(n) \geq 1/\poly(n)$ and $1 \leq k(n) \leq \poly(n)$. This problem considers two $k$-tuples of $O(\log n)$-qubit quantum states, denoted by $\sigma_1,\cdots,\sigma_k$ and $\sigma'_1,\cdots,\sigma'_k$, where the purifications of these states can be prepared by corresponding polynomial-size unitary quantum circuits acting on $O(\log n)$ qubits. For \textit{yes} instances, these two $k$-tuples are ``globally'' far, satisfying
\begin{equation}
    \label{eq:intro-far}
    \td\rbra*{\sigma_1 \otimes \cdots \otimes \sigma_k, \sigma'_1 \otimes \cdots \otimes \sigma'_k} \geq \alpha. 
\end{equation}
While for \textit{no} instances, each pair of corresponding states in these $k$-tuples are close, satisfying
\begin{equation}
    \label{eq:intro-indiv-close}
    \forall j \in [k], \quad \td\rbra*{\sigma_j,\sigma'_j} \leq \delta.
\end{equation}
Then we show that (1) the complement of \IndivProdQSD{}, 
\coIndivProdQSD{}, is \QSZKULHV{}-hard; and (2) 
\IndivProdQSD{} is in \BQL{}, which is contained in \QSZKULHV{} by definition.

\paragraph{\coIndivProdQSD{} is \QSZKULHV{}-hard.} The hardness proof draws inspiration from~\cite[Section 5]{Wat02QSZK}. Consider a $\QSZKULHV[2k,c,s]$ proof system, denoted by $\calB$. The logspace-bounded simulator $S_{\calB}$ produces good state approximations $\xi_j$ and $\xi'_j$ of the snapshot states $\rho_{\ttM_j\ttW_j}$ and $\rho_{\ttM'_j\ttW_j}$ after the $(2j-1)$-st turn and the $(2j)$-th turn in $\calB$, respectively, satisfying $\xi_j \approx_{\delta} \rho_{\ttM_j\ttW_j}$ and $\xi'_j \approx_{\delta} \rho_{\ttM'_j\ttW_j}$, where $\delta_{\calB}(n)$ is a negligible function. 

Since the verifier's actions are unitary and the verifier is honest, it suffices to check that the prover's actions do not change the verifier's private register, corresponding to the type (ii) constraints \Cref{eq:QSZKL-constraint} in the SDP formulation for \QIPL{} proof systems. For convenience, let $\sigma_j \coloneq \Tr_{\ttM_j}(\xi_j)$ and $\sigma'_j \coloneq \Tr_{\ttM'_j}(\xi'_j)$ for $j\in[k]$. 
We then establish \QSZKLHV{} hardness as follows: 
\begin{itemize}[itemsep=0.33em, topsep=0.33em, parsep=0.33em]
    \item For \textit{yes} instances, the message-wise closeness condition of the simulator $S_{\calB}$ implies \Cref{eq:intro-indiv-close} with $\delta(n) \coloneqq 2\delta_{\calB}(n)$. 
    \item For \textit{no} instances, the simulator $S_{\calB}$ produces the snapshot state before the final measurement, which accepts with probability $c(n)$ for all instances, while the proof system accepts with probability at most $s(n)$. The inconsistency between the simulator's state approximations and the snapshot states yields \Cref{eq:intro-far} with $\alpha(n) \coloneqq (\sqrt{c}-\sqrt{s})^2/4(l-1)$.  
\end{itemize}

\paragraph{$\IndivProdQSD{} \in \BQL$.} Since it holds that  $\BQL=\QMAL$~\cite{FKLMN16,FR21}, it suffices to establish that $\IndivProdQSD{} \in \QMAL$. By applying an averaging argument in combination with \Cref{eq:intro-far}, we derive the following: 
\begin{equation}
    \label{eq:intro-indiv-far}
    \sum_{j\in[k]} \td\rbra*{\sigma_j,\sigma'_j} \geq \td\rbra*{\sigma_1 \otimes \cdots \otimes \sigma_k, \sigma'_1 \otimes \cdots \otimes \sigma'_k} \geq \alpha 
    \quad \Rightarrow \quad 
    \exists j\in[k] \text{ s.t. } \td\rbra*{\sigma_j,\sigma'_j} \geq \frac{\alpha}{k}. 
\end{equation}

The \QMAL{} protocol works as follows: (1) The prover sends an index $i\in[k]$ to the verifier; and (2) The verifier accepts if $\Tr(\sigma_i,\sigma'_i) \geq \alpha/k$ and rejects if $\Tr(\sigma_i,\sigma'_i) \leq \delta$, in accordance with \Cref{eq:intro-indiv-far,eq:intro-indiv-close}. The resulting promise problem to be verified is precisely an instance of \GapQSDlog{}, which is known to be \BQL{}-complete~\cite{LGLW23}. 

\subsection{Discussion and open problems}

We introduce two models of space-bounded quantum interactive proof systems: \QIPL{} and \QIPUL{}. Unlike $\BQL = \BQUL$, we show that $\QIPUL \subsetneq \QIPL$ unless $\Ptime = \NP$. Our results highlight the distinctive role of (pinching) intermediate measurements in space-bounded quantum interactive proofs, setting them apart from space-bounded quantum computation. This prompts an intriguing question: 

\begin{enumerate}[label={\upshape(\alph*)},itemsep=0.3em,topsep=0.3em,parsep=0.3em]
    \item What is the computational power of space-bounded quantum interactive proofs beyond \QIPLHC{}, particularly when the high-concentration requirement for \textit{yes} instances (the completeness condition) is removed, as in the class \QIPL{}, or when the verifier is allowed to perform general quantum logspace computations?  \label{probitem:generalQIPL}
\end{enumerate}

A motivating example is a reversible generalization of \QIPL{}, particularly space-bounded \textit{isometric} quantum interactive proof systems ($\QIPL^{\diamond}$, see \Cref{remark:reversible-QIPL}), where all verifier actions are $O(\log{n})$-qubit \textit{isometric} quantum circuits. Remarkably, $\QIPL^{\diamond}$ at least contains \QMA{}:
Given a local Hamiltonian $H = \sum_{i=1}^m H_i$, we can construct a $\QIPL^{\diamond}$ proof system as follows:\footnote{A similar approach is used in a streaming version of \QMAL{} (with online access to the message) in~\cite{GR23online}.} 
\begin{enumerate}[label={\upshape(\roman*)},itemsep=0.3em,topsep=0.3em,parsep=0.3em]
    \item The verifier chooses a local term $H_i$ uniformly at random from the set $\{H_1,\cdots,H_m\}$. 
    \item The prover sends a ground state $\ket{\Omega}$ qubit by qubit, while the verifier sends a state $\ket{0}$ in each round and retains only the qubits associated with $H_i$ in its private registers. 
    \item The verifier performs the POVM corresponding to the decomposition $I = H_i + (I - H_i)$.\footnote{See the proof of~\cite[Proposition 14.2]{KSV02} for an explicit construction of such POVMs.} 
\end{enumerate}

Further analysis indicates that the verifier accepts with probability $1-m^{-1} \bra{\Omega} H \ket{\Omega}$, and direct sequential repetition yields a $\QIPL^{\diamond}$ proof system. 
Additionally, it is evident that all candidate models of Question \ref{probitem:generalQIPL} are contained in $\QIP$, and thus in $\PSPACE$.

\vspace{0.5em}
Furthermore, space-bounded \textit{unitary} quantum interactive proofs (\QIPUL{}) can simulate the classical counterparts with $O(\log{n})$ public coins~\cite{Fortnow89} (see \Cref{thm:QIPUL-bounds-informal}), raising the question: 
\begin{enumerate}[label={\upshape(\alph*)},itemsep=0.3em,topsep=0.3em,parsep=0.3em]
    \setcounter{enumi}{1}
    \item Can we achieve a tighter characterization of \QIPUL{}? For example, does \QIPUL{} contain space-bounded classical interactive proofs with $\omega(\log{n})$ public coins? \label{probitem:QIPUL}
\end{enumerate}

\vspace{0.5em}
Finally, for \textit{constant}-turn space-bounded quantum interactive proofs, the three models discussed here become equivalent due to the principle of deferred measurements, contrasting with the aforementioned polynomial-turn settings. However, this equivalence does not directly extend to more general verifiers (see \Cref{footref:constant-turn-QIPL}), leading to the following question: 

\begin{enumerate}[label={\upshape(\alph*)},itemsep=0.3em,topsep=0.3em,parsep=0.3em]
    \setcounter{enumi}{2}
    \item What is the computational power of constant-turn space-bounded quantum interactive proofs with a general quantum logspace verifier?  \label{probitem:constTurnQIPL}
\end{enumerate}

\subsection{Related works}
Variants of time-bounded quantum interactive proofs with short messages were explored in~\cite{BSW11,Pereszlenyi12}. Depending on the settings, these variants are as powerful as \QMA{} or \BQP{}. 

The concept of interactive proof systems has been extended to other computational models. Quantum interactive proofs for synthesizing quantum states, known as $\mathsf{stateQIP}$, were introduced in~\cite{RY22}. Follow-up research established the equivalence $\mathsf{stateQIP} = \mathsf{statePSPACE}$~\cite{MY23} and developed a parallelization technique for $\mathsf{stateQIP}$~\cite{INNSY22,Rosenthal23}. A Merlin-Arthur-type variant was also explored in~\cite{DLGLM23,DLG24}. More recently, quantum interactive proofs for unitary synthesis and related problems have been studied in~\cite{BEMPQY23,LMW24}.
Another interesting but less related variant is the exploration of interactive proof systems in distributed computing~\cite{KOS18,NPY20}, and more recently, quantum distributed interactive proof systems have been investigated~\cite{FPLGN21,LGMN23,HKN24}.

Finally, space-bounded (classical) statistical zero-knowledge, where the verifier has \textit{read-only (i.e., two-way) access} to (polynomial-length) messages during interactions, was studied in~\cite{DGRV11,AHT23,AGMTW22}. More recently, a variant where the verifier has \textit{online (i.e., one-way) access} to messages has also been explored~\cite{CDGH23}. 

\section{Preliminaries}

We assume that the reader is familiar with quantum computation and the theory of quantum information. For an introduction, the textbooks by~\cite{NC10} and~\cite{deWolf19} provide a good starting point, while for a more comprehensive survey on quantum complexity theory, refer to~\cite{Watrous08}. 

We introduce several conventions throughout the paper: (1) we denote $\sbra{n} \coloneqq \cbra{1, 2, \dots, n}$; (2) we use the logarithmic function $\log(x)$ with base $2$; and (3) we utilize the notation $\ket{\bar{0}}$ to represent $\ket{0}^{\otimes a}$ with $a>1$.
In addition to these conventions, we provide two useful definitions. 
We say that $\calI = (\calI_{\yes}, \calI_{\no})$ is a \textit{promise problem}, if it satisfies that $\calI_{\yes} \cap \calI_{\no} =\emptyset$ and $\calI_{\yes} \cup \calI_{\no} \subseteq \binset^*$. For simplicity, we use the abbreviation $x\in\calI$ to denote $x\in \calI_{\yes} \cup \calI_{\no}$.
A function $\mu(n)$ is said to be \textit{negligible}, if for every integer $c \geq 1$, there is an integer $n_c > 0$ such that for all $n \geq n_c$, $\mu\rbra{n} < n^{-c}$. 

\subsection{Distance-like measures for quantum states}

We will provide an overview of relevant quantum distances and divergences, along with useful inequalities among different quantum distance-like measures. We say that a square matrix $\rho$ is a \textit{quantum state} if $\rho$ is positive semi-definite and $\Tr(\rho)=1$.

\begin{definition}[Trace distance and fidelity]
    \label{def:quantum-distances}
    For any quantum states $\rho_0$ and $\rho_1$, we define two distance-like measures: 
    \begin{itemize}
        \item \textnormal{\textbf{Trace distance}.} $\td(\rho_0,\rho_1)\coloneqq\frac{1}{2}\Tr|\rho_0-\rho_1|=\frac{1}{2}\Tr(((\rho_0-\rho_1)^\dagger(\rho_0-\rho_1))^{1/2}).$
        \item \textnormal{\textbf{(Uhlmann) Fidelity}.} $\F(\rho_0,\rho_1)\coloneqq\Tr|\sqrt{\rho_0}\sqrt{\rho_1}|$.
    \end{itemize}
\end{definition}

We begin by listing two useful bounds on tensor-product quantum states with respect to the trace distance: 

\begin{lemma}[Trace distance on tensor-product states, adapted from Exercise 9.1.2 and Corollary 9.1.10 in~\cite{Wilde13}]
    \label{lemma:trace-distance-product-states}
    For any quantum states $\rho_1 \otimes \cdots \otimes \rho_k$ and $\rho'_1 \otimes \cdots \otimes \rho'_k$, where $\rho_i$ and $\rho'_i$ use the same number of qubits for all $i\in[k]$, it holds that
    \begin{enumerate}[label={\upshape(\arabic*)}, itemsep=0.33em, topsep=0.33em, parsep=0.33em]
        \item $\forall i \in [k]$, $\td\rbra*{ \rho_i, \rho'_i } \leq \td\rbra*{\rho_1 \otimes \cdots \otimes \rho_k, \rho'_1 \otimes \cdots \otimes \rho'_k}$.
        \item $\td\rbra*{\rho_1 \otimes \cdots \otimes \rho_k, \rho'_1 \otimes \cdots \otimes \rho'_k} \leq \sum_{i\in [k]} \td \rbra*{ \rho_i, \rho'_i }$. 
    \end{enumerate}
\end{lemma}

We then provide two fundamental properties of the trace distance.

\begin{lemma}[Data-processing inequality for the trace distance, adapted from~{\cite[Theorem 9.2]{NC10}}]
    \label{lemma:trace-distance-data-processing}
    Let $\rho_0$ and $\rho_1$ be quantum states. For any quantum channel $\calE$, it holds that
    \[ \td\rbra*{\calE(\rho_0), \calE(\rho_1)} \leq \td(\rho_0,\rho_1). \]
\end{lemma}

\begin{lemma}[Unitary invariance for the trace distance, adapted from~{\cite[Equation (9.21)]{NC10}}]
    \label{lemma:trace-distance-unitary-invariance}
    Let $\rho_0$ and $\rho_1$ be quantum states. For any unitary transformation $U$, it holds that
    \[ \td\rbra*{U\rho_0 U^{\dagger}, U\rho_1 U^{\dagger}} = \td(\rho_0,\rho_1). \]
\end{lemma}

Next, we present two basic properties for the fidelity. 

\begin{lemma}[Data-processing inequality for the fidelity, adapted from Theorem 9.6 in~\cite{NC10}]
    \label{lemma:fidelity-data-processing}
    Let $\rho_0$ and $\rho_1$ be quantum states. For any quantum channel $\calE$, it holds that
    \[ \F\rbra*{\calE(\rho_0), \calE(\rho_1)} \geq \F(\rho_0,\rho_1). \]
\end{lemma}

\begin{lemma}[{\cite[Lemma 2]{SR01} \&~\cite[Lemma 3.3]{NS03}}]
    \label{lemma:sum-of-squared-fidelity}
    Let $\rho_0$ and $\rho_1$ be $m$-qubit quantum states. Then, for any $m$-qubit quantum state $\xi$, it holds that
    \[ \F(\rho_0,\xi)^2 + \F(\xi,\rho_1)^2 \leq 1 + \F(\rho_0,\rho_1). \]
\end{lemma}

Lastly, we present a lemma concerning the freedom in purifications of quantum states: 
\begin{lemma}[Unitary equivalence of purifying the same state, adapted from~{\cite[Exercise 2.81]{NC10}}]
    \label{lemma:unitary-equivalent-in-purifications}
    Let $\ket{\psi}_{\sfA\sfB}$ and $\ket{\phi}_{\sfA\sfB}$ be pure states on the registers $\sfA$ and $\sfB$ such that 
    \[\Tr_{\sfB}\rbra*{\ketbra{\psi}{\psi}_{\sfA\sfB}} = \rho_{\sfA} = \Tr_{\sfB}\rbra*{\ketbra{\phi}{\phi}_{\sfA\sfB}},\] where $\rho_A$ is a mixed state on the register $\sfA$. Then, there exists a unitary transformation $U_{\sfB}$ acting on the register $\sfB$ such that $\ket{\psi}_{\sfA\sfB} = \rbra*{I_{\sfA} \otimes U_{\sfB}} \ket{\phi}_{\sfA\sfB}$. 
\end{lemma}

\subsection{Space-bounded quantum computation}

We say that a function $s(n)$ is \textit{space-constructible} if there exists a deterministic space $s(n)$ Turing machine that takes $1^n$ as an input and outputs $s(n)$ in the unary encoding. 
Moreover, we say that a function $f(n)$ is $s(n)$-\textit{space computable} if there exists a deterministic space $s(n)$ Turing machine that takes $1^n$ as an input and outputs $f(n)$. 
Our definitions of space-bounded quantum computation are formulated in terms of \textit{quantum circuits}. For a discussion on the equivalence between space-bounded quantum computation using \textit{quantum circuits} and \textit{quantum Turing machines}, we refer readers to \cite[Appendix A]{FL18} and \cite[Section 2.2]{FR21}. 

We begin by introducing three types of space-bounded quantum circuit families, as formalized in \Cref{def:space-bounded-quantum-circuits}. Our definitions align with~\cite[Section 2.3]{VW16}. Throughout this work, we adopt the shorthand notation $C_x$ to indicate that the circuit $C_{|x|}$ takes input $x$.

\begin{definition}[Space-bounded quantum circuit families: unitary, almost-unitary, and isometric]
    \label{def:space-bounded-quantum-circuits}
    Let us define three types of quantum circuits: 
    \begin{itemize}[itemsep=0.33em, topsep=0.33em, parsep=0.33em, leftmargin=2em]
        \item \emph{\textbf{Unitary quantum circuit}}. A unitary quantum circuit consists of a sequence of unitary quantum gates, each of which belongs to some fixed gate set that is universal for quantum computation, such as $\{\Had, \CNOT, \T\}$. 
        \item \emph{\textbf{Almost-unitary quantum circuit}}. An almost-unitary quantum circuit generalizes a unitary quantum circuit acting on $O(s(n))$ qubits by allowing $O(s(n))$ single-qubit measurement gates $\M$ in the computational basis, which are defined by a \emph{pinching} channel\emph{:}\footnote{See Section 4.1.1 (particularly Definition 4.4) in~\cite{Watrous18} for further details. This type of channel, also referred to as a completely phase-damping channel, is discussed in~\cite[Equation (2.24)]{VW16}.} 
        \[\Phi^{\M}(\rho) \coloneqq \ketbra{0}{0} \Tr\rbra*{M_0 \rho} + \ketbra{1}{1} \Tr\rbra*{M_1 \rho}, \text{ where } M_b \coloneqq \ketbra{b}{b} \text{ for } b\in\binset.\] 
 
        \item \emph{\textbf{Isometric quantum circuit}}. An isometric quantum circuit extends a unitary quantum circuit acting on $O(s(n))$ qubits by allowing $O(s(n))$ ancillary gates. An ancillary gate is a non-unitary gate that takes no input and produces a single qubit in the state $\ket{0}$ as output. 
    \end{itemize}
    For convenience, we treat almost-unitary quantum circuits as a special case of isometric quantum circuits.\footnote{More specifically, by applying the principle of deferred measurements (e.g.,~\cite[Section 4.4]{NC10}) to an almost-unitary quantum circuit, we obtain an isometric quantum circuit in which each measurement gate is simulated by an ancillary gate, and all ancillary qubits are measured at the end.\label{footref:almost-unitary-is-isometry}}
    For a promise problem $\calI = (\calI_{\yes},\calI_{\no})$, a family of unitary, almost-unitary, or isometric quantum circuits $\{C_x: x\in \calI\}$ is called \emph{$s(n)$-space-bounded} if there is a deterministic Turing machine that, given any input $x \in \calI$ with input length $n \coloneqq |x|$, runs in space $O(s(n))$ (and hence time $2^{O(s(n))}$) and outputs a description of $C_x$, where $C_x$ accepts if $x \in \calI_{\yes}$, rejects if $x \in \calI_{\no}$, acts on $O(s(n))$ qubits, and consists of $2^{O(s(n))}$ gates.
\end{definition}

\begin{remark}[Subtleties on space-bounded quantum circuit families]
    \label{remark:subtleties-space-bounded-circuits}
    In the context of space-bounded quantum circuits, as defined in \Cref{def:space-bounded-quantum-circuits}, there are important subtleties: 
    \begin{enumerate}[label={\upshape(\arabic*)}, topsep=0.33em, itemsep=0.33em, parsep=0.33em, leftmargin=2em]
        \item Unlike a full projective measurement that collapses a state into a specific outcome, a pinching measurement preserves a mixed-state structure while imposing \textit{decoherence} between subspaces. As a result, almost-unitary quantum circuits are \textit{oblivious} to intermediate measurement outcomes, meaning their description remains independent of these outcomes despite having access to them. This subtlety prevents qubits in such circuits from being directly reset to zero, differentiating them from isometric quantum circuits. 
        \label{remarkitem:pinching-measurement}
        \item Space-bounded unitary and almost-unitary quantum circuits are equivalent for promise problems via the principle of deferred measurements. However, such equivalences are unknown in more general settings, analogous to the scenario in \Cref{footref:BQL-eq-BQUL-strong-forms}. 
    \end{enumerate}
\end{remark}

In this work, we focus on (log)space-bounded quantum circuits with $s(n) = O(\log{n})$. 
The complexity classes corresponding to space-bounded unitary and general quantum circuits with $s(n) = \Theta(\log(n))$ are known as \BQUL{} and \BQL{}, respectively. 
As described in~\cite[Section 2.3]{VW16}, a \textit{general quantum circuit} extends a unitary quantum circuit by including ancillary gates and \textit{erasure gates}.\footnote{An erasure gate is a non-unitary gate that takes a single qubit as input and produces no output. Alternatively, a general quantum circuit can also be defined by extending a unitary quantum circuit with measurement gates and reset-to-zero gates, as in~\cite{FR21,GR22}. Notably, a reset-to-zero gate can be simulated by first applying an erasure gate to remove the original qubit and then using an ancillary gate to introduce a new qubit.}  
It has been established that $\BQUL{}=\BQL{}$ for promise problems~\cite{FR21} (see also~\cite{GRZ21,GR22}), whereas such equivalences remain unproven in more general forms.\footnote{Specifically, this refers to the transformation of a unitary quantum logspace circuit $C'$ from a general quantum logspace circuit $C$ (with all-zero states as input) such that the final state of $C$ and $C'$ are identical. This is a stronger requirement than merely ensuring that the output qubits of these circuits are the same. This general form only can be simulated in $\NC^2$~\cite{Wat99,Wat03}.\label{footref:BQL-eq-BQUL-strong-forms}}
For detailed definitions and known properties of these classes, we refer to~\cite[Section 2.3]{LGLW23} as a brief introduction. 

Lastly, we define \textit{logspace \emph{(}many-to-one\emph{)} reductions}. We begin by slightly abusing notation and considering parameterized promise problems of the form $\calI = \cbra*{\calI_{n,t_1,\cdots,t_r}}_{n \in \bbN}$ for functions $t_1,\cdots,t_r \colon \bbN \rightarrow \bbR$, where $\calI_{n,t_1,\cdots,t_r}$ consists of instances of size $n$ which satisfy conditions expressed in terms of $t_1(n),\cdots,t_r(n)$. 
We say that $\calI = \cbra*{\calI_{n,t_1,\cdots,t_r}}_{n \in \bbN}$ is (many-to-one) reducible to $\calI' = \big\{\calI'_{m,t'_1,\cdots,t'_{r'}}\big\}_{m \in \bbN}$ if there exist $(r+1)$-variable real polynomials $p_0, \cdots, p_{r'}$ such that for all $n\in \bbN$, there exists a function $g_n \colon \calI_{n,t_1,\cdots,t_r} \rightarrow \calI'_{m,t'_1,\cdots,t'_{r'}}$ satisfying the following conditions: (1) $m = p_0\rbra*{n, t_1(n), \cdots, t_r(n)}$; (2) $t'_j(m) = p_j(n,t_1(n),\cdots, t_r(n))$ for all $j \in [r']$; (3) $g_n(x) \in \calI'_{m,t'_1,\cdots,t'_{r'}}$ for all $x \in \calI_{n,t_1,\cdots,t_r}$. If the family of functions $\cbra*{g_n}_{n\in \bbN}$ is computable in deterministic logspace, we say that $\calI$ is logspace-reducible to $\calI'$, denoted by $\calI \leq^m_{\Lspace} \calI'$. 

\subsection{Space-bounded quantum state testing}
\label{subsec:space-bounded-state-testing}

We begin by defining the space-bounded quantum state testing problem with respect to the trace distance, denoted as \GapQSDlog{}:

\begin{definition}[Space-bounded Quantum State Distinguishability Problem, adapted from Definition 4.1 and 4.2~\cite{LGLW23}]
    \label{def:GapQSDlog}
    Let $\alpha(n)$, $\beta(n)$, $r(n)$ be logspace computable functions such that $0 \leq \beta(n) < \alpha(n) \leq 1$, $\alpha(n)-\beta(n) \geq 1/\!\poly(n)$ and $1 \leq r(n) \leq O(\log n)$. Let $Q_0$ and $Q_1$ be polynomial-size unitary quantum circuits acting on $O(\log n)$ qubits, with $r(n)$ specified output qubits. Here, $n$ represents the total number of gates in $Q_0$ and $Q_1$. For $b\in\binset$, let $\rho_b$ denote the quantum states obtained by running $Q_b$ on the all-zero state $\ket{\bar{0}}$ and tracing out the non-output qubits, then the promise is that one of the following holds:    
    \begin{itemize}[itemsep=0.33em,topsep=0.33em,parsep=0.33em]
        \item \emph{Yes} instances: A pair of quantum circuits $(Q_0,Q_1)$ such that $\td(\rho_0,\rho_1) \geq \alpha(n)$;
        \item \emph{No} instances: A pair of quantum circuits $(Q_0,Q_1)$ such that $\td(\rho_0,\rho_1) \leq \beta(n)$. 
    \end{itemize}
\end{definition}

Moreover, we use the notation \coGapQSDlog{} to denote the \textit{complement} of \GapQSDlog{} with respect to the chosen parameters $\alpha(n)$ and $\beta(n)$.
As established in~\cite{LGLW23}, \GapQSDlog{} is \BQL{}-complete, and we are particularly interested in the \BQL{} containment:

\begin{theorem}[\GapQSDlog{} is in \BQL{}, adapted from~{\cite[Theorem 4.10]{LGLW23}}]
    \label{thm:GapQSDlog-in-BQL}
    Let $\alpha(n)$ and $\beta(n)$ be logspace computable functions such that $\alpha(n)-\beta(n) \geq 1/\poly(n)$. It holds that 
    \[\GapQSDlog[\alpha(n),\beta(n)] \in \BQL{}.\] 
\end{theorem}

Lastly, it is worth noting that by removing the space constraints on the quantum circuits $Q_0$ and $Q_1$ and allowing $r(n) \leq n$, where $n$ denotes the input length of these state-preparation circuits, we obtain a variant of \Cref{def:GapQSDlog} that aligns with the definition of $\GapQSD[\alpha(n), \beta(n)]$. This promise problem was considered in~\cite{Wat02QSZK} with the condition $\alpha^2 > \beta$, referred to as $\QSD[\alpha(n), \beta(n)]$. 

\subsection{Classical concepts, tools, and complexity classes}
\label{subsec:classical-concepts-and-tools}

\paragraph{\threeSAT{}.}
The \threeSAT{} problem is one of the simplest examples of \NP{}-complete problems. We provide only a brief introduction to \threeSAT{} here. For further details, see~\cite[Section 2.3]{AB09}. 

A \threeSAT{} formula can be written as $\phi = C_1 \wedge \cdots \wedge C_k$, where each clause $C_i$ for $i \in [k]$ is of the form $\big(l_1^{(i)} \vee l_2^{(i)} \vee l_3^{(i)}\big)$, with each literal $l_j^{(i)}$ being either one of the variables $x_1,\cdots, x_n$ or its negation. For instance, $(x_1 \vee x_2 \vee x_3) \wedge (\neg x_4 \vee \neg x_2 \vee x_3) \wedge (x_4 \vee \neg x_1 \vee \neg x_3)$ illustrates the structure. An assignment of a \threeSAT formula assigns each variable $x_j$ for $j \in [n]$ a value of either $\top$ (true) or $\bot$ (false). The \threeSAT{} problem aims to decide whether a given formula $\phi$ is satisfiable. We say that $\phi$ is satisfiable if there exists an assignment $\alpha$ such that $\Phi(\alpha) = \top$.  

\begin{lemma}[{\cite[Exercise 4.6]{AB09}}]
    \label{lemma:3SAT-NP-hard}
    \threeSAT{} is \NP-complete{} under logspace reductions.
\end{lemma}

\paragraph{Fingerprinting of multisets.}
A fingerprint of a multiset $\{x_1,\cdots,x_k\}$, where all elements are non-negative integers and duplicates are allowed, is defined as $\Pi_{i=1}^k (x_i+r) \mod p$, with $p$ being a prime and $r \in [p-1]$. The fingerprinting lemma~\cite{Lipton90} aims to compare whether two multisets are equal by using short fingerprints: 

\begin{lemma}[Fingerprinting lemma, adapted from~{\cite[Theorem 3.1]{Lipton90}}]
    \label{lemma:fingerprint}
    Let $A \coloneqq \{x_1,\cdots,x_{\ell_1}\}$ and $B \coloneqq \{y_1,\cdots,y_{\ell_2}\}$ be two multisets in which all elements are $b$-bit non-negative integers, with $\ell \coloneqq \max\{\ell_1,\ell_2\}$. 
    If the prime $p$ is chosen uniformly at random from the interval $[(b\ell)^2,2(b\ell)^2]$ and the integer $r$ is chosen uniformly at random from the interval $[1,p-1]$, the probability that the \emph{distinct} multisets $A$ and $B$ produce the \emph{same} fingerprint is at most $O\big(\frac{\log{b}+\log{\ell}}{b\ell} + \frac{1}{b^2 \ell}\big)$.    
\end{lemma}

\paragraph{(Uniform) $\SAC^1$.} 
The complexity class (uniform) $\SAC^1$ is a restricted subclass of (uniform) $\AC^1$. Throughout this paper, $\SAC^1$ circuits will refer to (logspace-)uniform $\SAC^1$ circuits. We define $\SAC^1$ circuits and their corresponding circuit evaluation problem as follows: 
\begin{definition}[\textsc{Uniform $\SAC^1$ Circuit Evaluation}, adapted from~\cite{BCD+89}]
    \label{def:uniform-SAC1}
    A Boolean circuit $C \colon \binset^n \rightarrow \binset$ is defined as an $\SAC^1$ circuit if it has depth $O(\log{n})$, includes unbounded fan-in OR \emph{($\vee$)} gates, bounded fan-in AND \emph{($\wedge$)} gates \emph{(}e.g., with fan-in $2$\emph{)}, and has negation \emph{($\neg$)} gate restricted to the input level. The problem is to decide whether a given (logspace-)uniform $\SAC^1$ circuit $C$, whose description can be computed by a deterministic logspace Turing machine, evaluates to $1$.  
\end{definition}

Venkateswaran~\cite{Venkateswaran87} established that $\SAC^1$ is equivalent to \LOGCFL{}, the complexity class consists of languages that are logspace-reducible to context-free languages~\cite{Sudborough78}. 
To compare with other classes of logspace-uniform bounded-depth Boolean circuits, it is known that:
\[\NL \subseteq \SAC^1 = \LOGCFL \subseteq \AC^1 \subseteq \NC^2\] 
Here, \NL{} is a logspace version of \NP{} with polynomial-length witness, $\AC^1$ captures the power of $O(\log{n})$-depth Boolean circuits using \textit{unbounded} fan-in gates, and $\NC^2$ characterizes the power of $O(\log^2{n})$-depth Boolean circuits with \textit{bounded} fan-in gates. 

Similar to \NL{}, \LOGCFL{} (equivalently, $\SAC^1$) is closed under complementation~\cite{BCD+89}. However, whether $\SAC^1$ is contained in \BPL{} or \BQL{} remains an open problem. 

\section{Space-bounded (unitary) quantum interactive proofs}
\label{sec:poly-message-QIPL}

In this section, we introduce space-bounded quantum interactive proofs (\QIPL{}), where the verifier's actions are implemented using space-bounded \textit{almost-unitary} quantum circuits (see \Cref{def:space-bounded-quantum-circuits} and \Cref{remark:subtleties-space-bounded-circuits}\ref{remarkitem:pinching-measurement}); along with the variant \QIPUL{}, in which the verifier's actions are restricted to \textit{unitary} circuits. Both \QIPL{} and \QIPUL{} are variants of single-prover quantum interactive proofs (\QIP{})~\cite{Wat99QIP,KW00} that have a space constraint. We establish three theorems concerning the classes \QIPL{} and \QIPUL{}, focusing on space-bounded (unitary) quantum interactive proofs with a \textit{polynomial} number of messages. 

The first theorem shows that \QIPLHC{}, a subclass of \QIPL{} defined by a high-concentration condition on \textit{yes} instances, provides a new exact characterization of \NP{}, as stated in \Cref{thm:QIPL-eq-NP}. 
This result can be seen as a quantum analog of classical works~\cite{Lipton90,CL95}. 

\begin{theorem}[The equivalence of \QIPLHC{} and \NP]
    \label{thm:QIPL-eq-NP}
    The following holds: 
    \begin{enumerate}[label={\upshape(\arabic*)}, topsep=0.33em, itemsep=0.33em, parsep=0.33em]
        \item For any logspace-computable function $m(n)$ such that $1 \leq m(n) \leq \poly(n)$,  
        \[\cup_{c(n)-s(n) \geq 1/\poly(n)}\QIPL^\HC_m[c,s] \subseteq \NP.\]
        \label{thmitem:QIPLHC-in-NP}
        \item $\NP \subseteq \QIPL^\HC_m \subseteq \QIPL_m$, where $m(n)$ is some polynomial in $n$.
    \end{enumerate}
\end{theorem}

It is noteworthy that \Cref{thm:QIPL-eq-NP}\ref{thmitem:QIPLHC-in-NP} can be slightly strengthened to $\QIPL^{\HC(\poly)} \subseteq \NP$, where $\QIPL^{\HC(\poly)}$ denotes a variant of \QIPLHC{} that relaxes the high-concentration condition to \textit{polynomial}-size supports (see \Cref{remark:checking-consistency}).
The second theorem addresses two fundamental but crucial properties of \QIPL{} and \QIPUL{}. Specifically, closure under perfect completeness (\Cref{lemma:QIPL-perfect-completeness}) and error reduction through sequential repetition (\Cref{lemma:QIPL-error-reduction}): 

\begin{theorem}[Basic properties for \QIPL{} and \QIPUL{}]
    \label{thm:QIPL-basic-properties}
    Let $c(n)$, $s(n)$, and $m(n)$ be logspace-computable functions such that $0 \leq s(n) < c(n) \leq 1$, $c(n) - s(n) \geq 1/\poly(n)$, and $1 \leq m(n) \leq \poly(n)$. Then the following properties hold:
    \begin{enumerate}[label={\upshape(\arabic*)}, topsep=0.33em, itemsep=0.33em, parsep=0.33em]
        \item \textbf{\emph{Closure under perfect completeness}}.         \label{thmitem:QIPL-perfect-completeness}
        \[ \QIPL_m[c,s] \subseteq \QIPL_{m+2}[1,1-(c-s)^2/2] \text{ and } \QIPUL_m[c,s] \subseteq \QIPUL_{m+2}[1,1-(c-s)^2/2]. \]
        \item \textbf{\emph{Error reduction}}. For any polynomial $k(n)$, 
        \[ \QIPL_m[c,s] \subseteq \QIPL_{m'}[1,2^{-k}] \text{ and } \QIPUL_m[c,s] \subseteq \QIPUL_{m'}[1,2^{-k}]. \]
        Here, $m'$ is some polynomial in $n$. \label{thmitem:QIPL-error-reduction}
    \end{enumerate}
\end{theorem}

The third theorem provides a lower bound for \QIPUL{}, which serves as a quantum analog of the space-bounded public-coin classical interactive proof for $\SAC^1$ established in~\cite{Fortnow89}: 
\begin{theorem}
    \label{thm:LOGCFL-in-QIPUL}
    $\SAC^1 \cup \BQL \subseteq \QIPUL_m$, where $m(n)$ is some polynomial in $n$.
\end{theorem}

\vspace{1em}
In the remainder of this section, we provide the definitions of space-bounded (unitary) quantum interactive proofs, specifically the classes \QIPL{} and \QIPUL{}, in \Cref{subsec:QIPL-definition}. 
We then formulate \QIPL{} proof systems as semi-definite programs in \Cref{subsec:SDP-formulations-upper-bounds}, leading to the inclusion $\QIPL^\HC \subseteq \NP$ (\Cref{thm:QIPLHC-in-NP}). 
Next, the proof of the two basic properties in \Cref{thm:QIPL-basic-properties} is presented in \Cref{subsec:QIPL-basic-properties}. 
Lastly, the lower bounds for \QIPL{} and \QIPUL{}, particularly the inclusions $\NP \subseteq \QIPLHC \subseteq \QIPL$ (\Cref{thm:NP-in-QIPL}) and $\SAC^1 \subseteq \QIPUL{}$ (\Cref{thm:LOGCFL-in-QIPUL}), are established in \Cref{subsec:lower-bounds-QIPL}. 

\subsection{Definitions of space-bounded quantum interactive proof systems}
\label{subsec:QIPL-definition}

Our definitions of space-bounded quantum interactive proofs follow that of~\cite[Section 2.3]{KW00} and~\cite[Section 2.3]{Wat02QSZK}. In this framework, a (log)space-bounded quantum interactive proof system consists of two parties: an untrusted prover with unbounded computational power, and a verifier constrained to using only \textit{$O(\log n)$ qubits}, enabling at most polynomial-time quantum computation. 
The primary distinction between standard single-prover quantum interactive proofs and their space-bounded variants lies in this additional space constraint on the verifier, which prompts a subtle question: 
\begin{problem}
    \label{prob:subtlety-in-intermediate-measurements}
    Is it necessary to allow $O(\log n)$ \emph{pinching} intermediate measurements in the computational basis during \emph{each} verifier's action in space-bounded quantum interactive proofs?
\end{problem}
Interestingly, the answer to \Cref{prob:subtlety-in-intermediate-measurements} does not affect one-message proof systems, specifically (unitary) \QMAL{}~\cite{FKLMN16,FR21}, where the unitary verification circuit acts on $O(\log n)$ qubits and inherently allows not only $O(\log n)$ pinching intermediate measurements but also $O(\log n)$ ancillary gates (see also \Cref{remark:subtleties-space-bounded-circuits}). 

To address \Cref{prob:subtlety-in-intermediate-measurements}, we introduce two types of space-bounded quantum interactive proof systems, denoted by \QIPL{} and \QIPUL{}. 
In both proof systems, the verifier has \textit{direct access} to the messages exchanged during interactions, which limits each message size to $O(\log n)$. 
The key distinction between these proof systems lies in their differing responses to \Cref{prob:subtlety-in-intermediate-measurements}: In \QIPL{}, the verifier's actions correspond to space-bounded \textit{almost-unitary} quantum circuits (see \Cref{def:space-bounded-quantum-circuits}), allowing $O(\log{n})$ pinching intermediate measurements in each verifier's action. 
However, in \QIPUL{}, the verifier's actions are implemented using space-bounded \textit{unitary} quantum circuits. 
Consequently, the former model encompasses space-bounded (private-coin) classical interactive proof systems, particularly the model described in~\cite{CL95}.\footnote{For any proof system that ensures soundness against classical messages, we can construct a corresponding proof system that guarantees soundness against quantum messages by measuring the message in the computational basis at the beginning of each verifier's action.}

\begin{remark}[Restricting the number of pinching measurements preserves generality]
    \label{remark:restrict-measurement-times}
    We note that restricting the number of pinching measurements in each verifier turn does not result in any loss of generality.\footnote{This argument does not apply to \QIPLconst{}, a subclass of \QIPL{} with a constant number of messages (see \Cref{sec:cosnt-message-QIPL} for a formal definition). Moreover, the counterpart question is a special case of Question \ref{probitem:constTurnQIPL}.} 
    Let $\QIPL^\star$ be a variant of \QIPL{} in which the verifier is allowed to perform at most a \textit{polynomial} number of pinching measurements per turn. Any $\QIPL^\star$ proof system $\protocol{P}{V}$ can be simulated by a \QIPL{} proof system $\protocol{P'}{V'}$ that includes additional \textit{dummy} rounds. Specifically, each verifier turn in $\protocol{P}{V}$ that involves $t(n)$ pinching measurements, where $t(n)$ is polynomial in $n$, can be simulated in $\protocol{P'}{V'}$ by introducing $\ceil{t(n)/\log(n)}$ extra rounds. During each of these dummy rounds, the prover $P'$ is restricted to performing only dummy actions, such as doing nothing (a hypothetical action) or sending an all-zero state in the message register. 
\end{remark}

\paragraph{Formal definitions of \QIPL{} and \QIPUL{}.} Given a promise problem $\calI = (\calI_{\yes}, \calI_{\no})$, a quantum verifier is \textit{logspace-computable} mapping $V$, where for each input string $x \in \calI \subseteq \binset^*$, $V(x)$ is interpreted as an encoding of a $k(|x|)$-tuple $\rbra*{V(x)_1,\cdots,V(x)_k}$ of quantum circuits. These circuits represent the verifier's actions at each round of the proof system, with specific constraints depending on the proof system, \QIPL{} or \QIPUL{}: 
\begin{itemize}[leftmargin=2em, topsep=0.33em, itemsep=0.33em, parsep=0.33em]
    \item In a \QIPL{} proof system, each $V(x)_j$ corresponds to a space-bounded \textit{almost-unitary} quantum circuit $\widetilde{V}(x)_j$. By applying the principle of deferred measurements, $\widetilde{V}(x)_j$ can be transformed into an isometry $V(x)_j$ (see \Cref{def:space-bounded-quantum-circuits}) that takes qubits in registers $(\sfM,\sfW)$ as input and outputs qubits in registers $(\sfM,\sfW,\sfE_j)$, where $\sfE_j$ holds $q_{\sfE_j}(|x|)$ qubits.
    At the end of the verifier's $j$-th action $V(x)_j$, the (newly introduced) environment register $\sfE_j$ is measured in the computational basis, with the measurement outcome denoted as $u_j$. 
    The total number of qubits satisfies $q_{\sfM}(|x|) + q_{\sfW}(|x|) + q_{\sfE_j}(|x|) \leq O(\log n)$, with both $\sfW$ and $\sfE_j$ private to the verifier. 
    \item In a \QIPUL{} proof system, each $V(x)_j$ is a space-bounded \textit{unitary} quantum circuit acting on two registers $\sfM$ and $\sfW$, which hold $q_{\sfM}(|x|)$ and $q_{\sfW}(|x|)$ qubits, respectively. The total number of qubits satisfies $q_{\sfM}(|x|) + q_{\sfW}(|x|) \leq O(\log n)$, with $\sfW$ private to the verifier. 
\end{itemize}

Furthermore, the logspace-computability of $V(x)$ requires \textit{a strong notion of uniformity}: there must exist a logspace deterministic Turing machine $\calM$ that, for each input $x$, outputs the classical description of $\rbra{V(x)_1,\cdots,V(x)_k}$.\footnote{This uniformity requirement is slightly stronger and less general than merely requiring all quantum circuits $V(x)_1,\cdots,V(x)_k$ to be logspace-bounded (referred to as \textit{a weaker notion of uniformity}), as the classical descriptions of these quantum circuits may not be generated by a single logspace deterministic Turing machine (although a polynomial-time deterministic Turing machine would suffice).\label{footref:logspace-uniformity}} 
Lastly, the verifier $V$ is called $m(|x|)$-message if $k(|x|) = \floor{m(|x|)/2+1}$ for all integer $|x|$, depending on whether $m$ is even or odd. 

Similar to standard quantum interactive proofs, the prover and the verifier in the same space-bounded quantum interactive proof system must be \textit{compatible}. This means that they must agree on the maximum length $q_{\sfM}(|x|)$ of each message exchanged in the proof system and the total number $m(|x|)$ of these messages. Hence, a quantum prover $P$ is a function that maps each input $x\in \calI$ to an $l(|x|)$-tuple $\rbra*{P(x)_1,\cdots,P(x)_l}$ of quantum circuits, where $l(|x|) = \floor*{(m(|x|)+1)/2}$. Each circuit $P(x)_j$ acts on two registers $\sfQ$ and $\sfM$ with $q_{\sfQ}(|x|)$ and $q_{\sfM}(|x|)$ qubits, respectively, satisfying that $\sfQ$ is private to the prover. Since there are no restrictions on the prover $P$, each $P(x)_j$ can be viewed as an arbitrary unitary transformation in general. 

\begin{figure}[ht!]
    \centering
    \includegraphics[width=\textwidth]{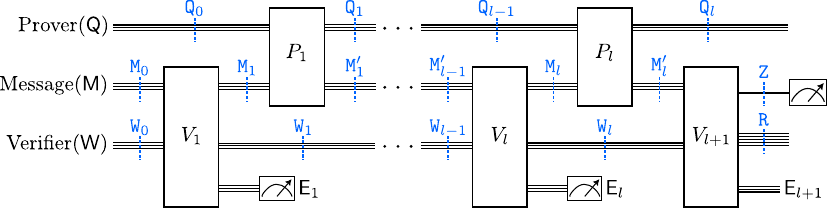}
    \caption{A $2l$-turn space-bounded quantum interactive proof system (with snapshots).}
    \label{fig:QIPL-even}
\end{figure}

Given an input $x \in \calI$, and a prover $P$ and a verifier $V$ that exchange $m(|x|)$ messages, we define an $m(|x|)$-turn space-bounded quantum interactive proof system $(\protocol{P}{V})(x)$, namely a \QIPL{} proof system, as a quantum circuit acting on the registers $\sfQ$, $\sfM$, $\sfW$, and additional environment registers $\{\sfE_j\}$ as follows: 
\begin{itemize}
    \item If $m(|x|) = 2 l(|x|)$ is even, circuits $V(x)_1, P(x)_1, \cdots, V(x)_l, P(x)_l, V(x)_{l+1}$ are applied in sequence to the registers $\sfM$, $\sfW$, and $\sfE_j$, or to the registers $\sfQ$ and $\sfM$ accordingly. It is important to note that the register $\sfE_j$ is inaccessible after the $j$-th round. 
    \item If $m(|x|) = 2 l(|x|) + 1$ is odd, the situation is similar, except that the prover starts the protocol, so the circuits $P(x)_1,V(x)_1,\cdots,P(x)_{l+1},V(x)_{l+1}$ are applied in sequence. 
\end{itemize}

Analogously, for an $m(|x|)$-turn space-bounded \textit{unitary} quantum interactive proof system, namely a \QIPUL{} proof system, the definition remains the same except that environment registers $\{\sfE_j\}$ are no longer involved. 
Without the loss of generality, we assume that the prover always sends the last message. 
See also \Cref{fig:QIPL-even} for an illumination of the case when $m(|x|)$ is even. 
For convenience, we sometimes omit the dependence on $x$ and $|x|$ when describing $P$ and $V$, e.g., using $P_j$ and $V_j$ to denote $P(x)_j$ and $V(x)_j$, respectively, and $m$ to denote $m(|x|)$. 

Assuming the mapping $V(x) = \rbra*{V(x)_1, \cdots, V(x)_k}$ in a \QIPL{} proof system is a collection of the isometries corresponding to the respective almost-unitary quantum circuits,\footnote{This assumption about the verifier's actions is crucial for adapting several techniques from standard quantum interactive proofs. For more on general verifiers in the standard scenario, see~\cite[Section 4.1.4]{VW16}.} the state of the qubits in the circuit $\protocol{P}{V}$ is an (unnormalized) \textit{pure state} on the registers $(\sfQ, \sfM, \sfV, \sfE_j)$ after the verifier's $j$-th action.\footnote{Specifically, a pure state occupies the registers $(\sfQ, \sfM, \sfV, \sfE_1)$ before the measurements at the end of the verifier's first action. To align with the proof of \Cref{lemma:QIPL-second-SDP-formulation} (the upper bound for \QIPLHC{}), which considers only a specific measurement outcome, this post-measurement pure state will be unnormalized.} 
A similar observation holds for \QIPUL{} proof systems. 
Thus, for a given input $x$, the probability that $\protocol{P}{V}$ accepts $x$ is defined as the probability that measuring the designated output qubit -- typically the first qubit of $(\sfM,\sfW)$ -- of $(\protocol{P}{V})(x) \ket{\bar{0}}_{\sfQ} \ket{\bar{0}}_{\sfM} \ket{\bar{0}}_{\sfW}$ in the computational basis yields the outcome $1$. 
We are now ready to formally define space-bounded quantum interactive proof systems: 

\begin{definition} [Space-bounded quantum interactive proofs, $\QIPL$]
    \label{def:QIPL}
    Let $c(n)$, $s(n)$, and $m(n)$ be logspace-computable functions of the input length $n \coloneqq |x|$ such that $0 \leq s(n) < c(n) \leq 1$ and $1 \leq m(n) \leq \poly(n)$.
    A promise problem $\calI = (\calI_{\yes}, \calI_{\no})$ is in $\QIPL_m[c,s]$, if there exists an $m(n)$-turn logspace-computable \emph{almost-unitary} quantum verifier $V$ such that\emph{:} 
    \begin{itemize}[topsep=0.33em, itemsep=0.33em, parsep=0.33em]
        \item \textbf{\emph{Completeness}}. For any $x \in \calI_{\yes}$, there exists an $m(n)$-message prover $P$ such that 
        \[\Pr{\rbra*{\protocol{P}{V}}(x)\text{ accepts}} \geq c(n).\]
        \item \textbf{\emph{Soundness}}. For any $x \in \calI_{\no}$ and any $m(n)$-message prover $P$,  
        \[\Pr{\rbra*{\protocol{P}{V}}(x)\text{ accepts}} \leq s(n).\]
    \end{itemize}

    \noindent Furthermore, we define $\QIPL_m \coloneqq \QIPL_m[2/3,1/3]$ and $\QIPL \coloneqq \cup_{m \leq \poly(n)} \QIPL_m$. 
\end{definition}

Next, we proceed with a formal definition of \QIPLHC{}. Let $\omega(V)$ denote the maximum acceptance probability of the verifier $V$ in the proof system $\protocol{P}{V}$. 
For \QIPLHC{} proof systems, we impose an additional restriction on the distribution of pinching intermediate measurement outcomes $u \coloneqq (u_1,\cdots,u_l)$, conditioned on acceptance, for \textit{yes} instances. To formalize this restriction, we define $\omega(V)|^u$ as the contribution of $u$ to $\omega(V)$, where all post-measurement states remain \textit{unnormalized}. A direct calculation then implies that
\begin{equation}
    \label{eq:pacc-decomposition}
    \omega(V) = \sum\nolimits_{u \in \binset^{q_{\sfE_1}+\cdots+q_{\sfE_l}}} \omega(V)|^u.
\end{equation} 

While the soundness condition of \QIPLHC{} coincides with that of \QIPL{}, the completeness condition is more stringent, which immediately implies that $\QIPLHC \subseteq \QIPL$: 

\begin{definition} [Space-bounded quantum interactive proofs with the high-concentration condition, \QIPLHC{}]
    \label{def:QIPLHC}
    Let $c(n)$, $s(n)$, and $m(n)$ be logspace-computable functions of the input length $n \coloneqq |x|$ such that $0 \leq s(n) < c(n) \leq 1$ and $1 \leq m(n) \leq \poly(n)$.
    A promise problem $\calI = (\calI_{\yes}, \calI_{\no})$ is in $\QIPLHC_m[c,s]$, if there exists an $m(n)$-turn logspace-computable \emph{almost-unitary} quantum verifier $V$ such that\emph{:} 
    \begin{itemize}[topsep=0.33em, itemsep=0.33em, parsep=0.33em]
        \item \textbf{\emph{Completeness}}. For any $x \in \calI_{\yes}$, there exists an $m(n)$-message prover $P$ such that there exists an intermediate measurement outcome $u^*=(u^*_1,\cdots,u^*_l)$ with $\omega(V)|^{u^*} \geq c(n)$. 
        \item \textbf{\emph{Soundness}}. For any $x \in \calI_{\no}$ and any $m(n)$-message prover $P$, $\omega(V) \leq s(n)$.
    \end{itemize}

    \noindent Furthermore, we define $\QIPL^\HC_m \coloneqq \QIPL^\HC_m[2/3,1/3]$ and $\QIPL^\HC \coloneqq \cup_{m \leq \poly(n)} \QIPL^\HC_m$. 
\end{definition}

In addition to \QIPL{} and \QIPLHC{}, we also require a \textit{reversible} generalization of \QIPL{} in which the verifier’s actions are given by \textit{isometries}:  

\begin{remark}[A reversible generalization of \QIPL{}]
    \label{remark:reversible-QIPL}
    We introduce the class $\QIPL^{\diamond}$ as a \textit{reversible} generalization of \QIPL{}, primarily for convenience. A $\QIPL^{\diamond}$ proof system is defined similarly to a \QIPL{} proof system, but with two crucial differences: 
    \begin{enumerate}[label={\upshape(\arabic*)}, topsep=0.33em, itemsep=0.33em, parsep=0.33em]
        \item All of the verifier's actions are \textit{isometric} quantum circuits, without restrictions. In particular, any unitary elementary gate can act on the ancillary qubits $\ket{\bar{0}}$,  which are introduced by $O(\log{n})$ ancillary gates, and qubits in the message register $\sfM$.
        \item The environment register $\sfE_k$, introduced during the verifier's $k$-th action, remains private to the verifier and is \textit{accessible only during that turn}. Importantly, the qubits in $\sfE_k$ are not measured at the end of the turn. Consequently, the qubits in $\sfE_k$ remain unchanged after that turn, although the entanglement shared among $\sfE_1$, $\cdots$, $\sfE_k$, and $\sfM$ may change. 
    \end{enumerate}
\end{remark}

\vspace{1em}
Finally, in analogy to \QIPL{}, we define the class of space-bounded \textit{unitary} quantum interactive proof systems (\QIPUL{}), which is naturally reversible and contained in both \QIPL{} and \QIPLHC{}:

\begin{definition} [Space-bounded unitary quantum interactive proofs, \QIPUL{}]
    \label{def:QIPUL}
    Let $c(n)$, $s(n)$, and $m(n)$ be logspace-computable functions of the input length $n \coloneqq |x|$ such that $0 \leq s(n) < c(n) \leq 1$ and $1 \leq m(n) \leq \poly(n)$.
    A promise problem $\calI = (\calI_{\yes}, \calI_{\no})$ is in $\QIPUL_m[c,s]$, if there exists an $m(n)$-turn logspace-computable \emph{unitary} quantum verifier $V$ such that\emph{:} 
    \begin{itemize}[topsep=0.33em, itemsep=0.33em, parsep=0.33em]
        \item \textbf{\emph{Completeness}}. For any $x \in \calI_{\yes}$, there exists an $m(n)$-message prover $P$ such that  
        \[\Pr{\rbra*{\protocol{P}{V}}(x) \text{ accepts}} \geq c(n).\]
        \item \textbf{\emph{Soundness}}. For any $x \in \calI_{\no}$ and any $m(n)$-message prover $P$, 
        \[\Pr{\rbra*{\protocol{P}{V}}(x) \text{ accepts}} \leq s(n).\] 
    \end{itemize}

    \noindent Furthermore, we define $\QIPUL_m \coloneqq \QIPUL_m[2/3,1/3]$ and $\QIPUL \coloneqq \cup_{m \leq \poly(n)} \QIPUL_m$. 
\end{definition}

\subsection{An upper bound for \QIPLHC{} via SDP formulations}
\label{subsec:SDP-formulations-upper-bounds}

We begin with the upper bound for the class \QIPLHC{}:\footnote{It is noteworthy that this upper bound also applies to a variant of \QIPLHC{} in which the verifier's mapping satisfies a weaker notion of uniformity (see \Cref{footref:logspace-uniformity} for details).}

\begin{theorem}[\QIPLHC{} is in \NP{}]
    \label{thm:QIPLHC-in-NP}
    Let $c(n)$, $s(n)$, and $m(n)$ be logspace-computable functions such that $0 \leq s(n) < c(n) \leq 1$, $c(n) - s(n) \geq 1/\poly(n)$, and $1 \leq m(n) \leq \poly(n)$, it holds that 
    \[\QIPL^\HC_m[c,s] \subseteq \NP.\]
\end{theorem}

\begin{remark}[Upper bounding \QIPL{} proof systems beyond \QIPLHC{}]
    \label{remark:checking-consistency}
    To establish upper bounds for \QIPL{} proof systems beyond \QIPLHC{}, it is necessary to ensure the \textit{consistency} of the prover strategies across different branches defined by measurement outcomes. In particular, the prover strategies associated with different branches -- those that satisfy the conditions of \Cref{lemma:QIPL-second-SDP-formulation} -- are expected to be consistent with a ``purified'' prover strategy satisfying \Cref{lemma:QIPL-first-SDP-formulation}. 
    
    A natural approach to checking this consistency is to verify that, for any two branches $u$ and $u'$ in an $m$-turn \QIPL{} proof system, if the verifier's measurement outcomes are identical in the first $i$ turns (for any $i \leq m$), then the corresponding prover strategies for these turns must also be identical. Consequently, two implications follow: 
    (a) The set lower bound protocol~\cite{GS86} for establishing \AM{} containment is not directly applicable to \QIPL{};\footnote{This explains why the inclusion $\QIPL \subseteq \AM$ in an earlier version of this work was withdrawn.
    More precisely, let $S \subseteq \binset^k$ be a set, where membership in $S$ can be efficiently verified using a witness string $w$. The set lower bound protocol~\cite{GS86} implies that certifying the size of $S$ up to a constant multiplicative error is in \AM{} (see also~\cite{Volkovich20}). However, when such a set $S$ is constructed from \Cref{lemma:QIPL-second-SDP-formulation}, the witnesses used to verify membership are correlated with the elements of $S$ in general, unlike the setting considered in~\cite{GS86,Volkovich20}.}
    and (b) It is not difficult to show that $\QIPL^{\HC(\poly)} \subseteq \NP$, where $\QIPL^{\HC(\poly)}$ denotes a variant of \QIPLHC{} defined with a high-concentration condition on a \textit{polynomial}-size support: $\sum_{u \in \calJ} \omega(V)|^u \geq c(n)$, where $\calJ$ is an index set of size \textit{polynomial} in $n$. 
\end{remark}

\vspace{1em}
Before presenting the proof, we introduce the term \textit{snapshot registers} to refer to the registers $\sfQ$, $\sfM$, and $\sfW$ after each turn in a \QIPL{} proof system. We also refer to the quantum state within these snapshot registers as \textit{snapshot states}. For example, the snapshot registers corresponding to $\sfW$ at distinct time points are $\ttW_0,\cdots,\ttW_l$, as illustrated in \Cref{fig:QIPL-even}. 
More precisely, for an $l$-round (i.e., $2l$-turn) \QIPL{} proof system, we define the following: 
\begin{enumerate}[label={\upshape(\arabic*)}, topsep=0.33em, itemsep=0.33em, parsep=0.33em]
    \item $\ttQ_0$, $\ttM_0$, and $\ttW_0$, which contain the all-zero state, are the snapshot registers of registers $\sfQ$, $\sfM$, and $\sfW$, respectively, before the protocol begins;
    \item $\ttM_j$ and $\ttW_j$ ($1 \leq j \leq l$) are the snapshot registers of the registers $\sfM$ and $\sfW$, respectively, after the verifier sends the message in the $j$-th round;
    \item $\ttM'_j$ and $\ttW_j$ ($1 \leq j \leq l$) are the snapshot registers of the registers $\sfM$ and $\sfW$, respectively, after the verifier receives the message from the prover in the $j$-th round; 
    \item $\ttQ_j$ ($1 \leq j \leq l$) is the snapshot registers of the register $\sfQ$ immediately after applying the prover's $j$-th action $P_j$, i.e., after the prover sends the message in the $j$-th round; 
    \item $(\ttZ,\ttR)$ represents the snapshot registers of registers $\sfM$ and $\sfW$, respectively, just before the verifier performs the final measurement, where $\ttZ$ corresponds to the designated output qubit of the verifier and $\ttR$ contains the remained qubits. 
\end{enumerate}

In the remainder of this subsection, we first present semi-definite programming (SDP) formulations for \QIPL{} proof systems in \Cref{subsubsec:SDP-formulation}, and then establish an upper bound for \QIPLHC{} (\Cref{thm:QIPLHC-in-NP}) in \Cref{subsubsec:QIPL-upper-bounds}. 

\subsubsection{Semi-definite program formulations for \QIPL{} proof systems}
\label{subsubsec:SDP-formulation}

To establish upper bounds for space-bounded (unitary) quantum interactive proofs, a commonplace approach involves solving the optimization problem of approximating the maximum acceptance probability of a \QIPL{} or \QIPUL{} proof system over all prover strategies. 
For clarity, we first present an SDP for characterizing \QIPL{} proof systems (\Cref{lemma:QIPL-first-SDP-formulation}), directly extended from~\cite[Section 4.3]{VW16} and~\cite[Section 4]{Watrous16tutorial}. Next, we introduce another SDP characterization (\Cref{lemma:QIPL-second-SDP-formulation}) that captures all the structural constraints of a \QIPL{} proof system within any given branch defined by pinching measurement outcomes. 

\paragraph{First SDP formulation for \QIPL{} proof systems.} For an $m$-turn \QIPL{} proof system with even $m$,\footnote{Adapting the proof to the case of odd $m$ is straightforward, so we omit the details.} we formulate this optimization problem as a semi-definite program (SDP) from the verifier's perspective, following the approach described in~\cite[Section 4.3]{VW16}: 

\begin{lemma}[First SDP formulation for \QIPL{} proof systems]
    \label{lemma:QIPL-first-SDP-formulation}
    For any $l(n)$-round space-bounded quantum interactive proof system $\protocol{P}{V}$ with completeness $c(n)$, soundness $s(n)$, which corresponds to a promise problem $\calI=(\calI_{\yes},\calI_{\no})$ in $\QIPL_{2l}[c,s]$, there is an SDP program to compute the maximum acceptance probability $\omega(V)$ of the proof system $\protocol{P}{V}$\emph{:}
    \begin{maxi}{}{\omega(V) = \Tr\rbra*{\widetilde{V}^{\dagger}_{l+1} \ketbra{1}{1}_{\Out} \widetilde{V}_{l+1} \rho_{\ttM'_l\ttW_l\sfE_1\cdots\sfE_l}}}{}{}{\label{eq:QIPL-first-SDP}}{}
        \addConstraint{\Tr_{\ttM'_1}\big(\rho_{\ttM'_1\ttW_1\sfE_1}\big)}{=\Tr_{\ttM_1}\rbra*{V_1 \ketbra{\bar{0}}{\bar{0}}_{\ttM_{0}\ttW_{0}}V_1^{\dagger}}}{}
        \addConstraint{\Tr_{\ttM'_j}\big(\rho_{\ttM'_j\ttW_j\sfE_1\cdots\sfE_j}\big)}{=\Tr_{\ttM_j}\rbra*{\widetilde{V}_j \rho_{\ttM'_{j-1}\ttW_{j-1}\sfE_1\cdots\sfE_{j-1}} \widetilde{V}_j^{\dagger} },\quad}{j \in [l] \setminus [1]}
        \addConstraint{ \Tr\big(\rho_{\ttM'_j\ttW_j\sfE_1\cdots\sfE_j}\big)}{= 1,\quad}{j \in [l]}
        \addConstraint{ \rho_{\ttM'_j\ttW_j\sfE_1\cdots\sfE_j}}{ \succeq 0,\quad}{j \in [l]}
    \end{maxi}
    Here, the verifier's actions $V_1, \cdots, V_{l+1}$ are considered space-bounded \emph{isometric} quantum circuits, with the notation $\widetilde{V}_j \coloneqq V_j \!\otimes\! I_{\sfE_1\cdots\sfE_{j-1}}$ for each $j \in [l+1] \setminus [1]$. 
    The variables in this SDP are $\rho_{\ttM'_1\ttW_1\sfE_1}, \cdots, \rho_{\ttM'_l\ttW_l\sfE_1\cdots\sfE_l}$, collectively holding $O(l^2(n) \cdot \log{n})$ qubits. 
\end{lemma}

\begin{remark}[The applicability of the first SDP formulation for \QIPL{}]
    \label{remark:QIPL-first-SDP-applicability}
    The SDP program in \Cref{eq:QIPL-first-SDP} essentially characterizes $\QIPL^{\diamond}$ proof systems, which are informally defined in \Cref{remark:reversible-QIPL} as a \textit{reversible} generalization of \QIPL{}, where the verifier's actions are \textit{isometric} quantum circuits.
    Additionally, by disregarding all environment registers $\{E_j\}$ in \Cref{eq:QIPL-first-SDP}, we immediately obtain an SDP formulation for \QIPUL{} proof systems, where the variables are $\rho_{\ttM'_1\ttW_1}, \cdots, \rho_{\ttM'_l\ttW_l}$, collectively holding $O(l(n) \cdot \log{n})$ qubits. 
\end{remark}

Our SDP program in \Cref{eq:QIPL-first-SDP} consists of two types of constraints, both of which can be described by simple equations: (1) the verifier remains honest, and (2) the prover's actions do not interfere with the verifier's private qubits. Importantly, every feasible solution to \Cref{eq:QIPL-first-SDP} corresponds to a valid strategy for the prover. 
We now proceed with the detailed proof: 

\begin{proof}[Proof of \Cref{lemma:QIPL-first-SDP-formulation}]
    For any $m(n)$-turn proof system $\protocol{P}{V}$, with completeness $c$, soundness $s$, and $m$ being even, which corresponds to a promise problem $\calI$ in $\QIPL_m[c,s]$, we consider the verifier's maximum acceptance probability $\omega(V)$ as the objective function to be maximized. 
    In our SDP formulation for the \QIPL{} proof system $\protocol{P}{V}$, we focus on the verifier's actions (e.g., $V_j$), specifically isometric quantum circuits that do not measure the new environment register (e.g., $\sfE_j$) at the end. As defined in \Cref{subsec:QIPL-definition}, the verifier $V$ is described by an $(l+1)$-tuple $(V_1,\cdots,V_{l+1})$ of space-bounded unitary quantum circuits $\{V_j\}_{j \in [l+1]}$, where $l = m/2$.   

    To represent the variables in our SDP program, which are the states in the snapshot registers corresponding to the message register $\sfM$, the verifier's private register $\sfW$, and the environment registers $\sfE_1\cdots,\sfE_j$ after the verifier's $j$-th action in $\protocol{P}{V}$, we use the notations defined in \Cref{fig:QIPL-even}. Specifically, let $\rho_{\ttM_j\ttW_j\sfE_1\cdots\sfE_j}$ denote the state in the snapshot registers $(\ttM_j, \ttW_j, \sfE_1, \cdots, \sfE_j)$. Similarly, we can define snapshot states $\rho_{\ttM'_j\ttW_j\sfE_1\cdots\sfE_j}$ for $j\in[l]$ and $\rho_{\ttZ\ttR\sfE_1\cdots\sfE_{l+1}}$ accordingly. 

    \vspace{1em}
    Assuming that $\protocol{P}{V}$ begins with the verifier, it follows that the objective function
    \begin{equation}
        \label{eq:QIPUL-objective-function}
        \begin{aligned}
        \omega(V) &= \| \ketbra{1}{1}_{\Out} V_{l+1} P_l V_l \cdots P_1 V_1 \ket{\bar{0}}_{\ttQ_0,\ttM_0,\ttW_0}\|_2^2\\
        &= \Tr\rbra*{ \Tr_{\sfQ}\rbra*{ \rbra*{V_{l+1} P_l V_l \cdots P_1 V_1}^{\dagger} \ketbra{1}{1}_{\Out} \rbra*{V_{l+1} P_l V_l \cdots P_1 V_1} \ketbra{\bar{0}}{\bar{0}}_{\ttQ_0,\ttM_0,\ttW_0} } }\\
        &= \Tr\rbra*{ \big(V^{\dagger}_{l+1} \!\otimes\! I_{\sfE_1\cdots\sfE_{l}}\big) \ketbra{1}{1}_{\Out} \big(V^{\dagger}_{l+1} \!\otimes\! I_{\sfE_1\cdots\sfE_{l}}\big)  \rho_{\ttM'_l\ttW_l\sfE_1\cdots\sfE_l} }.
        \end{aligned}
    \end{equation}
    
    Noting that the verifier $V$ remains honest and the verifier's $j$-th action does not act on the environment registers $\sfE_1,\cdots,\sfE_{j-1}$, we obtain the first type of constraints:
    \begin{equation}
        \label{eq:QIPUL-honest-verifier-cond}
        \begin{aligned}
            \rho_{\ttM_1\ttW_1\sfE_1} &= V_1 \rho_{\ttM_{0}\ttW_{0}} V_1^{\dagger} = V_1 \ketbra{\bar{0}}{\bar{0}}_{\ttM_0\ttW_0} V_1^{\dagger};\\
        \forall j \in \{2,\cdots,l\},~\rho_{\ttM_j\ttW_j\sfE_1\cdots\sfE_j} &= \big(V_j \!\otimes\! I_{\sfE_1\cdots\sfE_{j-1}} \big) \rho_{\ttM'_{j-1}\ttW_{j-1}\sfE_1\cdots\sfE_{j-1}} \big( V_j^{\dagger} \!\otimes\! I_{\sfE_1\cdots\sfE_{j-1}}  \big);\\
        \rho_{\ttZ\ttR\sfE_1\cdots\sfE_{l+1}} &= \big( V_{l+1} \!\otimes\! I_{\sfE_1\cdots\sfE_l} \big) \rho_{\ttM'_l\ttW_l\sfE_1\cdots\sfE_l} \big(V_{l+1}^{\dagger} \!\otimes\! I_{\sfE_1\cdots\sfE_l} \big).
        \end{aligned}
    \end{equation}

    Since the prover's actions are described by unitary quantum circuits and the verifier's actions by isometric quantum circuits, all intermediate states in $(\sfQ, \sfM, \sfW, \sfE_1, \cdots, \sfE_j)$ after the verifier's $j$-th action in $\protocol{P}{V}$ are pure states. These states, denoted by $\ket{\psi}_{\ttQ_{j-1}\ttM_j\ttW_j\sfE_1\cdots\sfE_j}$ and $\ket{\phi}_{\ttQ_{j}\ttM'_j\ttW_j\sfE_1\cdots\sfE_j}$ for $j \in [l]$, satisfy the relation
    $\ket{\phi}_{\ttQ_{j}\ttM'_j\ttW_j\sfE_1\cdots\sfE_j} = P_j \ket{\psi}_{\ttQ_{j-1}\ttM_j\ttW_j\sfE_1\cdots\sfE_j}$.
    By the unitary freedom in purifications (\Cref{lemma:unitary-equivalent-in-purifications}), we have:
    \begin{equation*}
        \forall j \in [l],~\Tr_{\ttQ_{j-1}\ttM_j}\rbra*{\ketbra{\psi}{\psi}_{\ttQ_{j-1}\ttM_j\ttW_j\sfE_1\cdots\sfE_j}} = \rho_{\ttW_j\sfE_1\cdots\sfE_j} = \Tr_{\ttQ_{j}\ttM'_j}\rbra*{\ketbra{\phi}{\phi}_{\ttQ_{j}\ttM'_j\ttW_j\sfE_1\cdots\sfE_j}}.
    \end{equation*}
    Here, the underlying unitary transformation corresponds to the prover's action $P_j$ in the $j$-th round. 
    As a consequence, the prover's actions do not interfere with the verifier's private register $\sfW$ or the environment registers $\sfE_1, \cdots, \sfE_j$ (after the verifier's $j$-th action) during the execution of $\protocol{P}{V}$. 
    This property leads to the second type of constraints: 
    \begin{equation}
        \label{eq:QIPUL-prover-actions-cond}
        \forall j \in [l],~\Tr_{\ttM_j}\rbra*{\rho_{\ttM_j\ttW_j\sfE_1\cdots\sfE_j}} = \Tr_{\ttM'_j}\rbra*{\rho_{\ttM'_j\ttW_j\sfE_1\cdots\sfE_j}}.
    \end{equation}

    Putting \Cref{eq:QIPUL-objective-function,eq:QIPUL-honest-verifier-cond,eq:QIPUL-prover-actions-cond} all together, we conclude our SDP formulation for the given \QIPL{} proof system $\protocol{P}{V}$, as detailed in \Cref{eq:QIPL-first-SDP}.
\end{proof}

\paragraph{Second SDP formulation for \QIPL{} proof systems.} 

The main challenge in fully utilizing all restrictions of \QIPL{} proof systems in an SDP program is effectively using the measurement outcomes from the newly introduced environment register $\sfE_j$ after the verifier's $j$-th action. Specifically, when applying the principle of deferred measurements to the verifier's $j$-th action (almost-unitary quantum circuit), an environment register $\sfE_j$ is introduced. 
Thus, the variables in \Cref{eq:QIPL-first-SDP} correspond to $\rho_{\ttM'_1\ttW_1\sfE_1}$, $\rho_{\ttM'_1\ttW_1\sfE_1\sfE_2}, \cdots$, $\rho_{\ttM'_l\ttW_l\sfE_1\cdots\sfE_l}$. In general, when the verifier's actions are isometric quantum circuits, the environment registers $\sfE_1,\cdots,\sfE_l$ may be \textit{entangled}.\footnote{For instance, the prover may send a highly-entangled $n$-qubit state $\varrho$ in $n/\log{n}$ batches, each containing $\log{n}$ qubits. In this case, the verifier keeps only $O(\log n)$ qubits, which are not necessarily adjacent, and swaps the remaining qubits with fresh qubits in the environment registers.}

However, an almost-unitary quantum circuit (the verifier's $j$-th action) corresponds to a specific type of isometric quantum circuit, followed by measuring the environment register $\sfE_j$ in the computational basis at the end of the circuit. 
Since pinching measurements destroy coherence between subspaces associated with different measurement outcomes (see \Cref{remark:subtleties-space-bounded-circuits}\ref{remarkitem:pinching-measurement}), the environment registers $\sfE_1,\cdots,\sfE_l$ remain \textit{independent} in this restricted setting. 
Consequently, within any given branch $u$ defined by pinching measurement outcomes, we can characterize the prover strategy in the \QIPL{} proof system $\protocol{P}{V}$. This characterization is given by the following SDP program, which computes an \textit{approximation} $\widehat{\omega}(V)|^u$ of $\omega(V)|^u$: 

\begin{lemma}[Second SDP formulation for \QIPL{} proof systems]
    \label{lemma:QIPL-second-SDP-formulation}
    For any $l(n)$-round space-bounded quantum interactive proof system $\protocol{P}{V}$ with completeness $c(n)$ and soundness $s(n)$, corresponding to a promise problem $\calI = (\calI_{\yes},\calI_{\no})$ in $\QIPL_{2l}[c,s]$, there is a family of SDP programs that compute an approximation $\widehat{\omega}(V)|^u$ of the contribution $\omega(V)|^{u}$ associated with the measurement outcome branch $u = (u_1,\cdots,u_l)$, satisfying
    \begin{equation}
        \label{eq:pacc-approx-bounds}
        \omega(V)|^u \leq \widehat{\omega}(V)|^u \leq \omega(V).
    \end{equation}
    Here, each $u_j$ for $j \in [l]$ is a measurement outcome from the environment register $\sfE_j$ following the verifier's $j$-th action. 
    \noindent Specifically, an SDP program for computing $\widehat{\omega}(V)|^{u}$ is given as follows\emph{:}
    \begin{maxi}{}{\widehat{\omega}(V)|^u = \Tr\rbra*{V^{\dagger}_{l+1} \ketbra{1}{1}_{\Out} V_{l+1} \rho_{\ttM'_l\ttW_l}}}{}{}{\label{eq:QIPL-second-SDP}}{}
        \addConstraint{\Tr_{\ttM'_1}\!\big(\rho_{\ttM'_1\ttW_1} \!\!\otimes\! \ketbra{u_1}{u_1}_{\sfE_1}\big)}{=\Tr_{\ttM_1}\! \rbra*{ \big(I_{\ttM_1\ttW_1} \!\!\otimes\! \ketbra{u_1}{u_1}_{\sfE_1}\big) V_1 \ketbra{\bar{0}}{\bar{0}}_{\ttM_{0}\ttW_{0}} \!V_1^{\dagger} }}{}
        \addConstraint{\Tr_{\ttM'_j}\!\big(\rho_{\ttM'_j\ttW_j} \!\!\otimes\! \ketbra{u_j}{u_j}_{\sfE_j} \big)}{=\Tr_{\ttM_j}\! \rbra*{  \big(I_{\ttM_j\ttW_j}\!\!\otimes\!\ketbra{u_j}{u_j}_{\sfE_j}\big) V_j \rho_{\ttM'_{j-1}\ttW_{j-1}} \!V_j^{\dagger} },\quad}{j \in [l] \!\setminus\! [1]}
        \addConstraint{ \Tr\big(\rho_{\ttM'_1\ttW_1} \!\!\otimes\! \ketbra{u_1}{u_1}_{\sfE_1}\big) }{ =\Tr\rbra*{ \big(I_{\ttM_1\ttW_1} \!\!\otimes\! \ketbra{u_1}{u_1}_{\sfE_1}\big) V_1 \ketbra{\bar{0}}{\bar{0}}_{\ttM_{0}\ttW_{0}} \!V_1^{\dagger}} }{}
        \addConstraint{ \Tr\big(\rho_{\ttM'_j\ttW_j} \!\!\otimes\! \ketbra{u_j}{u_j}_{\sfE_j} \big) }{ =\Tr\rbra*{  \big(I_{\ttM_j\ttW_j}\!\!\otimes\!\ketbra{u_j}{u_j}_{\sfE_j}\big) V_j \rho_{\ttM'_{j-1}\ttW_{j-1}} \!V_j^{\dagger} }, \quad}{j \in [l] \!\setminus\! [1]}
        \addConstraint{ \rho_{\ttM'_j\ttW_j} \!\!\otimes\! \ketbra{u_j}{u_j}_{\sfE_j} }{ \succeq 0,\quad}{j \in [l]}
    \end{maxi}

    \noindent Here, the verifier's actions $V_1, \cdots, V_{l+1}$ correspond to space-bounded \emph{almost-unitary} quantum circuits, which we interpret as a special class of \emph{isometric} quantum circuits. 
    The variables in this SDP are \textit{unnormalized} states $\rho_{\ttM'_1\ttW_1}, \cdots, \rho_{\ttM'_l\ttW_l}$, collectively holding $O(l(n) \cdot \log{n})$ qubits. 
\end{lemma}

\begin{proof}
    Our proof strategy follows a similar approach to that of \Cref{lemma:QIPL-first-SDP-formulation}, so we will only highlight the key differences. 
    As defined in \Cref{subsec:QIPL-definition}, the verifier $V$ is described by an $(l+1)$-tuple $(V_1,\cdots,V_{l+1})$ of space-bounded almost-unitary quantum circuits $\{V_j\}_{j \in [l+1]}$, which are a special class of isometric quantum circuits, where $l = m/2$. 
    
    To represent the variables in this SDP program, which are the \textit{unnormalized} states in the snapshot registers corresponding to the message register $\sfM$ and the verifier's private register $\sfW$, we use the notations from \Cref{fig:QIPL-even}. 
    Since we focus on a fixed measurement outcome branch $u=(u_1,\cdots,u_l)$ and the coherence between subspaces corresponding to different branches is eliminated by pinching measurements, it suffices to consider each measurement outcome $u_j$ obtained from measuring $\sfE_j$ in the computational basis at the end of the verifier's $j$-th action. Consequently, the state $\ketbra{u_j}{u_j}$ in the environment register $\sfE_j$ is treated as part of the SDP constraints, not as the variables. 
    In particular, we slightly abuse the notation by letting $\rho_{\ttM_j\ttW_j}\!\otimes\! \ketbra{u_j}{u_j}_{\sfE_j}$ denote the unnormalized snapshot state in the registers $(\sfM,\sfW,\sfE_j)$ after the $j$-th verifier's action. Similarly, we define unnormalized snapshot states $\rho_{\ttM'_j\ttW_j} \!\otimes\! \ketbra{u_j}{u_j}_{\sfE_j}$ for $1 \leq j \leq l$ and $\rho_{\ttZ\ttR\sfE_{l+1}}$ for the corresponding registers. 

    \paragraph{Objective function and dependence on measurement outcomes.}
    Assuming that $\protocol{P}{V}$ begins with the verifier, it follows that the objective function is defined as
    \begin{equation}
        \label{QIPL-objective-function}
        \widehat{\omega}(V)|^u = \| \ketbra{1}{1}_{\Out} V_{l+1} P_l V_l \cdots P_1 V_1 \ket{\bar{0}}_{\ttQ_0\ttM_0\ttW_0}\|_2^2
        = \Tr\rbra*{ V_{l+1}^{\dagger} \ketbra{1}{1}_{\Out} V_{l+1} \rho_{\ttM'_l\ttW_l} }. 
    \end{equation}

    Assuming that the constraints in \Cref{eq:QIPL-second-SDP} fully capture the structure of $\protocol{P}{V}$ within a fixed branch $u$, we now establish \Cref{eq:pacc-approx-bounds}. 
    Since the SDP program in \Cref{eq:QIPL-second-SDP} optimizes only over this specific branch, the left-hand side of \Cref{eq:pacc-approx-bounds} follows: the contribution $\omega(V)|^u$ of an optimal global prover strategy, defined across all branches, may not attain the maximum in this branch-specific SDP. To conclude the right-hand side of \Cref{eq:pacc-approx-bounds}, we note that the acceptance probability of any global prover strategy consistent with this branch-$u$-specific prover strategy (a solution to this SDP program) cannot exceed $\omega(V)$. 

    \paragraph{Constraints.}
    Note that the verifier $V$ remains honest, with the initial state given by $\rho_{\ttM_0\ttW_0} = \ketbra{\bar{0}}{\bar{0}}_{\ttM_0\ttW_0}$.
    The first type of constraints arises from applying the verifier's action $V_j$:
    \begin{equation}
        \label{eq:QIPL-honest-verifier-cond}
        \begin{aligned}
            \rho_{\ttM_1\ttW_1} \!\otimes\! \ketbra{u_1}{u_1}_{\sfE_1} 
            &= \big( I_{\ttM_1\ttW_1}\!\otimes\!\ketbra{u_1}{u_1}_{\sfE_1}\big) V_1 \ketbra{\bar{0}}{\bar{0}}_{\ttM_0\ttW_0} V_1^{\dagger};\\
            \forall j \in \{2,\cdots,l\},~\rho_{\ttM_j\ttW_j} \!\otimes\! \ketbra{u_j}{u_j}_{\sfE_j} &= \big( I_{\ttM_j\ttW_j} \!\otimes\! \ketbra{u_j}{u_j}_{\sfE_j} \big) V_j \rho_{\ttM'_{j-1}\ttW_{j-1}} V_j^{\dagger};\\
            \rho_{\ttZ\ttR\sfE_{l+1}} &= V_{l+1} \rho_{\ttM'_l\ttW_l} V_{l+1}^{\dagger}.
        \end{aligned}
    \end{equation}

    Since the prover's actions are described by unitary quantum circuits, and the verifier's $j$-th action is an isometric quantum circuit, followed by measuring the (newly introduced) environment register $\sfE_j$ at the end of this turn, all intermediate states in $(\sfQ, \sfM, \sfW, \sfE_j)$ after the verifier's $j$-th action within the fixed branch $u$ in $\protocol{P}{V}$ are pure states.
    These unnormalized states, denoted by $\ket{\psi}_{\ttQ_{j-1}\ttM_j\ttW_j} \!\otimes\! \ket{u_j}_{\sfE_j}$ and $\ket{\phi}_{\ttQ_{j}\ttM'_j\ttW_j} \!\otimes\! \ket{u_j}_{\sfE_j}$ for $j \in [l]$, satisfy the relation
    \begin{equation}
        \label{eq:prover-action}
        \ket{\phi}_{\ttQ_{j}\ttM'_j\ttW_j} \!\otimes\! \ket{u_j}_{\sfE_j} = P_j \ket{\psi}_{\ttQ_{j-1}\ttM_j\ttW_j} \!\otimes\! \ket{u_j}_{\sfE_j}.
    \end{equation}
    
    Let $p(u_1,\cdots,u_j)$ be the probability that the measurement outcomes in the first $j$ rounds are $(u_1,\cdots,u_j)$. For any $j \in [l]$, the following identities hold: 
    \begin{equation}
    \label{eq:branch-normalization-factor}
    \begin{aligned}
        \ket{\psi}_{\ttQ_{j-1}\ttM_j\ttW_j} \!\otimes\! \ket{u_j}_{\sfE_j} &= p(u_1,\cdots,u_j) \cdot \ket{\widetilde{\psi}}_{\ttQ_{j-1}\ttM_j\ttW_j} \!\otimes\! \ket{u_j}_{\sfE_j}\\
        \ket{\phi}_{\ttQ_{j}\ttM'_j\ttW_j} \!\otimes\! \ket{u_j}_{\sfE_j} &= p(u_1,\cdots,u_j) \cdot \ket{\widetilde{\phi}}_{\ttQ_{j}\ttM'_j\ttW_j} \!\otimes\! \ket{u_j}_{\sfE_j}. 
    \end{aligned}
    \end{equation}
    Here, $\ket{\widetilde{\psi}}_{\ttQ_{j-1}\ttM_j\ttW_j}$ and $\ket{\widetilde{\phi}}_{\ttQ_{j}\ttM'_j\ttW_j}$ denote the corresponding \textit{normalized} states.     
    Combining \Cref{eq:prover-action,eq:branch-normalization-factor} and
    the unitary freedom in purifications (\Cref{lemma:unitary-equivalent-in-purifications}), it follows that:
    \begin{equation*}
        \forall j \in [l],~\Tr_{\ttQ_{j-1}\ttM_j}\rbra*{\ketbra{\psi}{\psi}_{\ttQ_{j-1}\ttM_j\ttW_j} \!\otimes\! \ketbra{u_j}{u_j}_{\sfE_j}} = \rho_{\ttW_j} \!\otimes\! \ketbra{u_j}{u_j}_{\sfE_j}  = \Tr_{\ttQ_{j}\ttM'_j}\rbra*{\ketbra{\phi}{\phi}_{\ttQ_{j}\ttM'_j\ttW_j} \!\otimes\! \ketbra{u_j}{u_j}_{\sfE_j}}.
    \end{equation*}
    In this expression, the underlying unitary transformation corresponds to the prover's action $P_j$ in the $j$-th round, and $\rho_{\ttW_j}$ is an unnormalized state that satisfies 
    \begin{equation}
        \label{eq:branch-prob}
        \Tr\rbra[\big]{\rho_{\ttW_j}} = p(u_1,\cdots,u_j).
    \end{equation}
    
    Consequently, the prover's actions do not interfere with the verifier's private register $\sfW$ or the environment register $\sfE_j$ (introduced by the verifier's $j$-th action) during the execution of $\protocol{P}{V}$. This property gives rise to the second type of constraints: 
    \begin{equation}
        \label{eq:QIPL-prover-actions-cond}
        \forall j \in [l],~\Tr_{\ttM_j}\rbra*{\rho_{\ttM_j\ttW_j} \!\otimes\! \ketbra{u_j}{u_j}_{\sfE_j}} = \Tr_{\ttM'_j}\rbra*{\rho_{\ttM'_j\ttW_j} \!\otimes\! \ketbra{u_j}{u_j}_{\sfE_j}}.
    \end{equation}

    Furthermore, as shown in \Cref{eq:branch-prob}, the probability that the branch $(u_1,\cdots,u_j)$ occurs is determined solely by the verifier and remains hidden from the prover. Since the prover's action cannot alter this probability, we obtain the third type of constraints:
    \begin{equation}
        \label{eq:QIPL-prover-actions-normalization}
        \forall j \in [l],~\Tr\rbra*{\rho_{\ttM_j\ttW_j} \!\otimes\! \ketbra{u_j}{u_j}_{\sfE_j}} = \Tr\rbra*{\rho_{\ttM'_j\ttW_j} \!\otimes\! \ketbra{u_j}{u_j}_{\sfE_j}}.
    \end{equation}

    Replacing the constraints in \Cref{eq:QIPL-first-SDP} by \Cref{eq:QIPL-honest-verifier-cond,eq:QIPL-prover-actions-cond,eq:QIPL-prover-actions-normalization}, we conclude the second SDP formulation for the  given \QIPL{} protocol, as specified in \Cref{eq:QIPL-second-SDP}. 
\end{proof}

\subsubsection{\QIPLHC{} is in \NP{}}
\label{subsubsec:QIPL-upper-bounds}

To establish that $\QIPLHC \subseteq \NP$ (\Cref{thm:QIPLHC-in-NP}), we choose the classical witness for the \NP{} containment as consisting of the variables $\rho_{\ttM'_1\ttW_1}, \cdots, \rho_{\ttM'_l\ttW_l}$ from the SDP program specified in  \Cref{lemma:QIPL-second-SDP-formulation}, along with the measurement outcomes $u_1,\cdots,u_l$, which determine the SDP program. Then, the verification procedure is simply performing the basic matrix operations on polynomial-dimension matrices. We now proceed with the proof. 

\begin{proof}[Proof of \Cref{thm:QIPLHC-in-NP}]
    Without loss of generality, we assume that the number of turns $m(n)$ is even, and particularly $m(n) = 2l(n)$. 
    For any $l$-round proof system $\protocol{P}{V}$ with completeness $c$ and soundness $s$, corresponding to a promise problem $\calI$ in $\QIPL_{2l}[c,s]$, we can leverage \Cref{lemma:QIPL-second-SDP-formulation} to obtain a family of SDP program, which depends on a measurement outcome branch $u=(u_1,\cdots,u_l)$ obtained by the verifier, as detailed in \Cref{eq:QIPL-second-SDP}. Specifically, given a measurement outcome branch $u$, this SDP program computes an approximation $\widehat{\omega}(V)|^u$ of the contribution of $u$ to the verifier's maximum acceptance probability $\omega(V)|^u$ over all choices of unnormalized states $\rho_{\ttM'_1\ttW_1}, \cdots, \rho_{\ttM'_l\ttW_l}$. 
    Then, it suffices to consider the approximation $\widehat{\omega}(V)|^u$ for the reminder of the proof, based on the following reasoning:
    \begin{itemize}
        \item For \textit{yes} instances, there exists a measurement outcome branch $u^*$ such that $\omega(V)|^{u^*} \geq c(n)$. By \Cref{eq:pacc-approx-bounds}, this bound implies $\widehat{\omega}(V)|^{u^*} \geq c(n)$. 
        \item For \textit{no} instances, the soundness condition combined with \Cref{eq:pacc-approx-bounds} implies that for any measurement outcome branch $u$, $\widehat{\omega}(V)|^u \leq s(n)$. This bound contradicts $\widehat{\omega}(V)|^u \geq c(n)$, since $c(n)-s(n) \geq 1/\poly(n)$. 
    \end{itemize}

    To establish an \NP{} containment, we thus choose the classical witness, denoted by $w$, as the classical description of the unnormalized states $\rho_{\ttM'_1\ttW_1}, \cdots, \rho_{\ttM'_l\ttW_l}$, alongside the binary strings representing the measurement outcomes $u_1,\cdots,u_l$. The size of $w$ remains polynomial in $n$ for the following reasons: (i) the dimension of each unnormalized state $\rho_{\sfM'_j\sfW_j}$, for $1 \leq j \leq l$, is at most $2^{O(\log n)}$, which is polynomial in $n$; (ii) the length of each binary string $u_j$ is bounded by $O(\log n)$; and (iii) the number of rounds $l(n)$ is at most $\poly(n)$. 

    We now describe the \NP{} verification procedure to complete the proof. Given the classical witness $w$, the procedure $\hatV$ executes the following steps:
    \begin{enumerate}[label={\upshape(\arabic*)}, topsep=0.33em, itemsep=0.33em, parsep=0.33em] 
        \item Check whether $w$ represents a feasible solution of the SDP program for computing $\widehat{\omega}(V)|^u$, where $u$ is the measurement outcome branch given in $w$, as specified in \Cref{eq:QIPL-second-SDP}. 
        \item Compute the value of $\widehat{\omega}(V)|^u$ by performing a polynomial number of matrix multiplications and partial traces of polynomial-dimensional matrices. 
    \end{enumerate}
    
    It is evident that these steps can be accomplished in deterministic polynomial time. The verification procedure $\hatV$ accepts if the witness $w$ is a feasible solution to \Cref{eq:QIPL-second-SDP} concerning the given $u$ and the value of $\widehat{\omega}(V)|^u$ is at least $c(n)$; otherwise, $\hatV$ rejects. 
\end{proof}

\subsection{Basic properties: Perfect completeness and error reduction}
\label{subsec:QIPL-basic-properties}

The verifier's space constraint for the class \QIPL{} presents several challenges when adapting techniques from standard quantum interactive proofs. For instance, techniques such as error reduction through parallel repetition~\cite[Section 5]{KW00} and the parallelization approach described in~\cite[Section 4]{KW00} are applicable to \QIPL{} only under certain conditions. Nevertheless, two basic properties can still be established without additional assumptions: 
\begin{enumerate}[label={\upshape(\arabic*)}, topsep=0.33em, itemsep=0.33em, parsep=0.33em] 
    \item Achieving perfect completeness in \QIPL{} or \QIPUL{} proof systems by adapting the technique in~\cite[Section 4.2.1]{VW16}, as detailed in \Cref{lemma:QIPL-perfect-completeness}; 
    \item Error reduction for \QIPL{} and \QIPUL{} via sequential repetition, as stated in \Cref{lemma:QIPL-error-reduction}. 
\end{enumerate}

\begin{remark}[Perfect completeness and error reduction for \QIPLHC{}]
    \label{remark:QIPLHC-properties}
    Any \QIPLHC{} proof system $\protocol{P}{V}$ can be transformed to achieve perfect completeness, as implied by combining $\QIPLHC \subseteq \NP$ (\Cref{thm:QIPLHC-in-NP}) and $\NP \subseteq \QIPLHC[1,1/3]$ (\Cref{thm:NP-in-QIPL}), resulting in the proof system $\protocol{P'}{V'}$. However, the approach in \Cref{lemma:QIPL-perfect-completeness} ensures that only the concentrated branch accepts with certainty, while other branches may still reject with non-zero probability. 
    Additionally, error reduction for \QIPLHC{} proof systems can be obtained by applying \Cref{lemma:QIPL-error-reduction} to $\protocol{P'}{V'}$. This is because $\protocol{P'}{V'}$ has a single branch defined by measurement outcomes for \textit{yes} instances, whereas general \QIPLHC{} proof systems may exhibit acceptance concentrated on a particular branch but still allow non-zero acceptance across multiple branches. 
\end{remark}

We will provide detailed proofs of these properties in the remainder of this subsection.

\subsubsection{Achieving perfect completeness for \QIPL{} and \QIPUL{}}

Our construction and analysis in \Cref{lemma:QIPL-perfect-completeness} are inspired by~\cite[Section 4.2.1]{VW16}.  

\begin{lemma}[\QIPL{} and \QIPUL{} are closed under perfect completeness]
    \label{lemma:QIPL-perfect-completeness}
    Let $c(n)$, $s(n)$, and $m(n)$ be logspace-computable functions such that $0 \leq s(n) < c(n) \leq 1$, $c(n) - s(n) \geq 1/\poly(n)$, and $1 \leq m(n) \leq \poly(n)$. Then, it follows that 
    \[
    \QIPL_m[c,s] \subseteq \QIPL_{m+2}\sbra*{1, 1-(c-s)^2/2} \text{ and } \QIPUL_m[c,s] \subseteq \QIPUL_{m+2}\sbra*{1, 1-(c-s)^2/2}. 
    \]
\end{lemma}

\begin{proof}
    Since \QIPUL{} proof systems are a special subclass of \QIPL{} proof systems, it suffices to establish the inclusion for the latter. 
    For any $m$-turn \QIPL{} proof system $\protocol{P}{V}$ with completeness $c$ and soundness $s$, which corresponds to a promise problem $\calI$ in $\QIPL_m[c,s]$, $\protocol{P}{V}$ acts on registers $\sfQ$, $\sfM$, and $\sfW = (\ttZ, \ttR)$, where $\ttZ$ represents the output qubit just before the final measurement. To achieve perfect completeness, we propose a new proof system $\protocol{P'}{V'}$ based on $\protocol{P}{V}$, as detailed in \Cref{protocol:QIPL-perfect-completeness}. In this proof system, we also introduce a new single-qubit register $\sfZ'$, which is initialized to be $\ket{0}$ and is accessible only to $V'$. 
    
    It is noteworthy that the environment register introduced by the verifier's $j$-th action for $j \in \floor*{(m+1)/2}+1$ in $\protocol{P}{V}$ has no effect on the new proof system $\protocol{P'}{V'}$. This is because $\sfE_j$ is measured at the end of the turn in which it is introduced, it collapses to a quantum state that contains only a binary string and becomes inaccessible afterward. 

\begin{algorithm}[ht!]
    \SetAlgorithmName{Protocol}
    \SetEndCharOfAlgoLine{.}
    \setlength{\parskip}{5pt}
    \SetKwInOut{Parameter}{Parameters}

    \Parameter{$\alpha$ is a dyadic rational number such that $\frac{3c+s}{4} \leq \alpha \leq c$.}
    \textbf{1.} The verifier $V'$ executes the original proof system $\protocol{P}{V}$ (with the prover $P'$), except for the verifier's final measurement on the register $\ttZ$\;
    \textbf{2.} The verifier $V'$ creates a pseudo-copy of the output qubit of $V$ by applying a $\CNOT_{\ttZ \rightarrow \sfZ'}$ gate, and then sends the register $\sfW = (\ttZ, \ttR)$ to the prover\;
    \textbf{3.} The verifier $V'$ receives a single-qubit state from the prover, places it in the register $\ttZ$, and measures registers $\ttZ$ and $\sfZ'$ in the binary-valued measurement $\cbra*{\ketbra{\gamma}{\gamma}, I-\ketbra{\gamma}{\gamma}}$, where $\ket{\gamma} \coloneqq \sqrt{1-\alpha}\ket{00} + \sqrt{\alpha}\ket{11}$. The verifier $V'$ accepts if the measurement outcome is consistent with $\ket{\gamma}$; otherwise, it rejects. 
    \BlankLine
	\caption{Achieving perfect completeness for \QIPL{} (or \QIPUL{}).}
	\label[algorithm]{protocol:QIPL-perfect-completeness}
\end{algorithm}

    Next, we establish the correctness of \Cref{protocol:QIPL-perfect-completeness}: 
    \begin{itemize}[leftmargin=2em]
        \item For \textit{yes} instances, the state of the total system after step 2 can be adjusted with the help of an honest prover to
        \[\ket{\Phi} = \sqrt{1-\alpha} \ket{00}_{\sfZ'\ttZ}\ket{\phi_0} + \sqrt{\alpha} \ket{11}_{\sfZ'\ttZ}\ket{\phi_1},\]
        where $\ket{\phi_0}$ and $\ket{\phi_1}$ are normalized states that may not be orthogonal. 
        This adjustment can be done because $c\geq\alpha$. The prover then applies a unitary $U$ on all the qubits except $\sfZ'$ (which are owned by the prover) to ``disentangle'' $\ket{\Phi}$: 
        \[ (I_{\sfZ'}\otimes U) \ket{\Phi} = \rbra*{\sqrt{1-\alpha} \ket{00}_{\sfZ'\ttZ} + \sqrt{\alpha} \ket{11}_{\sfZ'\ttZ}} \otimes \ket{\phi}, \text{ where } \ket{\phi} \text{ is a normalized state.} \]
        Consequently in Step 3, the verifier $V'$ holds the state 
        $\rbra*{\sqrt{1-\alpha} \ket{00} + \sqrt{\alpha} \ket{11}} = \ket{\gamma}$ in registers $\ttZ$ and $\sfZ'$, ensuring that $V'$ accepts with certainty. 

        \item For \textit{no} instances, the original proof system $\protocol{P}{V}$ accepts with probability at most $s = \alpha - \varepsilon$ where $\varepsilon \geq  \frac{3c+s}{4} -s = \frac{3}{4}(c-s)$. The reduced density matrix in the register $\sfZ'$ after Step 2 is 
        \[  \rho_{\sfZ'} = \begin{pmatrix} 
            1-\alpha+\varepsilon & 0\\ 0 & \alpha-\varepsilon
        \end{pmatrix}.\]
        
        Let $\sigma_{\ttZ\sfZ'}$ denote a two-qubit quantum state in registers $(\ttZ,\sfZ')$ after the verifier $V'$ receives the single-qubit state from the prover $P'$ in Step 3. 
        Regardless of the prover's actions, $V'$ accepts with probability 
        \begin{align*}
            \Tr\rbra*{\ketbra{\gamma}{\gamma}\sigma_{\ttZ\sfZ'}} 
            &= \F\rbra*{\ketbra{\gamma}{\gamma}, \sigma_{\ttZ\sfZ'}}^2 \\
            &\leq \F\rbra*{\Tr_{\ttZ}\rbra*{\ketbra{\gamma}{\gamma}}, \rho_{\sfZ'}}^2 \\
            &= \F\rbra*{\begin{pmatrix} 
                    1-\alpha & 0\\ 0 & \alpha
                \end{pmatrix},
                \begin{pmatrix} 
                    1-\alpha+\varepsilon & 0\\ 0 & \alpha-\varepsilon
                \end{pmatrix}}^2\\
            &\leq 1-\varepsilon^2\\
            &\leq 1 - \frac{1}{2}(c-s)^2.
        \end{align*}

        Here, the second line owes to the data-processing inequality for the fidelity (\Cref{lemma:fidelity-data-processing}), the fourth line follows from~\cite[Equation (4.25)]{VW16}, and the last line is because of 
        \[1-\varepsilon^2 \leq 1 - \rbra*{\frac{3(c-s)}{4}}^2 \leq 1 - \frac{1}{2}(c-s)^2. \qedhere\]
    \end{itemize}

\end{proof}

\subsubsection{Error reduction for \QIPL{} and \QIPUL{}}

The main challenge in performing sequential repetition of a given \QIPL{} or \QIPUL{} proof system $\protocol{P}{V}$ lies in resetting the qubits in the verifier's private register $\sfW$ to the all-zero state after each execution of $\protocol{P}{V}$. 
In non-interactive proof systems with a unitary logspace verifier, the reset-to-zero operation can be achieved by running the inverse of the verification circuit, as shown in~\cite{MW05,FKLMN16}. However, in interactive proof systems with a restricted logspace verifier -- whether unitary or almost-unitary -- the reset-to-zero operation requires assistance from the prover,\footnote{In the case of $\QIPL{}^\diamond$ proof systems, where the verifier's actions are isometric quantum circuits, implementing the reset-to-zero operation becomes straightforward. However, for $\QIPL{}$ proof systems, the $O(\log{n})$ pinching intermediate measurements conducted during each verifier action are insufficient for re-using the workspace qubits. This limitation arises because the measurement outcome yields merely a binary string encoded in a state.} who may not be honest.
We now proceed to the formal statement:

\begin{lemma}[Error reduction for \QIPL{} and \QIPUL{}]
    \label{lemma:QIPL-error-reduction}
    Let $c(n)$, $s(n)$, and $m(n)$ be logspace-computable functions such that $0 \leq s(n) < c(n) \leq 1$, $c(n) - s(n) \geq 1/\poly(n)$, and $1 \leq m(n) \leq \poly(n)$. For any polynomial $k(n)$, it holds that
    \[ \QIPL_m[c,s] \subseteq \QIPL_{m'}\big[1, 2^{-k}\big] \text{ and } \QIPUL_m[c,s] \subseteq \QIPUL_{m'}\big[1, 2^{-k}\big]. \]
    Here, the number of turns $m' \coloneqq O\big(km/\log\frac{1}{1-(c-s)^2/2}\big)$.
\end{lemma}

\begin{proof}
    Since \QIPUL{} proof systems are a special subclass of \QIPL{} proof systems, it suffices to establish the inclusion for the latter. 
    For any $m$-turn proof system $\protocol{P'}{V'}$ with completeness $c$ and soundness $s$, corresponding to a promise problem $\calI$ in $\QIPL_m[c,s]$, applying \Cref{lemma:QIPL-perfect-completeness} yields a new $(m+2)$-turn proof system $\protocol{P}{V}$ with completeness $1$ and soundness $1-(c-s)^2/2$.
    This proof system $\protocol{P}{V}$, viewed as an isometric quantum circuit, acts on registers $\sfQ$, $\sfM$, and $\sfW = (\ttZ, \ttR)$, where $\ttZ$ denotes the output qubit just before the final measurement. Without loss of generality, we can assume that registers $\sfM$ and $\sfV$ are of equal size. 
        
    Error reduction for the given proof system $\protocol{P}{V}$ is achieved using $r$-fold AND-type sequential repetition of $\protocol{P}{V}$, as detailed in \Cref{protocol:QIPL-error-reduction}. In this resulting proof system $\protocol{\hatP}{\hatV}$, we introduce two new $\ceil*{\log{r}}$-qubit registers, $\widehat{\sfS}$ and $\widehat{\sfT}$, which are initialized to be the all-zero state and are private to $\hatV$. 
    The procedure from Step 3.a to Step 3.c aims to reset the verifier's original private register $\sfW$ with the help of the prover. Moreover, the multiple-controlled adder in Step 3.b can be implemented by $O(q_{\sfW}^2)$ uses of elementary quantum gates and the adder $U_{\rm add}$, following from~\cite[Lemma 7.5 and Corollary 7.6]{BBC+95}.

    A subtle but important point concerns the effect of the environment register $\sfE_j$ introduced by the verifier's $j$-th action for $j\in [\floor*{(m\!+\!1)/2}\!+\!1]$ in $\protocol{P}{V}$. Since $\sfE_j$ is measured at the end of the turn in which it is introduced, it collapses to a state that contains only a binary string and becomes inaccessible afterward. Therefore, this newly introduced register $\sfE_j$ has no impact on our sequential repetition protocol.
    
\begin{algorithm}[ht!]
    \SetAlgorithmName{Protocol}
    \SetEndCharOfAlgoLine{.}
    \SetAlgoVlined
    \setlength{\parskip}{5pt}
    \SetKwFor{For}{For}{:}{}
    \SetKwIF{If}{}{}{If}{:}{}{}{}
    \SetKwInOut{Parameter}{Parameters}

    \Parameter{$r \coloneqq O\rbra*{k/(c-s)^2}$.}
    \For{$i \leftarrow 1$ \KwTo $r$}{
        \textbf{1.} The verifier $\hatV$ executes the original proof system $\protocol{P}{V}$ (with the prover $\hatP$), except for the verifier's final measurement on the register $\ttZ$\;
        \textbf{2.} The verifier $\hatV$ performs a controlled adder, where the single-qubit control register is $\widehat{\ttZ}$, and the $\ceil*{\log{r}}$-qubit target register is $\widehat{\sfS}$\;
        \If{i < r}{
            \textbf{3.a} The prover $\hatP$ sends a $q_{\sfW}$-qubit state $\ket{O}$ to the verifier $\hatV$, where $\ket{O}=\ket{0}^{\otimes q_{\sfW}}$ for an honest prover\;
            \textbf{3.b} The verifier $\hatV$ performs a multiple-controlled adder, where the $q_{\sfW}$-qubit control register is $\widehat{\sfM}$ (containing $\ket{O}$), the $\ceil*{\log{r}}$-qubit target register is $\widehat{\sfT}$, and the adder is activated if $\ket{O} = \ket{0}^{\otimes q_{\sfW}}$\;
            \textbf{3.c} The verifier $\hatV$ performs a \SWAP{} gate between registers $\widehat{\sfM}$ and $\widehat{\sfW}$;
        }
    }
    \textbf{4.} The verifier $\hatV$ measures the registers $\widehat{\sfS}$ and $\widehat{\sfT}$ in the computational basis, with the outcomes denoted as $\Cnt_{\acc}$ and $\Cnt_{\clean}$, respectively\;
    \textbf{5.} The verifier $\hatV$ accepts if both $\Cnt_{\clean} = r-1$ and $\Cnt_{\acc} = r$ are satisfied. 
    \BlankLine
	\caption{Error reduction for \QIPL{} (or \QIPUL{}) via sequential repetition.}
	\label[algorithm]{protocol:QIPL-error-reduction}
\end{algorithm}

    \vspace{1em}
    It remains to establish the correctness of $\protocol{\hatP}{\hatV}$. Let $X_i$ be a random variable indicating whether the $i$-th execution of $\protocol{P}{V}$ is accepted, with $\Pr{X_i = 1}$ denoting the verifier $V$'s maximum acceptance probability of the $i$-th execution.  
    By a direct calculation, it holds that
    \begin{equation}
        \label{eq:QIPL-error-reduction-AND}
        \Pr{\Cnt_{\acc} = r} = \Pr{ X_1=1 \wedge \cdots \wedge X_r = 1 } = 
        \Pr{\protocol{P}{V}\text{ accepts}}^r.
    \end{equation}
    As a consequence, we conclude the following: 
    \begin{itemize}[leftmargin=2em]
        \item For \textit{yes} instances, an honest prover always sends $\ket{O} = \ket{0}^{\otimes q_{\sfW}}$, 
        and runs each of the executions independently, ensuring that the condition $\Cnt_{\clean} = r-1$ is satisfied. By combining \Cref{eq:QIPL-error-reduction-AND} with the completeness condition of $\protocol{P}{V}$, it follows that \Cref{protocol:QIPL-error-reduction} accepts with certainty. 
        \item For \textit{no} instances, it is enough to analyze the acceptance probability when the prover always sends $\ket{0}^{\otimes q_{\sfW}}$ at Step 3.a, as the verifier will otherwise reject. 
        Although the random variables $X_i$ may not be independent, 
        the probability of $X_i=1$ is at most $1 - (c-s)^2/2$ for any prover. Therefore, we can be upper bound the acceptance probability using independent binary random variables $X'_i$, where the probability of $X'_i=1$ is $1 - (c-s)^2/2$. 
        Consequently, applying \Cref{eq:QIPL-error-reduction-AND} (by substituting $X_i$ with $X'_i$), we can obtain that \Cref{protocol:QIPL-error-reduction} accepts with probability at most $2^{-k}$ by choosing $r = O\rbra*{k/\log\frac{1}{1-(c-s)^2/2}}$. This choice results in $m' = r(m+2) = O\rbra*{km/\log\frac{1}{1-(c-s)^2/2}}$. \qedhere
    \end{itemize}
\end{proof}

\subsection{Lower bounds for \QIPL{} and \QIPUL{}}
\label{subsec:lower-bounds-QIPL}

Our lower bounds for \QIPL{} and \QIPUL{} are motivated by the prior works on space-bounded classical interactive proofs, particularly those involving either private coins~\cite{CL95} or public coins~\cite{Fortnow89,FL93,Condon92survey, GKR15, CR23}:\footnote{Space-bounded classical interactive proofs with $\poly(n)$ public coins are in \Ptime{}, see~\cite[Theorem 6]{Condon92survey}.}
\begin{itemize}[itemsep=0.33em,topsep=0.33em,parsep=0.33em]
    \item \textbf{Private-coin proof systems vs.~\QIPL{}:} Soundness against classical messages may not extend to quantum messages for private-coin proof systems. This is because the prover could generate shared entanglement with the verifier, potentially leaking information about the private coins. Consequently, space-bounded private-coin classical interactive proofs are simulatable only by \QIPL{} proof systems (\Cref{thm:NP-in-QIPL}). 
    \item \textbf{Public-coin proof systems vs.~\QIPUL{}:} Public coins can be simulated by halves of EPR pairs sent from the verifier in \QIPUL{} proof systems, thus avoiding the soundness issue. However, the verifier is limited to sending only $O(\log{n})$ halves of EPR pairs, leading to a limitation on completeness. Therefore, only the work of~\cite{Fortnow89} can be simulated by \QIPUL{} proof systems (\Cref{thm:LOGCFP-in-QIPUL-weakErrorBound} and thus \Cref{thm:LOGCFL-in-QIPUL}). 
\end{itemize}

In the remainder of this subsection, we first establish that $\NP \subseteq \QIPLHC \subseteq \QIPL$ in \Cref{subsubsec:NP-in-QIPL}, and then demonstrate that $\SAC^1 \cup \BQL \subseteq \QIPUL$ in \Cref{subsubsec:LOGCFL-in-QIPUL}. 

\subsubsection{\NP{} is in \QIPLHC{}}
\label{subsubsec:NP-in-QIPL}

Our approach in \Cref{thm:NP-in-QIPL} draws inspiration from~\cite[Lemma 2]{CL95}.
Soundness against classical messages is guaranteed by \textit{the fingerprinting lemma}~\cite{Lipton90}, which is used for comparing multisets through short fingerprints (see \Cref{subsec:classical-concepts-and-tools}). 
To ensure soundness against quantum messages in private-coin proof systems, the verifier must measure the received quantum message in the computational basis at the beginning of each action. 

\begin{theorem}
    \label{thm:NP-in-QIPL}
    $\NP \subseteq \QIPL^\HC_m[1,1/3] \subseteq \QIPL_m[1,1/3]$, where $m(n)$ is some polynomial in $n$.
\end{theorem}

\begin{proof}
    We begin by noting that \threeSAT{} is \NP{}-hard under logspace reductions (\Cref{lemma:3SAT-NP-hard}), and that both \QIPL{} and its subclass \QIPLHC{} are closed under logspace reductions.\footnote{\QIPL{} is closed under logspace reductions if, for any logspace reduction $R$ such that $\calI \leq_{\Lspace}^m \calI'$, it holds that $\calI \in \QIPL$ if $\calI' \in \QIPL$. Let $\calM_{\calI'}$ and $\calM_R$ be the deterministic logspace Turing machines that compute (the description of) the verifier's mapping associated with $\calI'$ and the reduction $R$, respectively. The closure under logspace reductions for \QIPL{} is then achieved by considering the concatenation $\calM_{\calI} \coloneqq \calM_R \circ \calM_{\calI'}$. A similar argument applies to its subclass \QIPLHC{}, even under a weaker notion of the uniformity of the verifier's mapping, as discussed in \Cref{footref:logspace-uniformity}.} It suffices to establish that \threeSAT{} is in \QIPLHC{}. 

    To verify whether a \threeSAT{} instance $\phi=C_1 \wedge \cdots \wedge C_k$ is satisfied by an assignment $\alpha$, we encode $\phi(\alpha)$ as a collection of $3k$ tuples $(l,i,v)$, denoted as $\Enc(\phi(\alpha))$. In this encoding, each $(l,i,v)$ represents the literal $l\in \{x_j\}_{j \in [n]} \cup \{\neg x_j\}_{j \in [n]}$ in the clause $C_i$ is assigned the value $v \in \{\top,\bot\}$. Hence, $\Enc(\phi(\alpha))$ forms a multiset of $\ell=3k$ elements, with each element representable by $b=2\ceil*{\log{n}}+\ceil*{\log{k}}+1 = O(\log(kn))$ bits. For convenience, let $a_{(l,i,v)}$ denote the non-negative integer that encodes the triple $(l,i,v)$, and let $\var(l)$ be the variable in the literal $l$. 
    Next, we proceed with the detailed \QIPL{} proof system, satisfying the high-concentration condition for \textit{yes} instances, as outlined in \Cref{protocol:QIPL-3SAT-protocol}. 

\begin{algorithm}[ht!]
    \SetAlgorithmName{Protocol}
    \SetEndCharOfAlgoLine{.}
    \SetAlgoVlined
    \setlength{\parskip}{5pt}
    \SetKwFor{While}{}{:}{}
    \SetKwInOut{Parameter}{Parameters}
    \Parameter{$\phi(\alpha)$ is a \threeSAT formula $\phi$ with an assignment $\alpha$; $n$ and $k$ are the number of variables and clauses in $\phi$, respectively.}

    \textbf{1.} The verifier $V$ chooses a prime $p$ and an integer $r$ uniformly at random from the intervals $[(b\ell)^2, 2(b\ell)^2]$ and $[1,p-1]$, respectively. The verifier then initializes the (partial) fingerprints $F_\var =1$ and $F_{\cl}=1$.
    
    \textbf{2.} \While{\textsc{Consistency Check} \textnormal{(for each variable)}}{
        \textbf{$\bullet$} The prover $P$ sends the triples $(l,i,v)$ in $\Enc(\phi(\alpha))$, represented as quantum states $\ket{\psi_{\var(l),C_i}}$, to the verifier $V$, ordered by the variable $\var(l)$ in the literal $l$ and then by the clause index $i$\;
        \textbf{$\bullet$} \While{\textnormal{For each state $\ket{\psi_{\var(l),C_i}}$ received}}{
        \textbf{2.a} $V$ measures the state $\ket{\psi_{\var(l),C_i}}$ in the computational basis, with the measurement outcome denoted by the triple $(l,i,v)$\;
        \textbf{2.b} $V$ rejects if the following conditions hold: (i) $(l,i,v)$ is not the first triple in $\Enc(\phi(\alpha))$, (ii) $\var(l)=\var(l')$, and (iii) the value $v$ and $v'$ are inconsistent\;
        \textbf{2.c} $V$ updates the fingerprint $F_\var = F_\var \cdot \rbra*{a_{(l,i,v)}+r} \mod p$\;
        \textbf{2.d} $V$ sends the previous triple $(l',i',v')$ back to $P$, if applicable. $V$ then retains the current triple $(l,i,v)$ in its private memory, unless $(l,i,v)$ is the last triple in $\Enc(\phi(\alpha))$. 
        }
    }
    \textbf{3.} \While{\textsc{Satisfiability Check} \textnormal{(for each clause)}}{
        \textbf{$\bullet$} The prover $P$ sends the triples $(l_1,i,v_1)$, $(l_2,i,v_2)$, and $(l_3,i,v_3)$ in $\Enc(\phi(\alpha))$, represented as states $\ket{\psi_{C_i}}$, to the verifier $V$, ordered by the clause index $i$\;
        \textbf{$\bullet$} 
        \While{\textnormal{For each clause $C_i$, with the state $\ket{\psi_{C_i}}$ received}}{
        \textbf{3.a} $V$ measures the state $\ket{\psi_{C_i}}$ in the computational basis, where this state is expected to be $\ket{\psi_{\var(l_1),C_i}} \!\otimes\! \ket{\psi_{\var(l_2),C_i}} \!\otimes\! \ket{\psi_{\var(l_3),C_i}}$. The measurement outcomes are denoted by the triples $(l_1,i,v_1)$, $(l_2,i,v_2)$, and $(l_3,i,v_3)$\;
        \textbf{3.b} If $v_1 \vee v_2 \vee v_3 = \bot$, $V$ rejects\;
        \textbf{3.c} $V$ updates the fingerprints $F_{\cl} = F_{\cl} \cdot \prod_{j\in [3]}\big(a_{(l_j,i,v_j)}+r\big) \mod p$\; 
        \textbf{3.d} $V$ returns the triples $(l_1,i,v_1)$, $(l_2,i,v_2)$, and $(l_3,i,v_3)$ to $P$.
        }
    }
    \textbf{4.} The verifier $V$ accepts if the fingerprints $F_\var = F_{\cl}$; otherwise, it rejects. 
    \BlankLine
	\caption{A \QIPL{} proof system for \threeSAT{}.}
	\label[algorithm]{protocol:QIPL-3SAT-protocol}
\end{algorithm}

    \vspace{1em}
    To establish the correctness of \Cref{protocol:QIPL-3SAT-protocol}, 
    we first observe that since $k$ is a polynomial in $n$, the integer $2(bl)^2$ can be represented using $O(\log n)$ bits. 
    Following the argument in the proof~\cite[Lemma 2]{CL95},\footnote{See the first paragraph on Page 515 in~\cite{CL95} for the details.} a classical logspace verifier can find the prime $p$ and the integer $r$ in Step 1 of \Cref{protocol:QIPL-3SAT-protocol} with probability at least $3/4$, using $O(\log n)$ random bits. 
    Since a classical operation depending on $r$ random bits can be simulated by a corresponding unitary controlled by the state $\ket{+}^{\otimes r}$, Step 1 can be implemented using $O(\log{n})$ ancillary qubits, which remain untouched throughout the rest of the proof system. Additionally, because the state $\ket{\psi_{\var{l},C_i}}$, which encodes the triple $(l,i,v)$, requires at most $O(\log n)$ qubits, space-bounded almost-unitary quantum circuits suffice to implement the verifier's actions in \Cref{protocol:QIPL-3SAT-protocol}.
    
    Although there are $3k+k=4k$ rounds in \Cref{protocol:QIPL-3SAT-protocol}, the verifier's actions are of only constantly many kinds. Therefore, the verifier's mapping is logspace computable. We now turn to the analysis required to finish the proof:
    \begin{itemize}
        \item For \textit{yes} instances, where the \threeSAT{} formula $\phi$ is satisfiable, there is a prover strategy such that the verifier in \Cref{protocol:QIPL-3SAT-protocol} accepts with certainty. Additionally, the pinching measurement outcome string $u^*$ corresponds to this prover strategy, specifically the classical messages sent by the prover, implying that $\omega(V)|^{u^*} = 1$. 
        \item For \textit{no} instances, we first bound the probability that the verifier in \Cref{protocol:QIPL-3SAT-protocol} accepts an unsatisfiable $\phi$, assuming the prover sends only classical messages. This event may happen if: (1) the verifier fails to successfully choose the random prime $p$ and the random integer $r$; or (2) the two multisets sent by the prover -- specifically in Step 2 and Step 3 of \Cref{protocol:QIPL-3SAT-protocol} -- are unequal but still yield the same fingerprint. According to the proof of~\cite[Lemma 2]{CL95}, the former occurs with probability at most $1/4$. And the fingerprinting lemma (\Cref{lemma:fingerprint}) ensures that the latter occurs with probability at most 
        \[O\rbra*{ \frac{\log b + \log \ell}{b\ell}+ \frac{1}{b^2\ell} } = O\rbra*{ \frac{\log\log(kn) + \log(3k)}{\log(kn) \cdot 3k} + \frac{1}{\log(kn)^2 \cdot 3k}} = O\rbra*{\frac{1}{k}}.\]
        Therefore, the acceptance probability is at most $1/4 + O(1/k) \leq 1/3$. 

        Next, let $\omega(V)|_u$ be the verifier $V$'s maximum acceptance probability conditioned on the intermediate measurement outcome $u$. A direct calculation shows that $\omega(V)$ is a convex combination of $\omega(V)|_u$ over all obtainable measurement outcomes $u$. Since each verifier action begins with a measurement that forces the prover's message to be classical, we have: 
        \[ \omega(V) = \sum_{u \in \calJ_{\rm obt}} p_u \cdot \omega(V)|_u \leq \sum_{u \in \calJ_{\rm obt}} p_u \cdot \frac{1}{3} = \frac{1}{3}. \]
        Here, $p_u$ denotes the probability of obtaining the measurement outcome $u$, $\calJ_{\rm obt}$ represents the index set of all obtainable intermediate measurement outcomes, and the inequality follows from the established soundness against classical messages. 
    \end{itemize}
    In conclusion, we complete the proof by establishing that $\threeSAT \in \QIPL^\HC_{8k}[1,1/3]$. 
\end{proof}

\subsubsection{\texorpdfstring{$\SAC^1 \cup \BQL$ is in \QIPUL{}}{}}
\label{subsubsec:LOGCFL-in-QIPUL}

Our approach follows~\cite[Section 3.4]{Fortnow89}. To establish \Cref{thm:LOGCFL-in-QIPUL}, along with error reduction for \QIPUL{} (\Cref{lemma:QIPL-error-reduction}) and $\BQL \subseteq \QIPUL$, we need to prove the following: 

\begin{lemma}
    \label{thm:LOGCFP-in-QIPUL-weakErrorBound}
    $\SAC^1 \subseteq \QIPUL_{O(\log{n})}[1,1-1/p(n)]$, where $p(n)$ is some polynomial in $n$.
\end{lemma}

\begin{proof}
For any promise problem $\calI = (\calI_{\yes},\calI_{\no})$ in $\SAC^1$ where $\calI_{\yes} \cup \calI_{\no} = \binset^*$, it suffices to consider the corresponding (uniform) $\SAC^1$ circuit evaluation problem, as defined in \Cref{def:uniform-SAC1}. 
Let $C$ be the (uniform) $\SAC^1$ circuit associated with $\calI$,  taking $x \in \calI$ as input, such that $C(x)=1$ if and only if $x \in \calI_{\yes}$. 

We now present the \QIPUL{} proof system for $\SAC^1$, as detailed in \Cref{protocol:LOGCFL-in-QIPUL}.

\begin{algorithm}[ht!]
    \SetAlgorithmName{Protocol}
    \SetEndCharOfAlgoLine{.}
    \SetAlgoVlined
    \setlength{\parskip}{5pt}
    \SetKwFor{While}{}{:}{}
    \SetKwIF{If}{}{}{If}{:}{}{}{}
    \SetKwInOut{Parameter}{Parameters}

    \textbf{1.} \While{\textnormal{The prover $P$ and verifier $V$ start at the output level of the circuit $C$}}{
        \textbf{$\bullet$} \If{\textnormal{the top gate $G$ is an OR gate}}{
            \textbf{2.1} $P$ selects one of $G$'s child gates\; 
            \textbf{2.2} $P$ sends the selection to $V$ by sending a quantum state $\ket{\psi}$\;
        }
        \textbf{$\bullet$} \If{\textnormal{the top gate $G$ is an AND gate}}{
            \textbf{2.a} $V$ selects one of the two child gate of $G$ uniformly at random\;
            \textbf{2.b} $V$ send the selection to $P$ using half of an EPR pair\;
        }
    }

    \textbf{2.} At each intermediate level of the circuit $C$, the prover $P$ and verifier $V$ repeat this process on the selected child gate, as outlined in Step 1. 

    \textbf{3.} \While{\textnormal{At the input level of the circuit $C$}}{
        \textbf{3.1} The verifier $V$ measures the register containing the currently selected child gate in the computational basis. Here, the measurement outcomes corresponds to either $x_j$ or $\neg x_j$ for some $j \in [n]$\; 
        \textbf{3.2} \While{\textnormal{The verifier $V$ checks the following}}{
            \textbf{$\bullet$} For \textbf{an input $x_i$}: $V$ accepts if $x_i=1$; otherwise, $V$ rejects\; 
            \textbf{$\bullet$} For \textbf{a negation of an input $x_i$}: $V$ accepts if $x_i=0$; otherwise, $V$ rejects. 
        }
    }
    \BlankLine
	\caption{A \QIPUL{} proof system for (uniform) $\SAC^1$.}
	\label[algorithm]{protocol:LOGCFL-in-QIPUL}
\end{algorithm}

    To establish the correctness of \Cref{protocol:LOGCFL-in-QIPUL}, we first observe that  the depth of the circuit $C$ is $O(\log n)$, implying that the size is polynomial in $n$. Consequently, each prover's selection can be represented by an $O(\log n)$-bit string. 
    This observation also implies that \Cref{protocol:LOGCFL-in-QIPUL} has $O(\log n)$ rounds, but the verifier's actions are of only three kinds. Therefore, the verifier's mapping is logspace computable. We now turn to the analysis required to complete the proof: 
    \begin{itemize}
        \item For \textit{yes} instances, for any choices made by the verifier $V$ at the AND gates, there exist corresponding choices at the OR gates that lead to the circuit $C$ to accept. Therefore, the prover $P$ has a winning strategy for all verifier choices, ensuring that the verifier $V$ accepts with certainty. 
        \item For \textit{no} instances, since the computational paths in the circuit $C$ do not interfere with each other, it suffices to establish soundness against classical messages. Note that certain verifier choices will cause $V$ to reject. Given the $O(\log n)$ depth of the circuit $C$, $V$ makes one choice out of two at each AND gate, resulting in a rejection probability of at least $2^{-O(\log {n})} = 1/p(n)$ for some polynomial $p(n)$. Therefore, the verifier accepts with probability at most $1-1/p(n)$. \qedhere
    \end{itemize}
\end{proof}

\section{Constant-message space-bounded quantum interactive proofs}
\label{sec:cosnt-message-QIPL}

In this section, we investigate space-bounded quantum interactive proof systems with a \textit{constant} number of messages. Building on the definitions in \Cref{subsec:QIPL-definition}, we define the classes \QIPLconst{} and $\QMAML$ with \textit{constant} promise gap as follows: 
\begin{equation}
    \label{eq:QIPLconst-defs}
    \QIPLconst \coloneq \cup_{m \leq O(1)} \QIPUL_m[2/3,1/3] \text{ and } \QMAML \coloneq \QMAML[1,1/3].
\end{equation}
Here, the class \QMAML{} possesses \textit{public-coin} three-message space-bounded \textit{unitary} quantum interactive proofs, which is the space-bounded variant of the class \QMAM{} introduced in~\cite{MW05}. 

Importantly, the definitions in \Cref{eq:QIPLconst-defs} align with those in \Cref{subsec:QIPL-definition} without loss of generality. Specifically, the two notions of space-bounded quantum interactive proof systems, \QIPL{} and \QIPUL{}, coincide when the number of messages is constant:
\begin{remark}[The equivalence of $\QIPL_{O(1)}$ and $\QIPUL_{O(1)}$]
    \label{remark:QIPLconst-welldefined}
    In a $\QIPL_{O(1)}$ proof system (or its reversible generalization, $\QIPL^{\diamond}_{O(1)}$, see \Cref{remark:reversible-QIPL}), the number of turns (i.e., messages) is a constant. Consequently, the verifier's actions during the execution of the proof system introduce only a constant number of additional environment registers, each containing $O(\log{n})$ qubits. Therefore, a $\QIPL_{O(1)}$ proof system (even its reversible generalization) can be straightforwardly simulated by a $\QIPUL_{O(1)}$ proof system with the same parameters $m(n)$, $c(n)$, and $s(n)$. 
\end{remark}

\vspace{1em}
We establish the following upper bounds for \QIPUL{} and \QIPLconst{}. 
The first theorem, as detailed in \Cref{thm:QIPUL-in-P}, is obtained by combining \Cref{corr:QIPUL-in-QIPUL3SmallGap} and \Cref{lemma:QIPLconst-in-P}, where the former is a direct corollary of the parallelization.\footnote{Importantly, \Cref{thm:QIPUL-in-P} holds only when we define the class \QIPUL{} with the strong notion of uniformity for the verifier's mapping, as detailed in \Cref{subsec:QIPL-definition}. If the verifier's mapping satisfies only a weaker notion of uniformity, as specified in \Cref{footref:logspace-uniformity}, then the verifier's action in the resulting proof system may not be space-bounded. This is because its description might not be produced by a logspace deterministic Turing machine.\label{footref:QIPL-uniformity-in-parallelization}} It is noteworthy that the second inclusion in \Cref{thm:QIPUL-in-P} applies to the case where $m(n) \leq O(1)$ and $c(n) - s(n) \geq 1/\poly(n)$. 

\begin{theorem}[$\QIPUL \subseteq \Ptime$]
    \label{thm:QIPUL-in-P}
    Let $c(n)$, $s(n)$, and $m(n)$ be logspace-computable functions such that $0 \leq s(n) < c(n) \leq 1$, $c(n) \!-\! s(n) \geq 1/\poly(n)$, and $1 \leq m(n) \leq \poly(n)$. Then, it holds that
    \[ \QIPUL_m[c,s] \subseteq \QIPL_3\sbra*{1,1-\frac{1}{q(n)}} \subseteq \Ptime, \text{ where } q(n) \coloneqq \frac{2(m(n)+1)^2}{(c(n) - s(n))^2}. \]
\end{theorem}

The second theorem, as stated in \Cref{thm:QIPLconst-in-NC}, is derived by combining \Cref{corr:QIPLconst-in-QIPL3} and \Cref{lemma:QIPL3-in-NC}. The class \NC{} captures the power of (logspace-uniform) classical poly-logarithmic depth computation using bounded fan-in gates. Specifically, \Cref{corr:QIPLconst-in-QIPL3} is a direct consequence of the parallelization, while the \NC{} containment (\Cref{lemma:QIPL3-in-NC}) follows directly from the celebrated $\QIP=\PSPACE$ result~\cite{JJUW11}.

\begin{theorem}[$\QIPLconst \subseteq \NC$]
    \label{thm:QIPLconst-in-NC}
    Let $c(n)$, $s(n)$, and $m(n)$ be logspace-computable functions such that $0 \leq s(n) < c(n) \leq 1$, $c(n) - s(n) \geq \Omega(1)$, and $1 \leq m(n) \leq O(1)$. Then, it holds that
    \[ \QIPLconst[c,s] = \QMAML \subseteq \NC. \]
\end{theorem}

Unlike standard quantum interactive proofs, constant-message space-bounded quantum interactive proofs with constant promise gap (\QIPLconst{}) are unlikely to be as computationally powerful as their polynomial-message counterparts (\QIPL{}). 
Furthermore, space-bounded quantum interactive proofs (\QIPL{}) appears to be more powerful than their unitary counterparts (\QIPUL{}). 
These distinctions align with the widely believed conjectured separations $\NC \subsetneq \Ptime \subsetneq \NP$.

\vspace{1em}
In the remainder of this section, we first demonstrate error reduction for \QIPLconst{} via parallel repetition in \Cref{subsec:parallel-error-reduction-QIPLconst}. Then, \Cref{subsec:parallelization} provides the parallelization technique for \QIPUL{} and \QIPLconst{}, drawing inspiration from~\cite{KKMV09}, with a focus on the turn-halving lemma (\Cref{lemma:QIPL-halving-parallelization}) and its corollaries. Next, we proceed to establish an upper bound for \QIPLconst{} with weak error bounds in \Cref{subsec:QIPLconst-weakErrorBound-in-P}. Finally, \Cref{subsec:NC-cointainment} presents the equivalence of \QIPLconst{} and \QMAML{} (\Cref{corr:QIPL3-eq-QMAML}) and the \NC{} containment (\Cref{lemma:QIPL3-in-NC}), using a simplified version of the turn-halving lemma. 

\subsection{Error reduction for \QIPLconst{} via parallel repetition}
\label{subsec:parallel-error-reduction-QIPLconst}

Beyond the sequential repetition approach presented in the proof of \Cref{lemma:QIPL-error-reduction}, another common method for error reduction is the parallel repetition of the original proof system $\protocol{P}{V}$. In this approach, $k$ pairs of provers and verifiers execute $\protocol{P}{V}$ in parallel, where all $k$ provers are independent only when they are honest. However, when adapting this approach for \QIPL{} (or \QIPUL{}), the message size in the parallelized protocol $\protocol{P'}{V'}$ becomes $O(k \log n)$, meaning $\protocol{P'}{V'}$ remains a \QIPL{} (or \QIPUL{}) proof system only when $k$ is constant.

Next, we provide the formal statement and its proof inspired by~\cite[Section 4.3]{VW16}: 

\begin{lemma}[Error reduction for \QIPLconst{} via parallel repetition]
    \label{lemma:QIPL-error-reduction-parallel}
    Let $c(n)$, $s(n)$, and $m(n)$ be logspace-computable functions such that $0 \leq s(n) < c(n) \leq 1$, $c(n) - s(n) \geq 1/\poly(n)$, and $1 \leq m(n) \leq O(1)$. For any constant $k(n)$, it holds that
    \[ 
    \QIPL_m[c,s] \subseteq \QIPL_{m}\big[c^k, s^k\big]. 
    \]
\end{lemma}

\begin{proof} 
    For convenience, we prove the inclusion for $\QIPL^{\diamond}_{O(1)}$ proof systems, which serves as a reversible generalization of $\QIPL_{O(1)}$  (see \Cref{remark:reversible-QIPL}) and is equivalent to it (see \Cref{remark:QIPLconst-welldefined}). 
        
    For any $l$-round proof system $\protocol{P}{V}$ with completeness $c$ and soundness $s$, corresponding to a promise problem $\calI$ in $\QIPL^{\diamond}_m[c,s]$, the maximum acceptance probability of the verifier, denoted by $\omega(V)$, serves as the objective function in the SDP formulation provided in \Cref{lemma:QIPL-first-SDP-formulation}. See also \Cref{remark:QIPL-first-SDP-applicability} for the applicability of the SDP to $\QIPL^{\diamond}$ proof systems. 
    
    Now, consider a $k$-fold parallel repetition of $\protocol{P}{V}$, involving $k$ verifiers $V^{(i)} \coloneqq \big( V_1^{(i)}, \cdots, V_{l+1}^{(i)} \big)$ for $i \in [k]$. Let $V^{(1)} \!\otimes\! \cdots \!\otimes\! V^{(k)}$ be the combined verifier, obtained by executing $V^{(1)}, \cdots, V^{(k)}$ in parallel, with the output bit being the AND of the output bits of $V^{(1)},\cdots,V^{(k)}$. It follows that
    \[ \omega\rbra*{V^{(1)} \otimes \cdots \otimes V^{(k)}} \geq \omega\rbra*{V^{(1)}}  \cdots \omega\rbra*{V^{(k)}},\]
    since dishonest provers $P^{(1)}, \cdots, P^{(k)}$ may apply entangled actions. However, due to the strong duality of the SDP program for computing $\omega\rbra*{V^{(i)}}$ for each $i\in[k]$, the equality also holds: 
    \begin{equation}
        \label{eq:QIPL-parallel-repetation-pacc}
        \omega\rbra*{V^{(1)} \otimes \cdots \otimes V^{(k)}} = \omega\rbra*{V^{(1)}}  \cdots \omega\rbra*{V^{(k)}}.
    \end{equation}

    In particular, the SDP formulation specified in \Cref{lemma:QIPL-first-SDP-formulation} serves as the primal form of an SDP program for computing $\omega\rbra*{V^{(i)}}$, with the corresponding dual form obtainable similarly to~\cite[Figure 4.7]{VW16}. Following an argument analogous to Equation (4.49) through (4.53) in~\cite[Section 4.3]{VW16}, we arrive at \Cref{eq:QIPL-parallel-repetation-pacc}, with details omitted here. 
\end{proof}

\subsection{Parallelization via the turn-halving lemma}
\label{subsec:parallelization}

Our approach to do parallelization for $\QIPL_{O(1)}$ is inspired by~\cite[Section 4]{KKMV09}. We begin with the key lemma that halves the number of turns (i.e., messages) in the proof system:

\begin{lemma}[Turn-halving lemma]
    \label{lemma:QIPL-halving-parallelization}
    Let $s(n)$ and $m(n)$ be logspace-computable functions such that $0 \leq s(n) \leq 1$ and $1 \leq m(n) \leq \poly(n)$. Then, it holds that
    \[ \QIPUL_{4m+1}[1,s] \subseteq \QIPUL_{2m+1}\sbra*{1,(1+\sqrt{s})/2}. \]
\end{lemma}

\begin{remark}[Limitations on the turn-halving lemma]
    \label{remark:turn-havling-limitations}
    The parallelization technique described in~\cite[Section 4]{KKMV09} requires the verifier's actions to be \textit{reversible}. This requirement implies that these actions must be implemented using either unitary or isometric quantum circuits. Therefore, it is straightforward to extend \Cref{lemma:QIPL-halving-parallelization} to $\QIPL^{\diamond}$, a reversible generalization of $\QIPL{}$. However, this extended version can be applied recursively at most a \textit{constant} number of times; otherwise, the new verifier in the resulting proof system would no longer be space-bounded. Specifically, if the extended version is applied $\omega(1)$ times, it will introduce $2^{\omega(1)}$ environment registers, each containing $O(\log n)$ qubits, during a single verifier action. 
\end{remark}

\vspace{1em}
Before presenting the proof, we provide two corollaries of the turn-halving lemma (\Cref{lemma:QIPL-halving-parallelization}) with their proofs to illuminate its applications, as stated in \Cref{corr:QIPLconst-in-QIPL3,corr:QIPUL-in-QIPUL3SmallGap}. 

\begin{corollary}[$\QIPLconst \subseteq \QIPL_3$]
    \label{corr:QIPLconst-in-QIPL3}
    Let $c(n)$, $s(n)$, and $m(n)$ be logspace-computable functions with $0 \leq s(n) < c(n) \leq 1$, $c(n) - s(n) \geq \Omega(1)$, and $3 \leq m(n) \leq O(1)$. Then, it holds that 
    \[ \QIPL_m[c,s] \subseteq \QIPL_3[1,1/16]. \]
\end{corollary}

\begin{proof}
    It suffices to prove the following inclusions: 
    \begin{equation}
        \label{eq:QIPLconst-in-QIPL3}
        \QIPL_m[c,s] \subseteq \QIPL_{m+2}\sbra*{1,1-\frac{(c-s)^2}{2}} \subseteq \QIPL_{3}\sbra*{1,1-\frac{(c-s)^2}{2(m+1)^2}} \subseteq \QIPL_{3}\sbra*{1,\frac{1}{16}}.
    \end{equation}

    The first inclusion in \Cref{eq:QIPLconst-in-QIPL3} follows directly from \Cref{lemma:QIPL-perfect-completeness}. 
    To show the last inclusion in \Cref{eq:QIPLconst-in-QIPL3}, we consider an $r$-fold (AND-type) parallel repetition of this three-message \QIPL{} proof system. 
    By applying \Cref{lemma:QIPL-error-reduction-parallel} with $r = O\rbra*{k/\log\frac{1}{1-(c-s)^2/(2(m+1)^2)}}$, we derive the following inclusions:
    \[\QIPL_{3}\sbra*{1,1-\frac{(c-s)^2}{2(m+1)^2}} \subseteq \QIPL_{3}\sbra*{1,\rbra*{1-\frac{(c-s)^2}{2(m+1)^2}}^r} \subseteq \QIPL_{3}\sbra*{1,\frac{1}{16}}.\]

    It remains to show the second inclusion in \Cref{eq:QIPLconst-in-QIPL3}. 
    By repeatedly applying the turn-halving lemma (\Cref{lemma:QIPL-halving-parallelization}) $l$ times, where $l$ satisfies $2^l + 1 \leq m+2 \leq 2^{l+1}+1$, we obtain:
    \begin{equation}
        \label{eq:QIPLm-in-QIPL3}
        \QIPL_{m+2} \sbra*{1,1-\frac{(c-s)^2}{2}} \subseteq \QIPL_{2^{l+1}+1} \sbra*{1,1-\frac{(c-s)^2}{2}} \subseteq \QIPL_{3}\sbra*{1,1-\frac{(c-s)^2}{2(m+1)^2}}.
    \end{equation}
    Here, in the last inclusion, the reasoning behind the parameters -- particularly $m(n)$, $c(n)$, and $s(n)$ -- follows directly from the proof of~\cite[Lemma 4.2]{KKMV09}, so we omit the details. 

    Let $\protocol{P}{V}$ be the original proof system. For the resulting proof system $\protocol{P'}{V'}$, we now need to establish that the verifier's actions are space-bounded unitary circuits, and that the verifier's mapping is logspace computable. 
    This follows because (1) an operation depending on $r$ random coins can be simulated by applying a corresponding unitary controlled by $\ket{+}^{\otimes r}$, meaning that simulating $l$ random coins in all of the verifier's actions requires $l$ ancillary qubits; and (2) since a logspace deterministic Turing machine (DTM) can produce the description of all verifier's actions in $\protocol{P}{V}$, there is another logspace DTM which can produce the description of each verifier's action in $\protocol{P'}{V'}$. 
\end{proof}

\begin{corollary}[\QIPUL{} is parallelized to three messages with weaker error bounds]
    \label{corr:QIPUL-in-QIPUL3SmallGap}
    Let $c(n)$, $s(n)$, and $m(n)$ be logspace-computable functions with $0 \leq s(n) < c(n) \leq 1$, $c(n) - s(n) \geq 1/\poly(n)$, and $3 \leq m(n) \leq \poly(n)$. Then, it holds that 
    \[\QIPUL_m[c,s] \subseteq \QIPUL_3\sbra*{1,1-\frac{1}{q(n)}} \text{ where } q(n) \coloneqq \frac{2(m(n)+1)^2}{(c(n)-s(n))^2}.\] 
\end{corollary}

\begin{proof}
    It suffices to prove the following inclusions: 
    \begin{equation}
        \label{eq:QIPUL-in-QIPUL3}
        \QIPUL_m[c,s] \subseteq \QIPUL_{m+2}\sbra*{1,1-\frac{(c-s)^2}{2}} \subseteq \QIPUL_{3}\sbra*{1,1-\frac{(c-s)^2}{2(m+1)^2}}.
    \end{equation}

    The first inclusion in \Cref{eq:QIPUL-in-QIPUL3} follows directly from \Cref{lemma:QIPL-perfect-completeness}. The second inclusion in \Cref{eq:QIPUL-in-QIPUL3} is achieved using a similar approach as that used to prove \Cref{eq:QIPLm-in-QIPL3} in \Cref{corr:QIPLconst-in-QIPL3}. 
    Finally, we need to show that the verifier's actions in the resulting proof system are space-bounded unitary circuits, and that the verifier's mapping is logspace computable. By noticing that $l = O(\log n)$, we can achieve this using reasoning analogous to the proof of \Cref{corr:QIPLconst-in-QIPL3}, and we omit the details for brevity. 
\end{proof}

\subsubsection{Proof of the turn-halving lemma}
Next, we proceed with the detailed proof of the turn-halving lemma (\Cref{lemma:QIPL-halving-parallelization}):

\begin{proof}[Proof of \Cref{lemma:QIPL-halving-parallelization}]
    Our proof strategy is inspired by the proof of~\cite[Lemma 4.1]{KKMV09}. 
    For any $(4m+1)$-message \QIPUL{} proof system $\protocol{P}{V}$ with completeness $c$ and soundness $s$, corresponding to a promise problem $\calI \in \QIPUL{}$, we can construct a new $(2m+1)$-message proof system $\protocol{\hatP}{\hatV}$.
    We use the notations specified in \Cref{fig:QIPL-odd} to denote the  snapshot states on registers $\sfQ$, $\sfM$, and $\sfW$ (resp., $\widehat{\sfQ}$, $\widehat{\sfM}$, and $\widehat{\sfW}$) during the execution of $\protocol{P}{V}$ (resp., $\protocol{\hatP}{\hatV}$), with slight adjustments for convenience compared to \Cref{fig:QIPL-even} in \Cref{subsec:SDP-formulations-upper-bounds}. 

\begin{figure}[ht!]
    \centering
    \includegraphics[width=\textwidth]{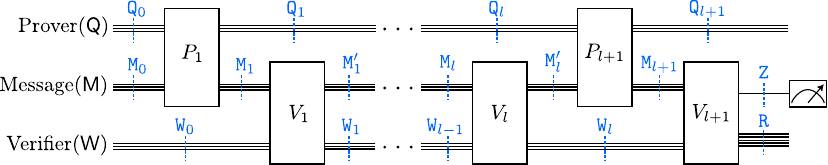}
    \caption{A $(2l+1)$-turn space-bounded unitary quantum interactive proof system.}
    \label{fig:QIPL-odd}
\end{figure}

    We now describe the new proof system $\protocol{\hatP}{\hatV}$ in an intuitive manner: the verifier $\hatV$ first receives the snapshot state $\rho_{\ttM_{m+1}\ttW_{m}}$, which corresponds to the state after the prover $P$ sent the $(2m+1)$-st message in $\protocol{P}{V}$. The verifier $\hatV$ then executes $\protocol{P}{V}$ either forward or backward from the given snapshot, with equal probability. In the forward execution, $\hatV$ accepts if $V$ accepts; while in the backward execution, $\hatV$ accepts if $\widehat{\ttW}_0$ contains the all-zero state.\footnote{Since the proof system $\calA$ begins with the prover $P$, the verifier $\hatV$ does not need to measure $\widehat{\ttM}_0$.} The detailed proof system $\protocol{\hatP}{\hatV}$ is presented in \Cref{protocol:parallelization}. 
    
\begin{algorithm}[ht!]
    \SetAlgorithmName{Protocol}
    \SetEndCharOfAlgoLine{.}
    \SetAlgoVlined
    \setlength{\parskip}{5pt}
    \SetKwFor{While}{}{:}{}
    \SetKwFor{For}{For}{:}{}
    \SetKwInOut{Parameter}{Parameters}

    \textbf{1.} The verifier $\hatV$ receives the snapshot state $\rho_{m+1} =\rho_{\ttM_{m+1}\ttW_{m}}$ from the prover $\hatP$, and then transfers the qubits corresponding to $\ttW_{m}$ in $\rho_{m+1}$ to its private register $\widehat{\sfW}$\;

    \textbf{2.} \While{\textnormal{The verifier $\hatV$ chooses $b \in \binset$ uniformly at random, and executes the original proof system $\protocol{P}{V}$ either forward (if $b=0$) or backward (if $b=1$)}}{
        \textbf{$\bullet$} \While{\textsc{Forward execution} \textnormal{of} $\protocol{P}{V}$ $(b=0)$}{
            \textbf{2.1} $\hatV$ applies $V_{m+1}$ to $(\widehat{\sfM}, \widehat{\sfW})$, and then sends $b$ and $\rho_{\widehat{\sfM}}$ to $\hatP$\;
            \textbf{2.2} \For{$j \leftarrow m+2$ \KwTo $2m$}{
                $\hatV$ receives $\rho_j = \rho_{\ttM_{j}}$ from $\hatP$, applies $V_{j}$ on $(\widehat{\sfM}, \widehat{\sfW})$, and sends $\rho_{\widehat{\sfM}}$ to $\hatV$\;
            }
            \textbf{2.3} $\hatV$ receives $\rho_{2m+1} = \rho_{\ttM_{2m+1}}$ from $\hatP$, applies $V_{2m+1}$ on $(\widehat{\sfM}, \widehat{\sfW})$. $\hatV$ accepts if $(\widehat{\sfM}, \widehat{\sfW})$ contains an accepting state of $\protocol{P}{V}$; otherwise, it rejects\;
        }
        \textbf{$\bullet$} \While{\textsc{Backward execution} \textnormal{of} $\protocol{P}{V}$ $(b=1)$}{
            \textbf{2.a} $\hatV$ sends $b$ and $\rho_{\widehat{\sfM}}=\rho_{\ttM_{m+1}}$ to $\hatP$\;
            \textbf{2.b} \For{$j \leftarrow m$ \KwTo $2$}{
                $\hatV$ receives $\rho_j = \rho_{\ttM_{j}}$ from $\hatP$, applies $V^{\dagger}_{j}$ on $(\widehat{\sfM}, \widehat{\sfW})$, and sends $\rho_{\widehat{\sfM}}$ to $\hatV$\;
            }
            \textbf{2.c} $\hatV$ receives $\rho_{1} = \rho_{\ttM_{1}}$ from $\hatP$, applies $V^{\dagger}_{1}$ on $(\widehat{\sfM}, \widehat{\sfW})$. $\hatV$ accepts if $\widehat{\sfW}$ contains the all-zero state; otherwise, it rejects\;
        }
    }
    \BlankLine
	\caption{A \QIPL{} proof system for halving the number of messages in $\protocol{P}{V}$.}
	\label[algorithm]{protocol:parallelization}
\end{algorithm}

    \vspace{1em}
    It remains to establish the correctness of the proof system $\protocol{\hatP}{\hatV}$. Since the verifier $V$'s mapping (of $\protocol{P}{V}$) is logspace computable, and in each turn, the verifier $\hatV$'s action corresponds to one of two possible actions of $V$ (depending on $b$),\footnote{The random bit $b$ can be simulated by a fresh qubit and a pinching intermediate measurement, following an argument similar to the proof of \Cref{thm:NP-in-QIPL} concerning Step 1 in \Cref{protocol:QIPL-3SAT-protocol}}, the verifier $\hatV$'s mapping is likewise logspace computable. Moreover, we need the following analyses: 
    \begin{itemize}[leftmargin=1.5em]
        \item For \textit{yes} instances, an honest prover $\hatP$ can prepare the pure state $\ket{\psi_{\ttQ_{m+1}\ttM_{m+1}\ttW_m}}$, which corresponds to the state in $(\sfQ,\sfM,\sfW)$ after the $(2m+1)$-st turn in $\protocol{P}{V}$. Depending on the value of $b$, the prover $\hatP$ then applies the corresponding prover's actions $P_j$ (if $b=0$) or $P^{\dagger}_j$ (if $b=1$) in $\protocol{P}{V}$ during the execution of Step 2 in \Cref{protocol:parallelization}. For any $x\in\calI_{\yes}$, as $(\protocol{P}{V})(x)$ accepts with certainty, it follows that $(\protocol{\hatP}{\hatV})(x)$ also accepts with certainty. 
        
        \item For \textit{no} instances, let $\ket{\psi}$ be the state in $(\widehat{\sfQ},\widehat{\sfM},\widehat{\sfW})$ just after the first turn in $\protocol{\widetilde{P}}{\hatV}$. Let $\widetilde{P}^{(b)}_j$ be the prover $\widetilde{P}$'s action, which is an arbitrary unitary transformation on $(\widehat{\sfQ},\widehat{\sfM})$, at the $(2j-1)$-st turn for $2 \leq j \leq m+1$. We can then define unitary transformations $U^{(0)}$ and $U^{(1)}$ corresponding to the forward and backward execution of $\protocol{P}{V}$, respectively: 
        \begin{equation}
            \label{eq:parallelization-unitaries-def}
            U^{(0)} \coloneqq V_{2m+1} \widetilde{P}^{(0)}_{m+1} V_{2m} \cdots \widetilde{P}^{(0)}_2 \text{ and } U^{(1)} \coloneqq V_1^{\dagger} \widetilde{P}_{m+1} \cdots V_m^{\dagger} \widetilde{P}_2. 
        \end{equation}

        Based on \Cref{eq:parallelization-unitaries-def}, we define the snapshot state $\ket{\Psi^{(b)}}$ at Step 2.3 for $b=0$ and at Step 2.c for $b=1$, respectively, before the corresponding final measurement:
        \begin{align*}
            &\ket{\Psi^{(0)}} \coloneqq \tfrac{1}{\sqrt{\pacc^{(0)}}} \Pi_\acc^{(0)} U^{(0)} \ket{\psi}, \text{ where } \pacc^{(0)} \coloneqq \norm*{\Pi_\acc^{(0)} U^{(0)} \ket{\psi}}^2_2 \text{ and } \Pi_\acc^{(0)} \coloneqq \ketbra{1}{1}_{\Out} = \Pi_\acc,\\
            &\ket{\Psi^{(1)}} \coloneqq \tfrac{1}{\sqrt{\pacc^{(1)}}} \Pi_\acc^{(1)} U^{(1)} \ket{\psi}, \text{ where } \pacc^{(1)} \coloneqq \norm*{\Pi_\acc^{(1)} U^{(1)} \ket{\psi}}^2_2 \text{ and } \Pi_\acc^{(1)} \coloneqq \ketbra{\bar{0}}{\bar{0}}_{\widehat{\sfW}}.
        \end{align*}

        As a result, the acceptance probability $\pacc^{(b)}$ for $b\in\binset$ can be expressed as: 
        \begin{equation}
            \label{eq:parallelization-pacc}
            \forall b\in\binset,~\pacc^{(b)} = \tfrac{1}{\norm*{\Pi_\acc^{(b)} U^{(b)} \ket{\psi}}^2_2} \abs*{ \braket{\psi}{{U^{(b)}}^{\dagger} \Pi_\acc^{(b)} U^{(b)}}{\psi} }^2
            = \abs*{\braket{\psi}{{U^{(b)}}^{\dagger}}{\Psi^{(b)}}}^2.
        \end{equation}

        Note that the acceptance probability of the proof system $(\protocol{\widetilde{P}}{\widehat{V}})(x)$ for $x \in \calI_{\no}$ can be written as $\pacc = \frac{1}{2} \big( \pacc^{(0)} + \pacc^{(1)} \big)$. Substituting \Cref{eq:parallelization-pacc} into this equality, we obtain:
        \begin{align*}
            \pacc &= \frac{1}{2}\rbra*{\F\rbra*{{U^{(0)}}^{\dagger} \ketbra{\Psi^{(0)}}{\Psi^{(0)}} U^{(0)}, \ketbra{\psi}{\psi}}^2 + \F\rbra*{{U^{(1)}}^{\dagger} \ketbra{\Psi^{(1)}}{\Psi^{(1)}} U^{(1)}, \ketbra{\psi}{\psi}}^2}\\
            &\leq \frac{1}{2} \rbra*{1 + \F\rbra*{{U^{(0)}}^{\dagger} \ketbra{\Psi^{(0)}}{\Psi^{(0)}} U^{(0)}, {U^{(1)}}^{\dagger} \ketbra{\Psi^{(1)}}{\Psi^{(1)}} U^{(1)}}}\\
            &\leq \frac{1}{2} \rbra*{1 + \norm*{\Pi_\acc U^{(0)}{U^{(1)}}^{\dagger}\ket{\Psi^{(1)}}}_2}\\
            &\leq \frac{1}{2} (1+\sqrt{s}).
        \end{align*}
        Here, the second line follows from \Cref{lemma:sum-of-squared-fidelity}, and the last line is due to the fact that $\pacc = \norm*{\Pi_\acc U^{(0)}{U^{(1)}}^{\dagger}\ket{\Psi^{(1)}}}^2_2 \leq s$, as guaranteed by the soundness condition of $\protocol{P}{V}$. \qedhere
    \end{itemize}    
\end{proof}

\subsection{Weakness of \QIPLconst{} with weak error bounds}
\label{subsec:QIPLconst-weakErrorBound-in-P}

We demonstrate an upper bound for \QIPLconst{} with weak error bounds: 
\begin{lemma}[\QIPLconst{} is in \Ptime{}]
    \label{lemma:QIPLconst-in-P}
    Let $c(n)$, $s(n)$, and $m(n)$ be logspace-computable functions such that $0 \leq s(n) < c(n) \leq 1$, $c(n) - s(n) \geq 1/\poly(n)$, and $1 \leq m(n) \leq O(1)$, it holds that 
    \[\QIPL_m[c,s] \subseteq \Ptime.\]
\end{lemma}

Our approach parallels the proof of $\QIP \subseteq \EXP$, as outlined in~\cite[Page 31]{Watrous16tutorial} and originally established in~\cite[Section 6]{KW00}. Specifically, the SDP program for characterizing constant-turn \QIPL{} proof systems, as detailed in \Cref{lemma:QIPL-first-SDP-formulation}, has a dimension of $\poly(n)$. Therefore, we conclude a deterministic polynomial-time algorithm using a standard SDP solver. 

\begin{proof}[Proof of \Cref{lemma:QIPLconst-in-P}]
    For any $m$-turn proof system $\protocol{P}{V}$, with completeness $c$, soundness $s$, and $m$ being even, which corresponds to a promise problem $\calI$ in $\QIPL_m[c,s]$, we can utilize \Cref{lemma:QIPL-first-SDP-formulation} to obtain an SDP program, as specified in \Cref{eq:QIPL-first-SDP}.
    This SDP program maximizes the verifier's maximum acceptance probability $\omega(V)$ over all choices of quantum states (variables) $\rho_{\ttM'_1\ttW_1\sfE_1}, \cdots, \rho_{\ttM'_l\ttW_l\sfE_1\cdots\sfE_l}$, where the number of rounds $l \coloneqq m/2$. 
    Since the number of turns $m(n) \leq O(1)$, the variables in this SDP collectively hold $O(\log{n})$ qubits. Thus, we can compute a description of this SDP program in deterministic polynomial time. 
    
    Next, consider the Frobenius norm defined as $\norm{X}_{\rm F} \coloneqq \sqrt{\Tr\rbra*{X^\dagger X}}$, and let $\{\sigma_i(X)\}$ be the singular values of a square matrix $X$. We then have the following: 
    \begin{align*}
        \norm{ \rho_{\ttM'_1\ttW_1\sfE_1} \otimes \cdots \otimes \rho_{\ttM'_l\ttW_l\sfE_1\cdots\sfE_l} }_{\rm F} 
        &=\sqrt{\Tr\rbra*{\rho_{\ttM'_1\ttW_1\sfE_1}^2 \otimes \cdots \otimes \rho_{\ttM'_l\ttW_l\sfE_1\cdots\sfE_l}^2}}\\
        &=\sqrt{ \sum_{i=1}^D \sigma^2_i\rbra*{ \rho_{\ttM'_1\ttW_1\sfE_1} \otimes \cdots \otimes \rho_{\ttM'_l\ttW_l\sfE_1\cdots\sfE_l} } }\\
        &\leq \sqrt{D}. 
    \end{align*}
    Here, $D$ represents the dimension of $\rho_{\ttM'_1\ttW_1\sfE_1} \otimes \cdots \otimes \rho_{\ttM'_l\ttW_l\sfE_1\cdots\sfE_l}$, which is bounded by $2^{O(\log n)}$, indicating that it is in $\poly(n)$. The last line follows from the fact that all singular values of the density matrix $\rho_{\ttM'_1\ttW_1\sfE_1} \otimes \cdots \otimes \rho_{\ttM'_l\ttW_l\sfE_1\cdots\sfE_l}$ are at most $1$. 
    
    As a consequence, by employing the standard SDP solver based on the ellipsoid method  (e.g.,~\cite[Theorem 2.6.1]{GM12}; see also~\cite[Chapter 3]{GLS93}), we obtain an  algorithm for approximately solving the SDP program in \Cref{eq:QIPL-first-SDP}, ensuring that the condition $\omega(V) \geq c(n)$ is satisfied. This algorithm runs in deterministic time $\poly(D) \cdot \polylog\big(\sqrt{D}/\varepsilon\big)$, or expressed as $\poly(D,\log(1/\varepsilon))$, outputting either an $\varepsilon$-approximate feasible solution $\hat{X}$ or a certificate indicating that no such solution exists. Particularly, the error parameter $\varepsilon(n)$ ensures that $\norm{\hat{X} - X}_{\rm F} \leq \varepsilon(n)$ for some feasible solution $X$, with $\omega(V)|_{\hat{X}} \geq c(n) - \varepsilon(n)$, where $\omega(V)|_{\hat{X}}$ represents the objective function evaluated at $\hat{X}$. 
    We conclude the proof by observing that $\varepsilon(n) \leq 1/\poly(n)$ holds as long as $c(n)-s(n) \geq 1/\poly(n)$. 
\end{proof}

\subsection{Weakness of \QIPLconst{}: \QMAML{} and \NC{} containment}
\label{subsec:NC-cointainment}

We present an upper bound of \QIPLconst{}: 

\begin{lemma}
    \label{lemma:QIPL3-in-NC}
    $\QIPL_3[1,1/16] \subseteq \QMAML^{\odot}[1,5/8] \subseteq \NC$. 
\end{lemma}

The proof of \Cref{lemma:QIPL3-in-NC} relies crucially on the \textit{single-coin} variant of public-coin three-message quantum interactive proofs $\QMAML^{\odot}$, where the second message is a single random coin. 
Specifically, the first inclusion corresponds to a space-bounded variant of $\QIP(3) \subseteq \QMAM$~\cite[Section 5]{MW05}, and the turn-halving lemma (\Cref{lemma:QIPL-halving-parallelization}) naturally extends this result. 
The second inclusion is exactly a down-scaling version of $\QMAM^{\odot} \subseteq \NC(\poly)$~\cite{JJUW11}.\footnote{This inclusion does not extend to the variant with $1/\poly(n)$ promise gap. Specifically, applying the parallel SDP solver in~\cite{JJUW11}, the resulting algorithm runs in parallel time (i.e., circuit depth) $\poly\log(n) \cdot \poly(1/\epsilon)$, using $\poly(n)$ processors (i.e., circuit width), as noted in the first paragraph of~\cite{JY11}. Since the parameter $\varepsilon$ for $\QMAML^{\odot}[1,1-1/p(n)]$ is polynomially small, this algorithm does not run in \NC{} as the parallel running time becomes $\poly(n)$. This issue also arises when applying width-independent parallel SDP solvers~\cite{JY11,AZLO16}.\label{footref:precision-in-parallel-SDP-solvers}}

Furthermore, the first inclusion in \Cref{lemma:QIPL3-in-NC} implies the following:\footnote{The inclusion $\QMAML^{\odot}[1,5/8] \subseteq \QMAML[1,125/512] \subseteq \QMAML[1,1/3]$ is obtained by applying three-fold parallel repetition (\Cref{lemma:QIPL-error-reduction-parallel}), where the second message in the resulting proof system is three random coins. } 

\begin{corollary}
    \label{corr:QIPL3-eq-QMAML}
    $\QIPL_3 = \QMAML$. 
\end{corollary}

We now move to the proof of \Cref{lemma:QIPL3-in-NC}. 

\begin{proof}[Proof of \Cref{lemma:QIPL3-in-NC}]
    We begin by proving the first inclusion using a variant of the turn-halving lemma (\Cref{lemma:QIPL-halving-parallelization}). Let $\protocol{P}{V}$ denote the original $\QIPL_3$ proof system, which acts on registers $\sfQ$, $\sfM$, and $\sfW$, following the notations in \Cref{fig:QIPL-odd} with $l=1$. We propose a $\QMAML^{\odot}$ proof system $\protocol{\hatP}{\hatV}$, which acts on registers $\widehat{\sfQ}$, $\widehat{\sfM}$, and $\widehat{\sfW}$, as described in \Cref{protocol:single-coin-QMAML}. It is noteworthy that this proof system is a simplified version of \Cref{protocol:parallelization}. 
    
\begin{algorithm}[ht!]
    \SetAlgorithmName{Protocol}
    \SetEndCharOfAlgoLine{.}
    \SetAlgoVlined
    \setlength{\parskip}{5pt}
    \SetKwFor{While}{}{:}{}
    \SetKwIF{If}{}{}{If}{:}{}{}{}
    \SetKwInOut{Parameter}{Parameters}

    \textbf{1.} The verifier $\hatV$ receives the qubits contained in $\ttW_1$ 
    from the prover $\hatP$, and then transfers 
    them to $\widehat{\sfW}$.

    \textbf{2.} The verifier $\hatV$ chooses $b \in \binset$ uniformly at random and sends $b$ to the prover $\hatP$. 

    \textbf{3.} \While{\textnormal{The verifier $\hatV$ receives the qubits written in $\widehat{\sfM}$ from the prover $\hatP$}}{
        \textbf{$\bullet$} If $b=0$, the verifier $\hatV$ applies $V_2$ on $(\widehat{\sfM},\widehat{\sfW})$. $\hatV$ accepts if  $(\widehat{\sfM},\widehat{\sfW})$ contains an accepting state of $\protocol{P}{V}$, and rejects otherwise.
        
        \textbf{$\bullet$} If $b=1$, the verifier $\hatV$ applies $V_1^{\dagger}$ on $(\widehat{\sfM},\widehat{\sfW})$. $\hatV$ accepts if $\widehat{\sfW}$ contains the all-zero state, and rejects otherwise. 
    }
    \BlankLine
	\caption{A $\QMAML^{\odot}$ proof system for verifying a $\QIPL_3$ proof system $\protocol{P}{V}$.}
	\label[algorithm]{protocol:single-coin-QMAML}
\end{algorithm}

    It remains to establish the correctness of \Cref{protocol:single-coin-QMAML}. This is straightforward for \textit{yes} instances. For \textit{no} instances, the desired bound essentially follows from the inequality in \Cref{lemma:sum-of-squared-fidelity}, using reasoning similar to that in the proof of \Cref{lemma:QIPL-halving-parallelization}. We omit the details. 

    Next, we address the second inclusion. We start by observing that \Cref{protocol:single-coin-QMAML} aligns with the definition of single-coin quantum Arthur-Merlin games as described in~\cite[Section 2.4]{JJUW11}, with two key differences: the message length $m$ is $\log(n)$ rather than $\poly(n)$, and the verifier is space-bounded instead of polynomial-time bounded. This proof system achieves completeness $1$ and soundness $5/8$ due to the first inclusion. Consequently, we can obtain the corresponding primal-dual SDP programs of dimension $\poly(n)$, as opposed to $\exp(\poly(n))$, following~\cite[Section 2.5]{JJUW11}. Therefore, we conclude an \NC{} containment by applying the parallel SDP solver from~\cite{JJUW11} to the resulting SDP programs of dimension $\poly(n)$. 
\end{proof}

\section{Space-bounded unitary quantum statistical zero-knowledge}

We now introduce (honest-verifier) space-bounded unitary quantum statistical zero-knowledge, denoted as \QSZKUL{} and \QSZKULHV{}, as specific types of space-bounded unitary quantum interactive proofs (\QIPUL{}) that possess an additional statistical zero-knowledge property.

Before presenting our results, we start by defining the promise problem \IndivProdQSD{}, which is analogous to \QSD{}~\cite{Wat02QSZK} and \GapQSDlog{}~\cite{LGLW23}: 
\begin{definition}[Individual Product State Distinguishability Problem, {$\IndivProdQSD[k,\alpha,\delta]$}]
    \label{def:IndivProdQSD}
    Let $k(n)$, $\alpha(n)$, $\delta(n)$, and $r(n)$ be logspace computable functions such that $1 \leq k(n) \leq \poly(n)$, $0 \leq \alpha(n),\delta(n) \leq 1$, $\alpha(n) - \delta(n) \cdot k(n) \geq 1/\!\poly(n)$, and $1 \leq r(n)\leq O(\log n)$. Let $Q_1, \cdots, Q_k$ and $Q'_1, \cdots, Q'_k$ be polynomial-size unitary quantum circuits acting on $O(\log n)$ qubits, each with $r(n)$ specified output qubits. For $j\in[k]$, let $\sigma_j$ and $\sigma'_j$ denote the states obtained by running $Q_j$ and $Q'_j$ on the all-zero state $\ket{\bar{0}}$, respectively, and tracing out the non-output qubits, then the promise is that one of the following holds:    
    \begin{itemize}[itemsep=0.33em,topsep=0.33em,parsep=0.33em]
        \item \emph{Yes} instances: Two $k$-tuples of quantum circuits $(Q_1,\cdots,Q_k)$ and $(Q'_1,\cdots,Q'_k)$ such that 
        \[\td\rbra*{\sigma_1\otimes \cdots \otimes \sigma_k, \sigma'_1\otimes \cdots \otimes \sigma'_k} \geq \alpha(n);\]
        \item \emph{No} instances: Two $k$-tuples of quantum circuits $(Q_1,\cdots,Q_k)$ and $(Q'_1,\cdots,Q'_k)$ such that 
        \[ \forall j \in [k], \quad \td\rbra*{\sigma_j, \sigma'_j} \leq \delta(n). \]
    \end{itemize}
\end{definition}
Additionally, we denote the \textit{complement} of $\IndivProdQSD[k(n),\alpha(n),\delta(n)]$, with respect to the chosen parameters $\alpha(n)$, $\delta(n)$, and $k(n)$, as \coIndivProdQSD{}. 

\vspace{1em}
With these definitions in hand, we now provide our first theorem in this section:
\begin{theorem}[The equivalence of \QSZKUL{} and \BQL{}]
    \label{thm:QSZKL-eq-BQL}
    The following holds: 
    \begin{enumerate}[label={\upshape(\arabic*)}, itemsep=0.33em, topsep=0.33em, parsep=0.33em]
    \item For any logspace-computable function $m(n)$ such that $1 \leq m(n) \leq \poly(n)$, 
    \[ \cup_{c(n)-s(n) \geq 1/\!\poly(n)} \QSZKULHV[m,c,s] \subseteq \BQL. \]
    \item $\BQL \subseteq \QSZKUL \subseteq \QSZKULHV$. \label{thmitem:BQL-in-QSZKUL}
    \end{enumerate}
\end{theorem}

The class \QSZKUL{} consists of space-bounded unitary quantum interactive proof systems that possess statistical zero-knowledge against \textit{any} verifier, whereas \QSZKL{} proof systems possess statistical zero-knowledge against only an \textit{honest} verifier. 
Consequently, the inclusion in \Cref{thm:QSZKL-eq-BQL}\ref{thmitem:BQL-in-QSZKUL} is straightforward, following directly from these definitions. 
To establish the direction $\QSZKULHV{} \subseteq \BQL{}$, we proceed by proving the following: 

\begin{theorem}[\IndivProdQSD{} is \QSZKULHV{}-complete]
    \label{thm:IndivProdQSD-QSZKLcomplete}
    The following holds: 
    \begin{enumerate}[label={\upshape(\arabic*)}, itemsep=0.33em, topsep=0.33em, parsep=0.33em]
        \item Let $c(n)$ and $s(n)$ be logspace computable functions such that $0 \leq s(n) < c(n) \leq 1$. For any logspace-computable function $m(n)$ such that $3 \leq m(n) \leq \poly(n)$, \label{thmitem:IndivProdQSD-QSZKLhard}
        \[ \coIndivProdQSD[m/2,\alpha,2\delta] \text{ is } \QSZKULHV[m,c,s]\text{-hard}.\]
        Here, $\alpha \coloneqq (\sqrt{c}-\sqrt{s})^2/(2m-4)$ and $\delta$ is some negligible function. 
        \item  Let $k(n)$, $\alpha(n)$ and $\delta(n)$ be logspace computable functions such that $1 \leq k(n) \leq \poly(n)$, $0 \leq \alpha(n),\delta(n) \leq 1$, and $\alpha(n) - \delta(n) \cdot k(n) \geq 1/\poly(n)$. Then, it holds that
        \label{thmitem:IndivProdQSD-in-QSZKL}
        \[ \IndivProdQSD[k,\alpha,\delta] \in \BQL \subseteq \QSZKULHV. \]
    \end{enumerate}
\end{theorem}

\vspace{1em}
In the remainder of this section, we first provide the definition of honest-verifier space-bounded quantum statistical zero-knowledge proofs (the class \QSZKULHV{}) in \Cref{subsec:QSZKL-definition}. Next, we establish that \IndivProdQSD{} is \QSZKULHV{}-hard (\Cref{thm:IndivProdQSD-QSZKLcomplete}\ref{thmitem:IndivProdQSD-QSZKLhard}) in \Cref{subsec:IndivProdQSD-QSZKLhard}. Subsequently, we present the \BQL{} upper bound for \QSZKULHV{} (\Cref{thm:IndivProdQSD-QSZKLcomplete}\ref{thmitem:IndivProdQSD-in-QSZKL}) in \Cref{subsec:QSZKL-in-BQL}. 

\subsection{Definition of space-bounded unitary quantum statistical zero-knowledge}
\label{subsec:QSZKL-definition}

Our definition of (honest-verifier) space-bounded quantum statistical zero-knowledge follows that of~\cite[Section 3.1]{Wat02QSZK}. In this framework, an honest-verifier space-bounded unitary quantum statistical zero-knowledge proof system is a space-bounded unitary quantum interactive proof system, as defined in~\Cref{subsec:QIPL-definition}, that satisfies an additional \textit{zero-knowledge} property. Intuitively, the zero-knowledge property in \QIPUL{} proof systems requires that, after each message is sent, the quantum states representing the verifier's view -- including snapshot states in the message register $\sfM$ and the verifier's private register $\sfW$ -- should be approximately indistinguishable by a space-bounded unitary quantum circuit on accepted inputs. 

We then formalize this notion. Consider a set $\{\rho_{x,i}\}$ of mixed states, we say that this state set is \textit{logspace-preparable} if there exists a family of $m$-tuples $S_x \coloneqq \rbra*{S_{x,1},\cdots,S_{x,m}}$, where each $S_{x,i}$ for $i \in [m]$ is a space-bounded unitary quantum circuit (see \Cref{def:space-bounded-quantum-circuits}) with a specified collection of output qubits, such that for each input $x$ and index $i$, the state $\rho_{x,i}$ is the mixed state obtained by running $S_{x,i}$ on the input state $\ket{\bar{0}}$, and then tracing out all non-output qubits. We refer to such $\{S_x\}_{x\in\calI}$ as the \textit{space-bounded simulator} for the promise problem $\calI$. 

Next, for any space-bounded quantum interactive proof system $\protocol{P}{V}$, we define the verifier's view after the $i$-th turn, denoted by $\view{P}{V}(x,i)$, as the reduced state in registers $(\sfM,\sfW)$ immediately after $i$ messages have been exchanged, with the prover's private qubits traced out. 

We are now ready for the formal definition:
\begin{definition}[Honest-verifier space-bounded unitary quantum statistical zero-knowledge, \QSZKULHV{}]
    \label{def:QSZKUL}
    Let $c(n)$, $s(n)$, and $m(n)$ be logspace-computable functions of the input length $n \coloneqq |x|$ such that $0 \leq s(n) < c(n) \leq 1$ and $1 \leq m(n) \leq \poly(n)$. A promise problem $\calI = (\calI_{\yes}, \calI_{\no})$ is in $\QSZKULHV{}[m,c,s]$, if there exists an $m(n)$-message space-bounded unitary quantum interactive proof system $(\protocol{P}{V})(x)$ such that\emph{:}    
    \begin{itemize}[topsep=0.33em, itemsep=0.33em, parsep=0.33em]
        \item \textbf{\emph{Completeness}}. For any $x \in \calI_{\yes}$, there exists an $m(n)$-message prover $P$ such that  
        \[\Pr{(\protocol{P}{V})(x) \text{ accepts}} \geq c(n).\] 
        \item \textbf{\emph{Soundness}}. For any $x \in \calI_{\no}$ and any $m(n)$-message prover $P$, 
        \[\Pr{(\protocol{P}{V})(x) \text{ accepts}} \leq s(n).\] 
        \item \textbf{\emph{Zero-knowledge}}. There exists a space-bounded simulator $\{S_x\}_{x\in\calI}$ and a negligible function $\delta(n)$ such that for any $x\in \calI_{\yes}$ and each message $i \in [m]$, the circuit $S_x(i)$ produces the corresponding state $\sigma_{x,i}$ satisfying         
        \[ \td\rbra*{\sigma_{x,i},\view{P}{V}(x,i)} \leq \delta(n).\]
    \end{itemize}
    We define $\QSZKULHV[m] \coloneqq \QSZKULHV\big[m,\frac{2}{3},\frac{1}{3}\big]$ and $\QSZKULHV \coloneqq \cup_{m \leq \poly(n)} \QSZKULHV[m]$.
\end{definition}

Since the inequality condition in the zero-knowledge property holds independently for each message in \Cref{def:QSZKUL}, error reduction via sequential repetition (\Cref{lemma:QIPL-error-reduction}) directly applies to an honest-verifier space-bounded quantum statistical zero-knowledge proof system, with the zero-knowledge property automatically preserved. 

\begin{remark}[Robustness of the zero-knowledge property in \QSZKULHV{}]
    \label{QSZKL-ZKproperty-robustness}
    Let $\QSZKULstar$ denote a weaker version of \QSZKULHV{}, where the threshold function $\delta(n) \coloneqq \rbra*{\sqrt{c}-\sqrt{s}}^2/\rbra*{2m^2}$,\footnote{This bound results from the reduction to the \QSZKULHV{}-hard problem \IndivProdQSD{}, see \Cref{thm:IndivProdQSD-QSZKLhard}.} rather than being negligible. While it is clear that $\QSZKULHV \subseteq \QSZKULstar$, the standard approach to establish the reverse direction does not apply to $\QSZKULHV$.\footnote{In particular, the polarization lemma for the trace distance~\cite[Section 4.1]{Wat02QSZK} is not applicable in the space-bounded scenario due to message size constraints.} Instead, the inclusion $\QSZKULstar \subseteq \QSZKULHV$ only follows from $\QSZKULstar = \BQL$ (\Cref{thm:QSZKL-eq-BQL}).
\end{remark}

\subsection{\coIndivProdQSD{} is \QSZKULHV{}-hard}
\label{subsec:IndivProdQSD-QSZKLhard}

Instead of directly proving that \coIndivProdQSD{} is \QSZKULHV{}-hard, we establish a slightly stronger result: the promise problem \coIndivProdQSD{} is hard for the class $\QSZKULstar$ that contains \QSZKULHV{} (\Cref{QSZKL-ZKproperty-robustness}), as detailed in \Cref{thm:IndivProdQSD-QSZKLhard}. This result mirrors the relationship between \coQSD{} and the class \QSZK{}. 
\begin{theorem}[\coIndivProdQSD{} is $\QSZKULstar$-hard]
    \label{thm:IndivProdQSD-QSZKLhard}
    Let $c(n)$, $s(n)$, and $m(n)$ be logspace computable functions such that $0 \leq s(n) < c(n) \leq 1$, $c(n) - s(n) \geq 1/\poly(n)$, and $3 \leq m(n) \leq \poly(n)$.\footnote{Without loss of generality, we can assume that $m \geq 3$ by adding one or two dummy messages when $m < 3$, as discussed in \Cref{footref:dummy-message}.} Then, it holds that 
    \[ \coIndivProdQSD\sbra*{\ceil*{m(n)/2}, \alpha(n), 2\delta(n)} \text{ is } \QSZKULstar[m(n),c(n),s(n)]\text{-hard}.\]
    Here, $\delta \coloneqq (\sqrt{c}-\sqrt{s})^2/(2m^2)$ and $\alpha \coloneqq (\sqrt{c}-\sqrt{s})^2/(2m-4)$. 
\end{theorem}

Before presenting the proof, we will first illustrate the properties of the simulator and explain the underlying intuition behind the proof. Our proof strategy follows some ideas from~\cite[Section 5]{Wat02QSZK}. Consider a space-bounded quantum interactive proof system $\protocol{P}{V}$ for a promise problem $\calI \in \QSZKULstar[m(n),c(n),s(n)]$ that is statistical zero-knowledge against an honest verifier. Without loss of generality, assume that the number of turns in $\protocol{P}{V}$ is even. We use the notations introduced in \Cref{fig:QIPL-even} and \Cref{subsec:SDP-formulations-upper-bounds}. 

Let us now focus on the space-bounded simulator $\{S_x\}_{x\in\calI}$. Let $\xi'_0,\cdots,\xi'_l$ and $\xi_1,\cdots,\xi_{l+1}$ denote the simulator's approximation to the reduced snapshot states in registers $(\sfM,\sfW)$ after the $(2j-1)$-st and the $(2j)$-th turn, respectively, during the execution of $\protocol{P}{V}$, as specified in \Cref{fig:QSZKL-simulators}. For \textit{yes} instances, these states closely approximate the actual view of the verifier (the corresponding snapshot states) during the execution of $\protocol{P}{V}$. However, there is no \textit{direct} closeness guarantee for \textit{no} instances. Consequently, we can assume that the state $\xi_{l+1}$ satisfies $\Tr\rbra*{\ketbra{1}{1}_{\ttZ} \xi_{l+1})} = c(n)$ for \textit{all} instances. 

\begin{figure}[ht!]
    \centering
    \includegraphics[width=\textwidth]{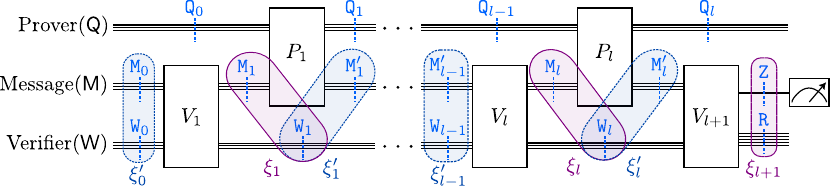}
    \caption{Quantum states $\xi'_0,\cdots,\xi'_l$ and $\xi,\cdots,\xi_{l+1}$ prepared by the simulator.}
    \label{fig:QSZKL-simulators}
\end{figure}

In addition, given that the verifier is always assumed to act honestly, we can take\footnote{Consequently, the simulator only needs to prepare $\xi'_{j-1}$, since $\xi_j$ is obtained by applying $V_j$ to this state.}  
\begin{equation}
    \label{eq:simulator-states}
    \xi'_0 = (\ket{0}\bra{0})^{\otimes (q_{\sfM}+q_{\sfW})} \text{ and } \xi_{j} = V_j \xi'_{j-1} V^\dagger_j \text{ for } j \in [l+1]. 
\end{equation}

\paragraph{Proof intuition.} Notably, the space-bounded simulator $\{S_x\}_{x\in\calI}$ essentially produces an approximation solution, in the form of snapshot states, to the SDP program \Cref{eq:QIPL-first-SDP} for computing the maximum acceptance probability $\omega(V)$ of the space-bounded unitary quantum interactive proof systems $\protocol{P}{V}$ for $\calI \in \QIPUL{}$. As we stated in the proof of \Cref{lemma:QIPL-first-SDP-formulation}, there are only two types of constraints: (1) Verifier's actions are honest; and (2) Prover's actions do not affect the verifier's private qubits. 

As mentioned in \Cref{eq:simulator-states}, these states produced by the simulator exactly satisfy the first type of constraints for all instances, but satisfy the second type of constraints \textit{only} for \textit{yes} instances. This observation leads to our proof and the hard problem \IndivProdQSD{}. Specifically, we consider two tensor product states, each consisting of a polynomial number of $O(\log n)$-qubit states, where all components are defined in \Cref{fig:QSZKL-simulators}: 
\begin{equation}
    \label{eq:two-states}
    \Tr_{\sfM}(\xi_1)\otimes \cdots \otimes \Tr_{\sfM}(\xi_l) \text{ and } \Tr_{\sfM}\rbra*{\xi'_1}\otimes \cdots \otimes \Tr_{\sfM}\rbra*{\xi'_l}.
\end{equation}

For \textit{yes} instances, the zero-knowledge property ensures a component-wise closeness bound $\Tr_{\sfM}\rbra*{\xi_j} \approx \Tr_{\sfM}(\xi'_j)$ for $j\in[l]$. 
For \textit{no} instances, we need to show that the two states in \Cref{eq:two-states} are far from each other, given that $\omega(V) \leq s(n)$. This follows directly from~\cite[Lemma 15]{Wat02QSZK}. We state the counterpart result below and omit the detailed proof: 
\begin{proposition}[Adapted from~{\cite[Lemma 15]{Wat02QSZK}}]
    \label{prop:QSZKLhard-soundness}
    Let $\protocol{P}{V}$ be an $m(n)$-turn space-bounded quantum interactive proof system, with even $m\coloneq 2l$, such that $\omega(V) \leq s(n)$. Let $\xi'_0,\cdots,\xi'_l$ and $\xi_k,\cdots,\xi_{l+1}$ be the states produced by the simulators as defined in \Cref{fig:QSZKL-simulators}. Assume that $\Tr\rbra*{\ketbra{\bar{0}}{\bar{0}}_{\ttM_0\ttW_0} \xi'_0} = 1$ and $\Tr\rbra*{\ketbra{1}{1}_{\ttZ} \xi_{l+1}} = c$. Then, it holds that
    \[ \td\rbra*{\Tr_{\sfM}(\xi_1)\otimes \cdots \otimes \Tr_{\sfM}(\xi_l), \Tr_{\sfM}\rbra*{\xi'_1}\otimes \cdots \otimes \Tr_{\sfM}\rbra*{\xi'_l}} \geq \frac{(\sqrt{c}-\sqrt{s})^2}{4(l-1)}.\]
\end{proposition}

Then, we proceed with the formal proof of \Cref{thm:IndivProdQSD-QSZKLhard}:
\begin{proof}[Proof of \Cref{thm:IndivProdQSD-QSZKLhard}]
    Let $\protocol{P}{V}$ be an $m(n)$-turn honest-verifier unitary quantum statistical zero-knowledge proof system for a promise problem $\calI \in\QSZKLstar[m,c,s]$, with completeness $c(n)$ and soundness $s(n)$. Without loss of generality, we assume that $m$ is even for all $x\in \calI$.\footnote{If $m$ is odd, we can add an initial turn to $\protocol{P}{V}$ in which the verifier sends the all-zero state to the prover.\label{footref:dummy-message}} Hence, we can denote the verifier's actions by $V_1,\cdots,V_{l+1}$ for $l=m/2$, and the verifier initiates the protocol. Let $\{\sigma_{x,i}\}_{x\in\calI,i\in[m+2]}$ represent the mixed states produced by the simulator $\{S_x\}_{x\in\calI}$, with the threshold function $\delta(n) \coloneqq 1/m(n)^2$. For any $x\in \calI$, we can define states $\xi'_0,\cdots,\xi'_l$ and $\xi_1,\cdots,\xi_{l+1}$ as illustrated in \Cref{fig:QSZKL-simulators}: 
    \begin{itemize}[itemsep=0.33em, topsep=0.33em, parsep=0.33em]
        \item Initial state before executing $\protocol{P}{V}$: $\xi'_0 \coloneqq \ketbra{\bar{0}}{\bar{0}}_{\sfM_0\sfW_0}$. 
        \item $(2j)$-th message for $j\in[l]$ in $\protocol{P}{V}$: $\xi'_j \coloneqq \sigma_{x,2j}$, where $\sigma_{x,2j}$ satisfies:
        \begin{equation}
            \label{eq:simulator-distance-bounds}
            \forall x\in \calI_\yes, \quad \td\rbra*{\sigma_{x,2j-1},\view{P}{V}(x,2j)} = \td\rbra*{\sigma_{x,2j-1},\rho_{\sfM_j\sfW_j}} \leq \delta(n).
        \end{equation}
        \item $(2j+1)$-st message for $j\in[l]$ in $\protocol{P}{V}$: $\xi_j \coloneqq V_j \xi'_{j-1} \smash{V_j^{\dagger}}$.
        \item State before the final measurement in $\protocol{P}{V}$: $\xi_{l+1} \coloneqq V_{l+1} \xi'_l \smash{V_{l+1}^{\dagger}}$ satisfies 
        \[\Tr\rbra*{\ketbra{1}{1}_{\ttZ} \xi_{l+1}}=c(n).\] 
    \end{itemize}

    Let $Q_1,\cdots,Q_k$ and $Q'_1,\cdots,Q'_k$ be polynomial-size unitary quantum circuits acting on $O(\log n)$ qubits which satisfy that $Q_j = S_{x,2j-1}$ and $Q'_j = S_{x,2j}$ for $j\in[l]$, and the output qubits are qubits in the verifier's private register $\sfW$. It is evident that $Q_j$ and $Q'_j$ prepare the states $\Tr_{\sfM}\rbra*{\xi_j}$ and $\Tr_{\sfM}\big(\xi'_j\big)$, respectively. We claim that the $l$-tuples $(Q_1,\cdots,Q_l)$ and $(Q'_1, \cdots Q'_l)$ form an instance of $\coIndivProdQSD[l(n),\alpha(n),\delta'(n)]$, satisfying the following conditions: 
    \begin{align}
        &\forall x \in \calI_\yes,& ~~ &\td\rbra*{\Tr_{\sfM}\rbra*{\xi_j}, \Tr_{\sfM}\rbra*{\xi'_j}} \leq 2\delta = \frac{(\sqrt{c}-\sqrt{s})^2}{4l^2} \coloneqq \delta'\text{ for } j \in [l]; \label{eq:QSZKLhard-yes-prob}\\
        &\forall x \in \calI_\no,& ~~ &\td\rbra*{\Tr_{\sfM}\rbra*{\xi_1} \otimes \cdots \otimes \Tr_{\sfM}\rbra*{\xi_l}, \Tr_{\sfM}\rbra*{\xi'_1} \otimes \cdots \otimes \Tr_{\sfM}\rbra*{\xi'_l}} \geq \frac{(\sqrt{c}-\sqrt{s})^2}{4(l-1)} \coloneqq \alpha. \label{eq:QSZKLhard-no-prob}
    \end{align}

    By substituting \Cref{eq:QSZKLhard-yes-prob} into \Cref{lemma:trace-distance-product-states}, it follows that: 
    \begin{equation}
        \label{eq:QSZKLhad-yes-prob-global}
        \begin{aligned}
             \td\rbra*{\Tr_{\sfM}\rbra*{\xi_1} \otimes \cdots \otimes \Tr_{\sfM}\rbra*{\xi_l}, \Tr_{\sfM}\rbra*{\xi'_1} \otimes \cdots \otimes \Tr_{\sfM}\rbra*{\xi'_l}} &\leq \sum_{j \in [l]} \td\rbra*{\Tr_{\sfM}\rbra*{\xi_j}, \Tr_{\sfM}\rbra*{\xi'_j}} \\
             &\leq \frac{(\sqrt{c}-\sqrt{s})^2}{4l}. 
        \end{aligned}
    \end{equation}
    Consequently, by comparing \Cref{eq:QSZKLhard-yes-prob,eq:QSZKLhard-no-prob,eq:QSZKLhad-yes-prob-global}, we can conclude the parameter requirement of $\coIndivProdQSD[l(n),\alpha(n),\delta'(n)]$, specifically that $\alpha(n) - \delta'(n) \cdot l(n) \geq 1/\poly(n)$.

    It remains to establish \Cref{eq:QSZKLhard-yes-prob,eq:QSZKLhard-no-prob}. The latter follows directly from \Cref{prop:QSZKLhard-soundness}. To prove the former, note that the prover's actions do not affect the verifier's private register for \textit{yes} instances, we thus derive the following for $j \in \{2,\cdots,l\}$: 
    \begin{align*}
        \td\rbra*{\Tr_{\sfM}\rbra*{\xi_j}, \Tr_{\sfM}\rbra*{\xi'_j}} 
        &\leq \td\rbra*{\xi_j, \xi'_j} \\
        &\leq \td\rbra*{\xi_j, \rho_{\ttM_j\ttW_j}} + \td\rbra*{\rho_{\ttM_j\ttW_j},\rho_{\ttM'_j\ttW_j}} + \td\rbra*{\rho_{\ttM'_j\ttW_j}, \xi'_j} \\
        &= \td\rbra*{\xi'_{j-1}, \rho_{\ttM'_{j-1}\ttW_{j-1}}} + \td\rbra*{\rho_{\ttM_j\ttW_j},\rho_{\ttM'_j\ttW_j}} + \td\rbra*{\rho_{\ttM'_j\ttW_j}, \xi'_j}\\
        &\leq \delta(n) + 0 + \delta(n)\\
        &= 2\delta(n).
    \end{align*}
    Here, the first line follows from the data-process inequality (\Cref{lemma:trace-distance-data-processing}), the second line is due to the triangle inequality, the third line owes to the unitary invariance (\Cref{lemma:trace-distance-unitary-invariance}) and the fact that $\rho_{\ttM_j\ttW_j} = V_j \rho_{\ttM'_{j-1}\ttW_{j-1}} V_j^{\dagger}$, and the fourth line is because of \Cref{eq:simulator-distance-bounds}. 
    We complete the proof by noting that similar reasoning applies to the case of $j=1$, using $\td\rbra*{\xi_1, \rho_{\ttM_1\ttW_1}} = 0$ instead of at most $\delta(n)$.
\end{proof}

\subsection{\QSZKULHV{} is in \BQL{}}
\label{subsec:QSZKL-in-BQL}

We will establish the hard direction in the equivalence of \QSZKULHV{} and \BQL{}. 
The key lemma underlying the proof involves a logspace (many-to-one) reduction  \IndivProdQSD{} to an ``existential'' version of \GapQSDlog{}, where \GapQSDlog{} is a \BQL{}-complete problem (see \Cref{subsec:space-bounded-state-testing}). This reduction leads to a \BQL{} containment for \IndivProdQSD{}: 
\begin{lemma}[\IndivProdQSD{} is in \BQL{}]
    \label{lemma:IndivProdQSD-in-BQL}
    Let $k(n)$, $\alpha(n)$ and $\delta(n)$ be logspace computable functions such that $1 \leq k(n) \leq \poly(n)$, $0 \leq \alpha(n),\delta(n) \leq 1$, and $\alpha(n) - \delta(n) \cdot k(n) \geq 1/\poly(n)$. Then, it holds that
    \[ \IndivProdQSD[k(n),\alpha(n),\delta(n)] \in \BQL. \]
\end{lemma}

As \coIndivProdQSD{} is \QSZKULHV{}-hard (\Cref{thm:IndivProdQSD-QSZKLhard}), and given that \BQUL{} is closed under complement~\cite[Corollary 4.8]{Wat99} and the equivalence $\BQL=\BQUL$~\cite{FR21}, we can directly conclude the following corollary:

\begin{corollary}
    $\QSZKULHV{} \subseteq \BQL{}.$
\end{corollary}

We now proceed with the formal proof of the key lemma: 

\begin{proof}[Proof of \Cref{lemma:IndivProdQSD-in-BQL}]
    We first establish a logspace (many-to-one) reduction from \IndivProdQSD{} to an ``existential'' version of \GapQSDlog{}. 
    Let $(Q_1,\cdots,Q_k)$ and $(Q'_1,\cdots,Q'_k)$ be an instance of $\IndivProdQSD[k,\alpha,\delta]$. For each $j\in[k]$, let $\sigma_j$ and $\sigma'_j$ denote the states obtained by running $Q_j$ and $Q'_j$ on the all-zero state $\ket{\bar{0}}$, respectively, and tracing out the non-output qubits. We now need to decide which of the following cases in \Cref{eq:IndivProdQSD-far,eq:IndivProdQSD-close} holds: 
    \begin{align}
        \td\rbra*{\sigma_1\otimes \cdots \otimes \sigma_k, \sigma'_1\otimes \cdots \otimes \sigma'_k} \geq \alpha(n). \label{eq:IndivProdQSD-far}\\
        \forall j \in [k], \quad \td\big(\sigma_j, \sigma'_j\big) \leq \delta(n) \label{eq:IndivProdQSD-close}.
    \end{align}

    By combining \Cref{lemma:trace-distance-product-states} with \Cref{eq:IndivProdQSD-far}, we obtain: 
    \begin{equation}
        \label{eq:IndivProdQSD-farE}
        \sum_{j\in[k]} \td\big(\sigma_j, \sigma'_j\big) \geq 
        \td\rbra*{\sigma_1\otimes \cdots \otimes \sigma_k, \sigma'_1\otimes \cdots \otimes \sigma'_k} \geq \alpha(n).
    \end{equation}

    Applying an averaging argument to \Cref{eq:IndivProdQSD-farE}, we can conclude that
    \begin{equation}
        \label{eq:IndivProdQSD-far-invidual}
        \exists j\in[k], \quad \td\big(\sigma_j, \sigma'_j\big) \geq \alpha/k. 
    \end{equation}
    
    Clearly, a violation of \Cref{eq:IndivProdQSD-far-invidual} implies a violation of \Cref{eq:IndivProdQSD-far}, without contradicting \Cref {eq:IndivProdQSD-close}. 
    For each $j\in[k]$, the pair of circuits $Q_j$ and $Q'_j$ forms an instance of \GapQSDlog{}. The resulting promise problem is thus an ``existential'' version of \GapQSDlog{}, where \textit{yes} instances satisfy \Cref{eq:IndivProdQSD-far-invidual} and \textit{no} instances satisfy \Cref{eq:IndivProdQSD-close}. 

    \vspace{1em}
    Next, we proceed by demonstrating the \BQL{} containment. Given the equivalence of $\BQL$ and  $\QMAL$~\cite{FKLMN16,FR21}, it remains to establish a \QMAL{} containment for this ``existential'' version of \GapQSDlog{}. The verification protocol is outlined in \Cref{protocol:IndivProdQSD-in-QMAL}. 

\begin{algorithm}[ht!]
    \SetAlgorithmName{Protocol}
    \SetEndCharOfAlgoLine{.}
    \SetAlgoVlined
    \setlength{\parskip}{5pt}
    \SetKwInOut{Parameter}{Parameters}

    \textbf{1.} The verifier receives an index $j \in [k]$ from the prover.

    \textbf{2.} The verifier executes the quantum logspace algorithm $\calA$ for $\GapQSDlog[\alpha/k,\delta]$ underlying in~\Cref{thm:GapQSDlog-in-BQL}, using the pair of circuits $Q_j$ and $Q'_j$ as the \GapQSDlog{} instance. The verifier accepts (or rejects) if $\calA$ accepts (or rejects).
    \BlankLine
	\caption{A \QMAL{} proof system for \IndivProdQSD{}.}
	\label[algorithm]{protocol:IndivProdQSD-in-QMAL}
\end{algorithm}

    To complete the proof, we establish the correctness of \Cref{protocol:IndivProdQSD-in-QMAL}. Since the algorithm $\calA$ is a \BQL{} containment for $\GapQSDlog[\alpha/k,\delta]$ (\Cref{thm:GapQSDlog-in-BQL}), we conclude the following: 
    \begin{itemize}[topsep=0.33em, itemsep=0.33em, parsep=0.33em]
        \item For \textit{yes} instances, \Cref{eq:IndivProdQSD-far-invidual} ensures that there exists an  $j \in [k]$ (the witness) such that $\td\big(\sigma_j, \sigma'_j\big) \geq \alpha/k$. Consequently, $\calA$ accepts with probability at least $2/3$. 
        \item For \textit{no} instances, \Cref{eq:IndivProdQSD-close} yields that for all $j \in [k]$, $\td\big(\sigma_j, \sigma'_j\big) \leq \delta$. This statement implies that $\calA$ accepts with probability at most $1/3$. \qedhere
    \end{itemize}
\end{proof}


\section*{Acknowledgments}
\noindent
The authors appreciate Dorian Rudolph for pointing out a gap in the earlier argument that \QIPL{} is contained in \AM{}, specifically noting that \Cref{lemma:QIPL-second-SDP-formulation} does not directly apply to \QIPL{}, as it does not verify the consistency of the prover strategies across different measurement outcome branches. 
YL is grateful to Uma Girish for a helpful discussion that inspired Question \ref{probitem:constTurnQIPL}. 
The authors also thank the anonymous reviewers for their helpful comments, which made the introduction more accessible to non-expert readers. 

This work was partially supported by MEXT Q-LEAP grant No.~\mbox{JPMXS0120319794}. 
FLG was also supported by JSPS KAKENHI grants Nos.~\mbox{JP20H05966}, \mbox{20H00579}, \mbox{24H00071}, and by MEXT JST CREST grant No.~\mbox{JPMJCR24I4}. 
YL was further supported by JST SPRING grant No.~\mbox{JPMJSP2125} and acknowledges the ``THERS Make New Standards Program for the Next Generation Researchers.''
HN was additionally supported by JSPS KAKENHI grants Nos.~\mbox{JP19H04066}, \mbox{JP20H05966}, \mbox{JP21H04879}, and \mbox{JP22H00522}.
QW was supported in part by the Engineering and Physical Sciences Research Council under Grant No.~\mbox{EP/X026167/1}. 


\bibliographystyle{alphaurlQ}
\bibliography{space-bounded-quantum-interactive-proofs}

\end{document}